\documentclass{article}
\usepackage[utf8]{inputenc}
\usepackage{macros}
\usepackage[backend=bibtex, style=numeric,sorting=none]{biblatex}
\usepackage{tikz}
\usetikzlibrary{arrows, decorations.markings, decorations.pathmorphing, decorations.pathreplacing, shapes.arrows, patterns, calc}
\usepackage{url}

\pgfdeclaredecoration{complete sines}{initial} 
{ \state{initial}[ width=+0pt, next state=sine, persistent precomputation={\pgfmathsetmacro\matchinglength{ \pgfdecoratedinputsegmentlength / int(\pgfdecoratedinputsegmentlength/\pgfdecorationsegmentlength)} \setlength{\pgfdecorationsegmentlength}{\matchinglength pt} }] {} \state{sine}[width=\pgfdecorationsegmentlength]{ \pgfpathsine{\pgfpoint{0.25\pgfdecorationsegmentlength}{0.5\pgfdecorationsegmentamplitude}} \pgfpathcosine{\pgfpoint{0.25\pgfdecorationsegmentlength}{-0.5\pgfdecorationsegmentamplitude}} \pgfpathsine{\pgfpoint{0.25\pgfdecorationsegmentlength}{-0.5\pgfdecorationsegmentamplitude}} \pgfpathcosine{\pgfpoint{0.25\pgfdecorationsegmentlength}{0.5\pgfdecorationsegmentamplitude}} } \state{final}{} } 

\tikzset{
  ->-/.style={decoration={markings, mark=at position 0.5 with {\arrow{to}}},
              postaction={decorate}},
}

\tikzset{
  -<-/.style={decoration={markings, mark=at position 0.5 with {\arrow{to reversed}}},
              postaction={decorate}},
}

\tikzset{
  dbl->-/.style={
double, 
double equal sign distance,
shorten >= 1pt,
shorten <= 1pt,
 decoration={markings, mark=at position 0.5 with {\arrow{implies}}},
              postaction={decorate}},
}

\tikzset{
  dbl-<-/.style={
double, 
double equal sign distance,
shorten >= 1pt,
shorten <= 1pt,
 decoration={markings, mark=at position 0.5 with {\arrowreversed{implies}}},
              postaction={decorate}},
}

\newcommand{\half}{\tfrac{1}{2}}
\newcommand{\Lap}{\triangle}
\title{Celestial holography meets twisted holography: 4d amplitudes from chiral correlators}
\author{Kevin Costello$^1$, Natalie M. Paquette$^2$}
\date{%
    $^1$Perimeter Institute for Theoretical Physics\\%
    $^2$Department of Physics, University of Washington, Seattle\\[2ex]%
}

\begin{document}

\maketitle

\begin{abstract}
  We propose a new program for computing a certain integrand of scattering amplitudes of four-dimensional gauge theories which we call the \textit{form factor integrand}, starting from 6d holomorphic theories on twistor space. We show that the form factor integrands can be expressed as sums of products of 1.) correlators of a 2d chiral algebra, related to the algebra of asymptotic symmetries uncovered recently in the celestial holography program, and 2.) OPE coefficients of a 4d non-unitary CFT. We prove that conformal blocks of the chiral algebras are in one-to-one correspondence with local operators in 4d. We use this bijection to recover the Parke-Taylor formula, the CSW formula, and certain one-loop scattering amplitudes.  Along the way, we explain and derive various aspects of celestial holography, incorporating techniques from the twisted holography program such as Koszul duality. This perspective allows us to easily and efficiently recover the infinite-dimensional chiral algebras of asymptotic symmetries recently extracted from scattering amplitudes of massless gluons and gravitons in the celestial basis. We also compute some simple one-loop corrections to the chiral algebras and derive the three-dimensional bulk theories for which these 2d algebras furnish an algebra of boundary local operators.
\end{abstract}

\section{Introduction}
A great deal of progress has been made in recent years on the structure of scattering amplitudes for supersymmetric gauge theory on flat space.  In one direction, inspired by twistor string theory \cite{Witten:2003nn},  exact loop-level results have been obtained for the integrand of $\mc{N}=4$ Yang-Mills scattering amplitudes. 

In a different direction, there has been a surge of recent work on the asymptotic symmetries of scattering amplitudes in flat space (see e.g. \cite{Strominger:2017zoo}). Perhaps the greatest success in this direction has been the realization \cite{Guevara:2021abz, Strominger:2021lvk, Himwich:2021dau} that there are beautiful chiral algebras and infinite-dimensional Lie algebras emerging from the study of conformally-soft gluons and gravitons.    

These developments are not completely unrelated, although the precise connection has been somewhat mysterious.  A starting point for Witten's twistor-string theory work was Nair's observation \cite{Nair:1988bq}, relating tree-level amplitudes for $\mc{N}=4$ gauge theory to correlators of a super Kac-Moody algebra, which appears to be related to the chiral algebras of celestial holography.  However, Nair's algebra has a non-zero Kac-Moody level, unlike the Kac-Moody algebras found in celestial holography.  Further, Nair's identity only holds after discarding multi-trace terms in the Kac-Moody correlators.

In this work, we provide a general method for understanding form factors of certain non-supersymmetric gauge theories as correlators of chiral algebras of the type studied in \cite{Guevara:2021abz}; such form factors are, in turn, related to certain scattering amplitudes in QCD. Our main result is a formula for a certain integrand, which we dub the \textit{form factor integrand}, that computes scattering amplitudes in the presence of a local operator insertion (i.e. a form factor), as a sum of products of two quantities:
\begin{enumerate} 
	\item Correlators of a chiral algebra closely related to that appearing from the study of soft gluons \cite{Guevara:2021abz}, and in particular containing a level $0$ Kac-Moody algebra;
	\item  OPE coefficients of a four-dimensional non-unitary CFT.  
\end{enumerate}
We explicitly check our formulae against known results for certain tree-level and one-loop amplitudes.  

Both quantities in our formula are very tightly constrained by associativity or crossing symmetry, in dimensions $2$ and $4$ respectively.  This suggests that one can use this method to bootstrap the integrand for scattering amplitudes at loop level.

\subsection{The $4d$ CFT}
The starting point for our analysis is a class of $4d$ CFTs considered in  \cite{Costello:2021bah}.  These are theories that come from local holomorphic field theories on twistor space. At the classical level, any self-dual gauge theory can be described in this way.  For non-supersymmetric theories, this can be spoiled at the quantum level by anomalies \cite{Costello:2021bah}.  Fortunately, in many cases, the anomaly can be cancelled by an unusual Green-Schwarz mechanism which requires the introduction of an axion field.   

This cancellation works with gauge group $SU(2)$, $SU(3)$, $SO(8)$ or an exceptional group. In these cases the Lagrangian is
\begin{equation}
	\int \op{tr}(B F(A)_-) - \half \int (\Delta \rho)^2 - \frac{\sqrt{10} \sh^\vee  }{8 \pi \sqrt{3} \sqrt{\dim g + 2}  }  \int \rho \op{tr}(F(A)^2). \label{eqn:sample_lagrangian}
\end{equation}
where $\sh^\vee$ is the dual Coxeter number.  (The constants come from the coefficients of a trace identity, and the coefficients of a Feynman diagram on twistor space).  

In this expression, $\rho$ is a scalar field and $B$ is an adjoint-valued ASD\footnote{Here we use the opposite conventions to those in \cite{Costello:2021bah} in order to match the ``mostly $+$'' conventions of the scattering amplitudes literature.}  $2$-form.  (One can also take the gauge group to be $SU(2)$, $SO(8)$ or an exceptional group. If we introduce matter, we could take  $SU(N_c)$ with $N_f =  N_c$. In each case the axion coupling needs to be tuned to cancel the anomaly).     

The fact that the theory arises from an anomaly free theory on twistor space implies that all correlation functions, and OPE coefficients, are rational functions. 

We are interested in deforming this theory by $g_{YM}^2 \op{tr} (B^2)$.  As is well known, once we add $\op{tr} (B^2)$ to the Lagrangian we get a theory that is perturbatively equivalent to ordinary Yang-Mills theory, plus an axion field.  Thus, one can compute quantities in ordinary Yang-Mills theory at order $2n$ in the coupling constant $g_{YM}$ by placing the operators $\op{tr}(B^2)$ at points $x_1,\dots,x_n \in \R^4$ and then integrating over their position.

The quantity of interest in this paper is what we shall refer to as the \emph{form factor integrand}: the scattering amplitudes of the gauge theory in the presence of the operator $\op{tr}(B^2)$ at points $x_1,\dots, x_n$. The name is chosen to emphasize that amplitude is computed in the presence of an operator \footnote{There is a large literature on the computation of form factors, especially in $\mathcal{N}=4$ SYM; see e.g. \cite{Bern:2005iz, Brandhuber:2012vm} for some loop-level results and \cite{Yang:2019vag} for a review with further references.}.

We should emphasize that the form-factor integrand is \emph{not}  the same as what other authors call the integrand,  although it is related.  Our form-factor integrand is closely related to natural quantities appearing in twistor-string theory \cite{Witten:2003nn}, where amplitudes are expressed as integrals over spaces of curves in twistor space.  The connection is given by noting that each point $x_i \in \R^4$ gives rise to a curve $\CP^1_{x_i}$ in twistor space. 

As mentioned, we will present a formula for the form factor integrand which is a sum of products of the OPE coefficients of this CFT, together with the correlation functions of a chiral algebra that we will now discuss.

\subsection{The chiral algebra}
The chiral algebra we use is very closely related to that studied in the celestial holography literature \cite{Guevara:2021abz}. Here  we will write down the generators of the chiral algebra and their OPEs explicitly. They are derived in the bulk of the paper by starting with the twistor space description of the theory and using the method of Koszul duality  \cite{Costello:2020jbh, PW}. We work in Euclidean signature here, although since our integrand is an entire analytic function, we can readily move to other signatures.  We write $\op{Spin}(4)$ as $SU(2)_+ \times SU(2)_-$.  The chiral algebra lives on a $\CP^1$ with coordinate $z$, which is rotated by $SU(2)_-$. 

The chiral algebra has four towers of states, each living in an infinite sum of finite-dimensional representation of $SU(2)_+$.  A state in the chiral algebra has a spin, in the usual sense of chiral algebras; a weight under the Cartan of $SU(2)_+$; and also a lives in a $SU(2)_+$ representation of some heighest weight. The generators, and the 4d fields to which they couple (as described in more detail in the main text and below), are listed in table \ref{table:chiralalgebra}. 

\begin{table}
\begin{tabular}{c |  c |  c |  c | c |c }
	Generator  & Spin & Weight & $SU(2)_+$ representation & Field & Dimension \\
	$J[m,n]$, $m,n \ge 0$ & $1-(m+n)/2$   & $(m-n)/2$ & $(m+n)/2$ & $A$ & $-m-n$  \\
	$\til{J}[m,n]$, $m,n \ge 0$ & $-1-(m+n)/2$   & $(m-n)/2$ & $(m+n)/2$ & $B$ & $-m-n-2$ \\
	$E[m,n]$, $m+n > 0$ &  $-(m+n)/2$   & $(m-n)/2$ & $(m+n)/2$ & $\rho$ & $-m-n$ \\
	$F[m,n]$, $m,n \ge 0$ & $ -(m+n)/2  $   & $(m-n)/2$ & $(m+n)/2$ & $\rho$ & $-m-n-2$  
\end{tabular}
	\caption{The generators of our 2d chiral algebra and their quantum numbers.  Dimension refers to the charge under scaling of $\R^4$. \label{table:chiralalgebra}}
\end{table}

At tree level, the OPEs for the $J$, $\til{J}$ currents are
\begin{equation} 
	\begin{split}
		J^a[r,s](0) J^b[t,u](z) &\sim \frac{1}{z} f^{ab}_c J^c[r+t, s+u] (0) \\
		J^a[r,s](0) \til{J}^b[t,u](z) &\sim \frac{1}{z} f^{ab}_c \til{J}^c[r+t, s+u] (0) 
	\end{split} \label{eqn:OPEs}
\end{equation}
This OPE is subject to loop corrections. The method of Koszul duality gives a well-defined prescription for computing these, but we have not yet fully analyzed all loop corrections.  At one loop we do know there is an additional term in the OPE
\begin{equation} 
	\begin{split} 
		J_a [1,0] (0)  J_b [0,1] (z) & \sim \frac{C}{z}  K^{fe} f_{ae}^c f_{bf}^d ( \til{J}_c[0,0] J_d[0,0] +  J_c[0,0] \til{J}_d[0,0]   )      \\  
		J_a [1,0] (0)  \til{J}_b [0,1] (z)  &\sim \frac{C}{z}  K^{fe} f_{ae}^c f_{bf}^d \til{J}_c[0,0] \til{J}_d[0,0]    
	\end{split}
\end{equation}
Here we have only written the relation in the case the indices $a,b$ are such that $[\t_a,\t_b] = 0$; and $C$ is a constant we have not determined.   

To write the OPEs involving the $E,F$ towers, it is convenient to introduce a constant $\lambda_{\g}$ so that 
\begin{equation} 
	\op{Tr}(X^4) = \lambda_{\g}^2 \op{tr}(X^2)^2 
\end{equation}
where on the right hand side we take trace in the fundamental, and on  the left in the adjoint. This trace identity only holds for the gauge groups we consider.  Explicitly \cite{} we have
\begin{equation} 
	\lambda_{\g} = \frac{\sqrt{10} \sh^\vee}{\sqrt{\dim \g +2} } \end{equation}
where $\sh^\vee$ is the dual Coxeter number.  Then we set
\begin{equation} 
	\what{\lambda}_{\g} = \frac{\lambda_{\g}} { (2 \pi \i)^{3/2} \sqrt{12} } .
\end{equation}
The constant arises from the coupling constant on twistor space required to cancel the anomaly \cite{Costello:2021bah}.  Then, we have the additional OPEs
\begin{equation} 
	 \begin{split}	
		 J^a[r,s](0) E[t,u](z) &\sim     \frac{ \what{\lambda}_{\g}   }{z} \frac{(ts - ur)}{t + u} \til{J}^a [t+r - 1, s + u -1](0)  \\
		J^a[r,s](0) F[t,u](z) &\sim  -  \frac{  \what{\lambda}_{\g}  }{z} \partial_z \til{J}^a[r+t, s + u](0)  - \frac{   \what{\lambda}_{\g}   }{z^2} (1 + \frac{r + s}{t+u+2}) \til{J}^a[r+t, s+u](0) \\
		J^a[r,s](0) J^b[t,u](z) \sim& \frac{   \what{\lambda}_{\g}  }{z} K^{ab} (ru-st) F[r+t-1,s+u-1] (0)\\
		&- \frac{    \what{\lambda}_{\g}  }{z} K^{ab} (t+u)  \partial_z E[r+t,s+u](0) - \frac{1}{z^2} K^{ab} (r+s+t+u) E[r+t,s+u](0).  
	 \end{split} \label{eqn:axion_opes}
\end{equation}
(It can also be natural to include the coefficient of the coupling between the axion and the gauge field in these expressions explicitly, but this can be removed by a redefinition of the generators $E,F,\til{J}$.)

\subsection{Formula for form factors}
We can now put the 4d and 2d pieces together to obtain the advertised expression for the form factor integrand, which are related to certain integrands of scattering amplitudes by the previous discussion. To explain our formula for form factors, we need to first state some properties of the relation between the chiral algebra and the four-dimensional CFT. 
\begin{enumerate}
	\item The generators of the vertex algebra, as listed above, are in bijection with single-particle conformal primary states of the four-dimensional theory in the sense of \cite{Pasterski:2017kqt}, of mostly \emph{negative} conformal dimension (the conformal dimension is the spin of the field in table \ref{table:chiralalgebra}).  The generators $J^a[r,s]$ correspond to gluons of positive helicity, and $\til{J}^a[r,s]$ to gluons of negative helicity. 
    \item Conformal blocks of our vertex algebra are in bijection with local operators in the $4d$ theory.
\end{enumerate}
For our purposes, a \emph{conformal block} is a way of defining correlation functions of the vertex algebra compatible with the OPEs.  

Thus, given any local operator $\mc{O}$ of the $4d$ theory, we can define the correlation functions of the vertex algebra by using the conformal block corresponding to $\mc{O}$. Such correlation functions will be denoted by
\begin{equation}
		\ip{\mc{O} \mid V_1(z_1) \dots V_n(z_n) } 
\end{equation}
where $V_i$ are elements of the vacuum module of the vertex algebra placed at points $z_i$. (We lose no generality by taking the $V_i$ to be single-particle conformal primary states such as $J[r,s]$, $\til{J}[r,s]$).

Our first result is:
\begin{proposition}
The following two quantities are equal:
\begin{enumerate}
	\item Scattering amplitudes of the $4d$ theory in the presence of our chosen local operator at fixed position (these quantities are known as \emph{form factors}). 
    \item Correlation functions of the chiral algebra defined using the corresponding conformal block. 
\end{enumerate}
\end{proposition}

We have stated that conformal primary generators of the chiral algebra are the same as single-particle states of the $4d$ theory in the conformal basis.  To translate to standard formulae for scattering amplitudes, we should express states in the momentum basis in terms of the chiral algebra.  A null momentum $p_{\alpha \dot{\alpha}}$ can be expressed as a pair of spinors, $\lambda^{\alpha}$ and $\mu^{\dot{\alpha}}$.    The momentum eigenstates of positive and negative helicity corresponding to the pair of spinors $\lambda,\mu$ are obtained by taking $\lambda = (1,z)$, and looking at the generating function
\begin{equation}
	\begin{split} 
		J(\mu,z) &= \sum \frac{(\mu^{\dot{1}})^r (\mu^{\dot{2}})^s}{r! s!} J[r,s]. \\
		\til{J}(\mu,z) &= \sum \frac{(\mu^{\dot{1}})^r (\mu^{\dot{2}})^s}{r! s!} \til{J}[r,s].  
	\end{split}\label{eqn:generating_function}	
\end{equation}
The expansion in powers of $\mu$ is an expansion of a momentum eigenstate in soft modes, where the energy has been absorbed into the scale of $\mu$. 

Correlators of the chiral algebra will then be expressed in terms of\begin{equation} 
	\begin{split}
		\ip{ij}&= z_i - z_j\\
		[ij] &= \eps_{\dot{\alpha} \dot{\beta}} \mu_i^{\dot{\alpha}} \mu_j^{\dot{\beta}}.
	\end{split}
\end{equation}
Our theorem relating correlators and form factors is best implemented using these generating functions.   Suppose we have a Lorentz invariant local operator $\mc{O}$.  Lorentz invariance tells us that 
\begin{equation} 
	\ip{ \mc{O} \middle| J_{a_1}(\mu_1,z_1) \dots \til{J}_{a_n}(\mu_n, z_n) }  
\end{equation}
is expressed as a (finite) sum involving only $[ij]$, $\ip{ij}$, and contractions of the colour indices $a_i$.   These expressions can then be identified with standard expressions in the amplitudes literature.

In this work, we will focus on the scattering amplitudes in the presence of the operator $\op{tr}(B^2)$ placed at points $x_1, \dots, x_n$. This is the quantity we called the form factor integrand. 

There is an operator product expansion
\begin{equation} 
	\op{tr}(B^2)(0) \op{tr}(B^2)(x_1) \dots \op{tr}(B^2)(x_{n-1)} \sim \sum F^i(x_1,\dots, x_{n-1}) \mc{O}_i(0) 
\end{equation}
where $\mc{O}_i$ runs over a basis of local operators in the $4d$ CFT, and if $\mc{O}_i$ has dimension $d$ then $F$ is a rational function of the $x_i$ of degree $d -2n$. It is important to note that all CFTs that come from local theories on twistor space do not have anomalous dimensions of local operators, so that $d$ is an integer.  

Our formula is:
\begin{theorem}\label{mainthm}
The form factor integrand for scattering amplitudes of $n$ positive helicity and $m$ negative helicity conformal primary states has an expansion 
	\begin{equation} 
		\begin{split} 
			\sum F^i(x_1,\dots, x_{n-1})           \Big\langle \mc{O}_i(0) \mathrel{\Big|} J^{a_1} (\mu_1,z_1) & \dots J^{a_n}(\mu_n,z_n)  \\
				&  \til{J}^{b_1}(\mu'_1,z'_1) \dots \til{J}^{b_m}(\mu'_m, z'_m) \Big \rangle. 
		\end{split}
	\end{equation}	
\end{theorem}
We note that on the right hand side of the formula we find a sum of products of the OPE coefficients $F^i$ and of correlation functions of the chiral algebra.

\subsection{The Parke-Taylor formula}
In the body of the paper we will prove this result carefully. Here, we will give some examples, starting with the case $n = 1$. Then, we are studying the scattering amplitudes of self-dual gauge theory in the presence of the operator $\op{tr}(B^2)$ at the origin.  At tree level, these are the same as MHV amplitudes, given by the Parke-Taylor formula.  

We will check that our formula at tree level reproduces the Parke-Taylor formula. Since we work at tree level we do not need to concern ourselves with the axion field.

First, we find the conformal block corresponding to the operator $\op{tr}(B^2)$.  By considering how conformal blocks transform under the Lorentz group, we find (as we will explain in more detail later)  
\begin{equation} 
	\ip{ \op{tr}(B^2) \middle| \til{J}^a [0,0] (z_1) \til{J}^b[0,0] (z_2) } = K^{ab} (z_1 - z_2)^2.  
\end{equation}
In this conformal block, insertions of any operator $J[i,j]$ or $\til{J}[i,j]$ with $i+j > 0$ give zero, as do insertions of three or more $\til{J}$.  The non-zero correlation functions are those involving two $\til{J}[0,0]$'s and $n$ $J[0,0]$, and they are completely determined by the OPEs \eqref{eqn:OPEs}.  

In this calculation, we will not have any dependence on the spinors $\mu_i$ in the generating functions $J(\mu_i,z_i)$. Thus, we will write $J$, $\til{J}$ for $J[0,0]$ and $\til{J}[0,0]$. 

The three-point correlation function is 
\begin{equation}
	\begin{split} 
		\ip{ \op{tr}(B^2) \middle| \til{J}^a  (z_1) \til{J}^b (z_2) J^c(z_3)  }  =& f^{cb}_d \frac{1}{z_{23} }  \ip{\op{tr}(B^2) \middle| \til{J}^a (z_1) \til{J}^d (z_2) } \\
		&+ f^{ca}_d \frac{1}{z_{13} }  \ip{\op{tr}(B^2) \middle| \til{J}^d (z_1) \til{J}^b (z_2) } \\    
		=& \frac{z_{12}^3}{z_{13} z_{23} } f^{abc} 
	\end{split} 
\end{equation}
which matches the Parke-Taylor formula.  Proceeding by induction, it is not difficult to show that the colour-ordered\footnote{Here we mean that we consider the term where the colour indices are contracted by $\op{tr}(\t_{a_1} \dots \t_{a_n})$. }  correlation function in the chiral algebra is
\begin{equation}
	\begin{split} 
		\ip{ \op{tr}(B^2) \middle|J^{a_1}(z_1) \dots  \til{J}^{a_i}  (z_i)  \dots \til{J}^{a_j}  (z_j) \dots J^{a_n}(z_n)  }  &=   \frac{ z_{ij}^4 } {z_{12} z_{23} \dots z_{n1} }
	\end{split}\label{eqn:PT}	
\end{equation}
again matching the Parke-Taylor amplitude (without the momentum conserving delta function, as we discuss later). 

It is important to note that there are no multi-trace terms in our correlator, just as there are no multi-trace terms in the tree-level amplitude. All terms in our correlator are permutations of the Parke-Taylor amplitude.  This tells us that our amplitude matches exactly with the tree-level MHV gauge theory amplitude. This would not be the case, for instance, if we had a non-zero Kac-Moody level for the $J$ currents, which would lead to multi-trace terms. 

\subsection{Cachazo-Svrcek-Witten formula}
Next, let us consider what happens when we have two copies of $\op{tr}(B^2)$, still working at tree level.  One of the terms in the tree-level OPE is
\begin{equation} 
	\op{tr}(B^2)(0) \op{tr}(B^2)(x) \sim \frac{1}{\norm{x}^2} B^a_{\alpha_1 \beta_1} B^b_{\alpha_2 \beta_2} B^c_{\alpha_3 \beta_3} f_{abc} \eps^{\beta_1 \alpha_2} \eps^{\beta_2 \alpha_3} \eps^{\beta_3 \alpha_1}. 
\end{equation}
We write $\op{tr}(B^3)$ as short hand for the operator on the right hand side, with the understanding that the spinor indices of $B_{\alpha \beta}$ are contracted in the unique Lorentz invariant way. The operator on the right hand side corresponds to the conformal block characterized by
\begin{equation} 
	\ip{\op{tr}(B^3) \middle| \til{J}^a (z_1) \til{J}^b(z_2) \til{J}^c(z_3) } = f^{abc} z_{12} z_{13} z_{23}. \label{eqn:B3}  
\end{equation}
The non-zero correlators are those with three $\til{J}$ insertions and $n$ $J$ insertions. They are determined from the correlator \eqref{eqn:B3} by the poles in the OPEs. 

We find that these correlators reproduce an un-integrated veresion of the Cachazo-Svrcek-Witten \cite{Cachazo:2004kj} formula for NMHV amplitudes.  In the CSW formula, one builds NMHV amplitudes by treating the MHV amplitudes as a vertex in a Feynman diagram, and then connecting these vertices by a propagator. 

Our prescription with the OPE has a similar description.  For example, we have the following formula.   If $V_i(z_i)$ denote the $n$ chiral algebra insertions, states, $3$ of which are $\til{J}$ and $n-3$ are $J$, we have
\begin{multline} 
	\ip{\op{tr}(B^3) \middle| V_1(z_1) \dots V_n(z_n)   } \\ = -\frac{1}{6} \sum \ip{\op{tr}(B^2) \middle| V_{i_1} (z_{i_1}) \dots V_{i_k}(z_{i_k}) \til{J}^a (z) }    \ip{\op{tr}(B^2) \middle|J_a(z)  V_{j_1} (z_{j_1}) \dots V_{j_{n-k}}(z_{j_{n-k} }) }   \label{eqn:NMHV}
	\end{multline}
where $z$ is arbitrary.  The sum on the right hand side is over all ways of distributing the chiral algebra insertions among the correlators.    

Clearly this formula is reminiscent of the CSW formula, as it expresses an NMHV correlator by gluing together MHV correlators.  It may seem at first sight that the CSW propagator (which is a propagator for a scalar field) is missing.  To see this propagator, we should recall that  $\op{tr}(B^3)$ appears as  the coefficient of $\norm{x}^{-2}$ in the OPE of two copies of $\op{tr}(B^2)$.  Correlators with respect to the conformal block $\op{tr}(B^3)$, when multiplied by $\norm{x}^{-2}$, thus contribute to the NMHV integrand.   Since $\norm{x}^{-2}$ is the propagator of a scalar field, equation \eqref{eqn:NMHV} is a close match with the CSW prescription.

Equation \eqref{eqn:NMHV} is proved in the bulk of the paper by an inductive method. The initial case is when $n = 3$, and is the identity
\begin{align*}
	-3	\ip{\op{tr}(B^3) \middle| \til{J}^{a_1}(z_1) \til{J}^{a_2} (z_2) \til{J}^{a_3}(z_3)   } =&  \ip{\op{tr}(B^2) \middle| \til{J}^{a_1}(z_1)  \til{J}^b (z)   }    \ip{\op{tr}(B^2) \middle|J_b(z) \til{J}^{a_2}(z_2) \til{J}^{a_3}(z_3)   }  \\
		+&  \ip{\op{tr}(B^2) \middle| \til{J}^{a_2}(z_2)  \til{J}^b (z)   }    \ip{\op{tr}(B^2) \middle|J_b(z) \til{J}^{a_3}(z_3) \til{J}^{a_1}(z_1)   }  \\
		+&  \ip{\op{tr}(B^2) \middle| \til{J}^{a_3}(z_3)  \til{J}^b (z)   }    \ip{\op{tr}(B^2) \middle|J_b(z) \til{J}^{a_1}(z_1) \til{J}^{a_2}(z_2)   }. 
\end{align*}
which is entirely elementary using the definitions for the OPEs and correlators given above.

We expect that at tree level, our formula for amplitudes is equivalent to an integrand version of the CSW prescription.  Our formula works equally well at loop level, as long as one understands loop corrections in both the chiral algebra and the OPEs in the $4d$ CFT.

\subsection{One loop amplitudes}
Our $4d$ CFT by itself does not have any non-trivial amplitudes; this is true of any local field theory on twistor space.  However, self-dual gauge theory does have non-trivial one-loop amplitudes. The simplest of these is the one-loop $4$-point amplitude with all particles of positive helicity:
\begin{equation} 
	\ip{1 2 3 4 } = \frac{ [12][34] }{\ip{12} \ip{34} } \op{tr}(\t^{a_1} \t^{a_2} \t^{a_3} \t^{a_4} ) + \text{ permutations }, 
\end{equation}
up to a prefactor which we are not concerned with.

In the $4d$ CFT we consider, these amplitudes are not present.  Therefore, as suggested  by Lionel Mason and Atul Sharma,  they must be cancelled by an axion exchange. Because the axion is part of a Green-Schwarz mechanism, this should be a tree-level exchange of axions. 

Working directly with the Lagrangian \eqref{eqn:sample_lagrangian}, we can see the $4$-point one-loop  amplitude as follows.
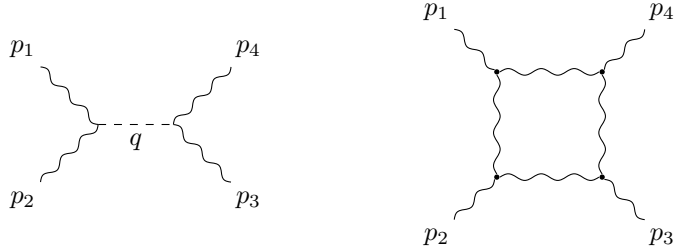
\begin{figure}
		\begin{center}
	\begin{tikzpicture}
	\begin{scope}
		\node(P1) at (-1,1) {$p_1$};
		\node(P2) at (-1,-1) {$p_2$};
		\node(P3) at (2,-1) {$p_3$};
		\node(P4) at (2,1) {$p_4$};
		\draw[decorate,decoration=complete sines] (P1)  -- (0,0); 
		\draw[decorate,decoration=complete sines] (P2) -- (0,0);
		
		\draw[dashed] (0,0) -- (1,0) node[midway, below]{$q$}; 	
		\draw[decorate,decoration=complete sines] (P3) -- (1,0); 
		\draw[decorate,decoration=complete sines] (P4) -- (1,0) ;
	\end{scope}
	\begin{scope}[shift={(6,0)}]
		\node(P1) at (-1.5,1.5) {$p_1$};
		\node(P2) at (-1.5,-1.5) {$p_2$};
		\node(P3) at (1.5,-1.5) {$p_3$};
		\node(P4) at (1.5,1.5) {$p_4$};

		\node[circle,fill=black, inner sep=0pt, minimum size=2pt ](V1) at (-0.7,0.7) {};
		\node [circle,fill=black, inner sep=0pt, minimum size=2pt ] (V2) at (-0.7,-0.7) {};
		\node[circle,fill=black, inner sep=0pt, minimum size=2pt ] (V3) at (0.7,-0.7) {};
		\node[circle,fill=black, inner sep=0pt, minimum size=2pt ](V4) at (0.7,0.7) {};

		\draw[decorate,decoration=complete sines] (P1)  -- (V1);	
		\draw[decorate,decoration=complete sines] (P2)  -- (V2);
		\draw[decorate,decoration=complete sines] (P3)  -- (V3); 
		\draw[decorate,decoration=complete sines] (P4)  -- (V4); 
		\draw[decorate, decoration=complete sines](V1) -- (V2);  
\draw[decorate, decoration=complete sines](V2) -- (V3);  
\draw[decorate, decoration=complete sines](V3) -- (V4);
\draw[decorate, decoration=complete sines](V4) -- (V1);  
	\end{scope}
\end{tikzpicture}
	\end{center}
	\caption{On the left we have the tree-level axion exchange, which by a Green-Schwarz mechanism matches the one-loop four-point amplitude on the right  \label{fig:axion_exchange}}
\end{figure}
In Figure \ref{fig:axion_exchange} we depict the exchange of an axion.  The axion propagator is $\frac{1}{q^4}$, because the Lagrangian \eqref{eqn:sample_lagrangian} has a fourth-order kinetic term for the axion. The axion couples to the gauge field by $F(A)^2$, which on-shell is the same as $F(A)_{+}^2$.  In spinor-helicity notation, the coupling of the axion to the gauge field can be written $[i,j]^2$. 

Therefore, the amplitude for the diagram \eqref{fig:axion_exchange} is
\begin{equation} 
	\frac{[12]^2[34]^2} { q^4 }. 
\end{equation}
Momentum conservation at the vertices, together with the fact that the incoming particles are massless, tells us that  $q^2 = 2 p_1 \cdot p_2 = 2 p_3 \cdot p_4$.  Therefore the amplitude (including the colour factors) is
\begin{equation} 
	\frac{[12][34]} { \ip{12} \ip{34}  } \op{tr}(\t_{a_1} \t_{a_2}) \op{tr}(\t_{a_3} \t_{a_4} ).  
\end{equation}
The expression $\frac{[12][34]} { \ip{12} \ip{34} }$ is totally symmetric\footnote{This is a consequence of conservation of momentum, which in spinor-helicity notation takes the form $\sum_{j}  \ip{ij} [jk] = 0$.  With four particles, this means $[12]\ip{13} = [42]\ip{43}$, so that $[12][34] \ip{13}\ip{24} = [13][24] \ip{12}\ip{34}$.  }. The Green-Schwarz mechanism on twistor space works precisely when the 
\begin{equation} 
	\op{tr}(\t_{(a_{1}} \t_{a_2} ) \op{tr}(\t_{a_3} \t_{a_{4})})   =\lambda_\g^2 \op{Tr}(\t_{(a_{1}} \t_{a_2} \t_{a_3} \t_{a_{4})})  
	\label{eqn:lieidentity} 
\end{equation}
where on both sides we have symmetrized the colour indices.  We conclude that the amplitude for the diagram \eqref{fig:axion_exchange} is proportional to
\begin{equation} 
	\frac{[12][34]} { \ip{12} \ip{34}  } \op{Tr}(\t_{a_1} \t_{a_2}\t_{a_3} \t_{a_4} )  
\end{equation}
which is the correct one-loop amplitude. 

How can we see this from the chiral algebra perspective?  The chiral algebra construction only works in the presence of the axion field, as it requires the theory on twistor space to be anomaly free. However, we are free to add local operators to the $4d$ theory, as we did when moving from self-dual gauge theory to the integrand for Yang-Mills theory by adding on $\op{tr}(B^2)$. 

To see the one-loop all $+$ scattering amplitudes we will need to add a local operator which has the effect of decoupling the axion field.  A first guess might be to try to add $\rho F(A)^2$. This doesn't work, however, as only the derivatives of $\rho$ -- and not $\rho$ itself -- are really part of the $4d$ theory ($\rho$ is a periodic scalar).  

What does work is to introduce the operator $(\tr\rho )^2$. If we add this term to the Lagrangian, then with the appropriate coefficient it will cancel the kinetic term of the axion field.  Introducing it as a local operator will have much the same effect: scattering processes in the presence of the operator $(\tr\rho)^2$ will cancel those processes which have a single axion exchange, as long as the sum of the external momenta vanishes.

We conclude that the one-loop all $+$ scattering amplitudes should be chiral algebra correlators using the conformal block corresponding to $(\tr\rho)^2$.  To determine this conformal block,  we will first find the conformal block corresponding to $\tr \rho$.  This is a Lorentz invariant conformal block which involves one axion field and no other fields.  It must pair with an operator in the chiral algebra which is of spin $0$, dimension $-2$,  and invariant under $SU(2)_+$.  Looking at Table \ref{table:chiralalgebra}, we see that the only such operator is $F[0,0]$, so that 
\begin{equation} 
	\ip{\tr \rho \middle| F[0,0](z) } = C 
\end{equation}
for some non-zero constant $C$. (As usual, other correlation functions in the presence of this conformal block are determined from this identity by the OPE).  Since in this section we are only computing the amplitude up to an overall prefactor, we will set $C = 1$.  

Similarly, we must have
\begin{equation} 
	\ip{(\tr \rho)^2 \middle| F[0,0](z_1) F[0,0](z_2)} = 1. 
\end{equation}
Let us now use the OPEs in \eqref{eqn:axion_opes} to derive the amplitude.  Since we are interested in the four-point all $+$ amplitude, we need to consider OPEs where four $J$'s become two $F[0,0]$'s.  The only relevant OPE is
\begin{equation} 
	J^a[1,0] (z_1) J^b[0,1](z_2)  =\frac{1}{z_{12}} F[0,0] \op{tr}(\t^a \t^b).  
\end{equation}
 $J[1,0]$, $J[0,1]$ form a doublet under $SU(2)_+$.   

 It is convenient to arrange them into a generating function in terms of an auxiliary spinor $\mu$, as in equation \eqref{eqn:generating_function}: 
\begin{equation} 
	J[1](z,\mu) = J[1,0](z)\mu^{\dot 1}  +  J[0,1](z) \mu^{\dot 2} .
\end{equation}
 We then identify $[ij] = \eps_{\dot\alpha \dot\beta} \mu^{\dot\alpha}_i v^{\dot\beta}_j$, and as before $\ip{ij} = z_i - z_j$. 

In this notation, we find 
\begin{equation} 
	J^a[1] (z_1, \mu_1) J^b[1](z_2, \mu_2)  = \frac{[12]}{\ip{12}} F[0,0] \op{tr}(\t^a \t^b)   \what{\lambda}_{\g}  
\end{equation}
From this, it is immediate that
\begin{multline} 
	\ip{(\tr \rho)^2 \mid J^{a_1}[1](z_1,v^\alpha_1)  J^{a_2}[1](z_2,v^\alpha_2)  J^{a_3}[1](z_3,v^\alpha_3)  J^{a_4}[1](z_4,v^\alpha_4) } \\
=   (\what{\lambda}_{\g})^2  \frac{[12][34]} {\ip{12} \ip{34} } \op{tr}(\t_{a_1} \t_{a_2} ) \op{tr}(\t_{a_3} \t_{a_4} )  + \text{ permutations }. \end{multline}	

Since $[12][34]/\ip{12}\ip{34}$ is totally symmetric,  using equation \eqref{eqn:lieidentity}, we can rewrite this as
\begin{equation} 
	\ip{(\tr \rho)^2 \mid J^{a_1}[1](z_1,\zbar_1)  J^{a_2}[1](z_2,\zbar_2)  J^{a_3}[1](z_3,\zbar_3)  J^{a_4}[1](z_4,\zbar_4) } =  \frac{1} { (2 \pi \i)^{3} 12 }   \frac{[12][34]}{\ip{12}\ip{34}} \op{Tr}(\t^{(a_1} \dots \t^{a_4)}) 
\end{equation}
which is the correct amplitude, up to normalization.  

More generally, in \S \ref{s:oneloop}, we show that one loop amplitudes for $n$ positive helicity gluons match\footnote{We do not attempt to match the overall normalization, which can be factored into the normalization of the conformal  block $(\tr \rho)^2$. } chiral algebra correlators. For this, we use the generating function of equation \eqref{eqn:generating_function}, 
\begin{equation} 
	J(\mu,z) = \sum \frac{(\mu^{\dot{1}})^r (\mu^{\dot{2}})^s}{r! s!} J[r,s](z) 
\end{equation}
We have 
\begin{multline} 
	\ip{ (\tr \rho)^2 \mid J_{a_1}(\mu_1, z_1) \cdots J_{a_n}(\mu_n, z_n) } \\
	=   \frac{1}{n} \sum_{\sigma \in S_n}  \frac{   \sum_{1 \le i_1 < i_2 < i_3 < i_4 \le n} \ip{\sigma_{i_1} \sigma_{i_2}} [\sigma_{i_2} \sigma_{i_3} ] \ip{\sigma_{i_3} \sigma_{i_4} } [\sigma_{i_4} \sigma_{i_1} ]       }{\ip{\sigma_1 \sigma_2} \ip{\sigma_2 \sigma_3}  \dots \ip{\sigma_n \sigma_1} } \op{Tr}(\t_{a_{\sigma_1}} \dots \t_{a_{\sigma_n}} )
\end{multline}
matching, up to normalization,  the one-loop amplitudes computed in \cite{Bern:1993qk} and \cite{Mahlon:1993fe}.  

\subsection{WZW correlators as scattering amplitudes in the presence of an axion}
We have expressed many amplitudes of gauge theory in terms of correlators of a chiral algebra which includes the Kac-Moody algebra at level zero. One can ask, is it possible to modify the gauge theory so that the Kac-Moody algebra acquires a level? 

We will see that we can do this by consider gauge theory in the presence of an axion field with a logarithmic profile.  This means we add the term
\begin{equation} 
	\int \log \norm{x}^k F(A)^2 
\end{equation}
to the Yang-Mills Lagrangian.

We study the tree-level scattering amplitudes where all incoming particles are of positive helicity.  Without the axion, this amplitude vanishes. However, in the presence of the axion, we find it is non-zero and is equal to the correlators of the currents in chiral WZW model at level $k$: 
\begin{equation} 
	\ip{ J^{a_1}(z_1) \dots J^{a_n} (z_n) }_{WZW_k} = \text{ scattering amplitudes of } n \text{ gluons} \label{eqn:WZW_amplitude} 
\end{equation}
As in the Parke-Taylor formula, we can rewrite the left hand side using the spinor-helicity formalism.  We trivialize the canonical bundle of $\CP^1$ using the meromorphic $1$-form $\d z$. This allows us to view the left hand side of \eqref{eqn:WZW_amplitude} as a rational function in the $n$ variables $z_i$.  Since it is invariant under an overall translation, it can be rewritten as an expression in $\ip{ij} = z_{ij}$.  

For instance, the colour-ordered single-trace $WZW_k$ correlators are given by the Parke-Taylor denominator  
\begin{equation} 
	\ip{J^{a_1} (z_1) \dots J^{a_n} (z_n) }_{WZW_k}  = - k \op{tr}(\t^{a_1} \dots \t^{a_n} ) \frac{1}{z_{12} z_{23} \dots z_{(n-1)n} z_{n1} } +  \dots 
\end{equation}
where $\dots$ indicates terms with a different colour ordering, as well as terms of order $k^2$ and higher that are not single-trace.

In our identity \eqref{eqn:WZW_amplitude} we include all terms on the right hand side, including multi-trace terms.  Multi-trace terms can appear in the gauge theory scattering amplitude from diagrams where the background axion field appears several times, in disconnected tree-level diagrams.   (In our identification between amplitudes and correlators, it is most natural to include disconnected diagrams; it just so happened that these did not play a role in our other computations).

We also expect, but do not prove, that scattering amplitudes of $k$ states of positive conformal dimension with $n$ states of negative conformal dimension are given by WZW correlators in the presence of $k$ modules.

\subsection{Connections to celestial holography}

As advertised, our program has many natural connections to the celestial holography program. As we work towards derivations of our main result Thm \ref{mainthm}, we explain these connections from the point of view of 6d holomorphic theories on twistor space. Each such 6d theory can be viewed as the parent theory of both the 4d CFTs described above and, via Koszul duality, the 2d chiral algebra. 

Let us briefly recall the appearance of 2d chiral algebras in celestial holography. The chiral algebras in that context capture asymptotic symmetries in flat spacetime. Although the story of asymptotic symmetries begins in the usual momentum space basis (see e.g. \cite{Strominger:2017zoo}), we will be largely interested in amplitudes of massless states expressed in the celestial, or conformal, basis. To pass to the conformal basis from the momentum basis, one performs a Mellin transform for massless\footnote{There is also a transform for massive states, which we will not consider further in the present paper.} momentum eigenstates $\mathcal{O}_i$. From this procedure, one can obtain a normalizable basis of 2d conformal primaries in which the dilatation operator is diagonal. Restricting to the principle series for massless operators $\Delta \in 1 + i \lambda, \lambda \in \mathbb{R}$ guarantees that the operators are invertible and normalizable with respect to the Klein-Gordon inner product \cite{Pasterski:2017kqt}. The 4d scattering amplitudes expressed in terms of the Mellin-transformed variables which diagonalize boosts are referred to as \textit{celestial amplitudes}. Recent reviews on aspects on celestial amplitudes include \cite{Raclariu:2021zjz, Pasterski:2021rjz}.

Recall that the Mellin transform and its inverse for massless states are 
\begin{align}
    \hat{\mathcal{O}}^{\pm}(\Delta, z, \bar{z})&= \int_0^{\infty}d\omega \omega^{\Delta-1} \mathcal{O}(\pm \omega, z, \bar{z}) \\
    \mathcal{O}(\pm |\omega|, z, \bar{z}) &= \int_{1 - i \infty}^{1 + i \infty}{d \Delta \over 2 \pi i }|\omega|^{-\Delta}\hat{\mathcal{O}}^{\pm}(\Delta, z, \bar{z})
\end{align}where the signs $\pm$ denote in, respectively out, states. We will parameterize in/out null momenta in the usual celestial presentation via 
\begin{equation}
    p(z, \bar{z}, \pm \omega) = \pm \omega (1 + |z|^2, 2 \textrm{Re}(z), 2 \textrm{Im}(z), 1 - |z|^2)
\end{equation} where $z, \bar{z}$ are coordinates on the celestial sphere\footnote{We also remind the reader, as in the previous subsection, that null momenta can be determined by a choice of two-component complex spinor up to scale by $p_{\alpha \dot\alpha} =\mu_{\dot\alpha}\lambda_{\alpha}$, so that the direction of null vector is given by $\lambda$, $\mu$ up to scale, or equivalently a point $z$ on the celestial sphere $\mathbb{CP}^1$. In affine coordinates, $\lambda = (1, \  z)$, $\langle \lambda_1  \lambda_2 \rangle = z_1 - z_2$, and so on.}.

The Mellin transformed scattering amplitude transforms as a 2d conformal correlation function,  and the Mellin transformed operators correspond to insertions of local operators (for massless states). The Lorentz symmetry can be interpreted in this basis, for example, as a global conformal symmetry $SL(2, \mathbb{C})$. 

We will be most interested in the ``conformally soft'' symmetries of celestial amplitudes, given by currents satisfying $h \rightarrow 0$ (for negative helicity states) or $\bar{h} \rightarrow 0$ (for positive helicity states). In the momentum space basis, soft theorems are associated to conservation laws corresponding to large gauge symmetries. One can check by direct computation that the $\Delta \rightarrow 1$ \footnote{We recall that $h = {1 \over 2}(\Delta + J), \bar{h} = {1 \over 2}(\Delta - J)$ in terms of the conformal dimension and spin.} limit of a Mellin operator of positive helicity coincides with the $\omega \rightarrow 0$ limit of the momentum space operator. Similar limits can be taken to extract the subleading soft factors, assuming the insertion of the operator in the amplitude falls off sufficiently fast with energy; for example the subleading soft photon in the celestial basis corresponds to a $\Delta \rightarrow 0$ limit. Taking similar $\Delta \rightarrow -n$ limits for all $n=-1, 1, 0, \ldots$ in an expansion $\mathcal{O}^+ = \sum_k \omega^k \mathcal{O}^+_k$ of the positive helicity states leads to an infinite tower of conformally soft currents, which yield conformally soft constraints on amplitudes \cite{Pate:2019lpp}. 

It is the algebra of this infinite tower of currents \cite{Guevara:2021abz, Strominger:2021lvk, Himwich:2021dau} that we revisit from a twisted holography point of view, and compute using Koszul duality. It is not obvious that there should be a chiral algebra hidden in scattering amplitudes of massless particles.  For self-dual theories, 
as we will see, this is neatly explained by twistor theory. However, the symmetries are known to persist beyond the self-dual limit, at least at tree-level and for certain one-loop amplitudes \cite{Ball:2021tmb}.

Let us recall the celestial holography results on 2-to-1 scattering processes of positive-helicity massless particles at tree-level, which  possess a beautiful chiral algebra structure. Consider the situation where $z, \bar{z}$ are independent real coordinates, so that four dimensional spacetime becomes signature $(2, 2)$ and the Lorentz group becomes $SL(2, \mathbb{R})_L \times SL(2, \mathbb{R})_R$. The OPEs of celestial operators at tree-level were studied in \cite{Pate:2019lpp}, where it was shown that poles in the OPEs of operators (say, for $z_{12} \rightarrow 0, \bar{z}_{12}$ fixed) on the celestial sphere are the Mellin transforms of collinear limits\footnote{For massless particles that couple via a three-point vertex, collinear limits arise when the particles' momenta become parallel, giving a ${1 \over p_1 \cdot p_2}$ pole.} of momentum space operators, and hence can be computed using Feynman diagrammatics. Alternatively, taking an ansatz for the form of an OPE in a holomorphic collinear expansion, \cite{Pate:2019lpp} showed that the OPE coefficients can also be fixed by application of leading and sub(sub)leading soft theorems and global conformal (i.e. Lorentz) invariance. A generalization of these OPE computations, whose results we will reinterpret in a twistorial language, was developed in \cite{Guevara:2021abz}, and reorganized in \cite{Strominger:2021lvk}. There, the $SL(2, \mathbb{R})_R$ (which acts on the $\bar{z}$ coordinates) descendants of the primary operators in the OPE were resummed using an OPE block. The infinite tower of conformally soft operators were further expanded in powers of $\bar{z}$. Studying the algebra of holomorphic modes resulting from this expansion (i.e. the holomorphic residues of the resulting $\bar{z}$-Laurent series) naturally produces the Kac-Moody algebra of area-preserving symmetries of the plane, i.e. the loop algebra of $w_{1 + \infty}$. Restricting to those modes which form representations of $SL(2, \mathbb{R})_R$ produces the corresponding wedge subalgebra. 

We will study these chiral algebras (more precisely, enlargements of them which include axion contributions and states of both helicities) from a twistorial point of view. It is not a surprise that twistor theory places a role in such symmetries. In Penrose's \cite{Penrose:1976jq}   non-linear graviton construction, it was found that solutions to the self-dual Einstein equation can be built on twistor space by a gluing construction, where the gluing data is an element of \emph{precisely} the same Lie algebra as found in the recent work \cite{Strominger:2021lvk}.  For gauge theory, the same thing holds \cite{Ward:1977ta} but where the gluing data in the Penrose-Ward correspondence is the Lie algebra of celestial symmetries for gluons found in \cite{Guevara:2021abz}. As shown in Proposition 3.5 of \cite{DunajskiMason} and recently emphasized, and given an ambitwistor string interpretation, in \cite{Adamo:2021lrv}, twistor space makes manifest that the gravitational celestial algebra studied in \cite{Guevara:2021abz, Strominger:2021lvk} is the loop algebra $Lw_{1 + \infty}$ of Poisson diffeomorphisms of the plane (with similar results for gauge theory). Indeed, twistor space neatly and geometrically captures, via the usual Penrose-Ward correspondence, the physics of the self-dual sectors of gauge theory and gravity.

As we discussed already, our main result is a derivation of form factors and scattering amplitudes from chiral algebra correlators.  
We also attempt to address some other questions in celestial holography.  We use a twistor construction to build a bulk $3d$ theory, whose boundary algebra is the celestial chiral algebra.  We also discuss how the Koszul duality perscription suggests that the chiral algebras must be quantized, although we do not fully understand the quantization. 


\subsection{Outline}

Our plan for the rest of this paper is as follows. In section \S \ref{s:twistor} we will review some basic aspects of the twistor correspondence. In \S \ref{s:anomalies} we will review the appearance of one-loop anomalies in holomorphic theories on twistor space, their cancellation in certain theories, and briefly discuss the ramifications for chiral algebras. In \S \ref{s:gauge} we will explain how the celestial symmetry algebras are realized as gauge transformations of certain holomorphic theories on twistor space. In \S \ref{s:states}, we illustrate how to obtain the 4d states of negative conformal dimension, which correspond to chiral algebra generators, from twistor space using the free scalar field theory as an example. In \S \ref{s:3d}, we present an alternative way to understand the 2d chiral algebra, as the boundary of a 3d theory holomorphic-topologically twisted theory. We discuss various features of this bulk-boundary system. In \S \ref{s:Koszul}, we explain an alternative, efficient way to obtain the celestial chiral algebras using inspiration from twisted holography, via Koszul duality. In \S \ref{s:blocks} we derive our main result, Theorem \ref{mainthm}, and illustrate it by reproducing the Parke-Taylor formula for tree-level MHV amplitudes. In \S \ref{s:CSW} we derive the formula for the (unintegrated) CSW formula in terms of chiral algebra data. In \S \ref{s:oneloop}, we derive the all-plus one-loop amplitudes. In \S \ref{s:KacMoody}, we derive the result that tree-level amplitudes in gauge theory in the presence of an axion field with nontrivial profile are captured by Kac-Moody correlators. We conclude with brief discussions of works in progress and open questions in \S \ref{s:conclusions}.

\section{Recollections on the twistor correspondence}\label{s:twistor}

Twistor space $\PT$ is the total space of the bundle $\Oo(1) \oplus \Oo(1)$ over $\CP^1$.  Since the pioneering work of Penrose \cite{Penrose:1977in} (see \cite{Adamo} for a pedagogical review), it has been known that \emph{holomorphic} field theories on twistor space become \emph{massless} field theories on real space, either Euclidean or Lorentzian\footnote{Twistor space is signature-agnostic, and gives rise to theories which can be analytically continued to any signature.}.

Here we will recall some aspects of the twistor correspondence. Let us give twistor space coordinates $z,v_1,v_2$ where $z$ is a coordinate on $\CP^1$, and $v_i$ are coordinates on the $\Oo(1)$ fibres.  We choose our coordinates so that $v_i$ have poles at $z = \infty$.

Twistor space is closely connected with analytically-continued space-time $\C^4$. Let us give $\C^4$ complex coordinates $u_i$ with complex metric $\sum \d u_i^2$.  There is a bijection between $\C^4$ and complex lines $\CP^1 \subset \PT$ (which are embedded linearly).

A point $(u_1,\dots,u_4) \in \C^4$ corresponds to the complex line in $\PT$ cut out by the equations 
\begin{equation} 
	\begin{split}
		v_1 &= u_1 + \i u_2 + z(u_3 - \i u_4)\\
		v_2 &= u_3 + \i u_4 - z(u_1 - \i u_2).
	\end{split}
\end{equation}
We refer to this curve as $\CP^1_u$.

Two curves $\CP^1_x$, $\CP^1_y$ for $x,y \in \C^4$ intersect if and only if $x,y$ are null-separated, that is, $\norm{x-y}^2 = 0$. In particular, if we work in Euclidean signature by taking all the $u_i$ to be real, then the curves $\CP^1_u$ are all disjoint. The curves $\CP^1_u$, for $u \in \R^4$, foliate $\PT$, giving an isomorphism of real manifolds
\begin{equation} 
	\PT \iso \R^4 \times \CP^1. 
\end{equation}

\subsection{The free scalar field theory}
The very simplest example of the twistor correspondence relates the free scalar field theory on $\mathbb{R}^4$ with an Abelian gauge theory on twistor space.  The field of the gauge theory is\footnote{In general, the Penrose transform is a bijection between zero rest mass fields of helicity $h$ on analytically continued spacetime, and the Dolbeault cohomology group $H^{0, 1}(\PT, \mc{O}(2h-2))$. Here, $2h-2$ can be viewed as the weight of the field under the homogeneous scaling symmetry of twistor space, viewed as an open subset of $\mathbb{CP}^3$. So, we are studying $(0, 1)$-forms on twistor space with fixed weights.}
\begin{equation} 
	\mc{A} \in \Omega^{0,1}(\PT, \Oo(-2)). 
\end{equation}
The Lagrangian is
\begin{equation} 
	\int_{\PT} \mc{A} \dbar \mc{A} 
\end{equation}
which makes sense because the canonical line bundle $K_{\PT} = \Oo(-4)$. This field is subject to gauge transformations
\begin{equation} 
	\mc{A} \mapsto \mc{A} + \dbar \chi 
\end{equation}
where 
\begin{equation} 
	\chi \in \Omega^{0,0}(\PT, \Oo(-2)).	 
\end{equation}
Gauge-equivalence classes of on-shell fields on twistor space are the Dolbeault cohomology group
\begin{equation} 
	H^1(\PT, \Oo(-2)). 
\end{equation}
Penrose \cite{Penrose:1977in} shows that this cohomology group is isomorphic to the space of entire analytic functions
\begin{equation} 
	\phi : \C^4 \to \C 
\end{equation}
which are harmonic:
\begin{equation} 
	\partial_{x_i} \partial_{x_i} \phi = 0. 
\end{equation}
Such a $\phi$ of course restricts to a solution of the free-field equations in \emph{any} signature.  

We build the field $\phi$ from the gauge-field $\mc{A}$ on twistor space by
\begin{equation} 
	\phi(x) = \int_{\CP^1_x} \mc{A} .
\end{equation} 
The expression on the right hand side makes sense, as $\mc{A}$ is twisted by $\Oo(-2)$ which is the canonical bundle of $\CP^1_x$. This measurement of $\mc{A}$ is gauge invariant.

\subsection{Self-dual Yang-Mills theory}
The next example of the twistor correspondence is the Penrose-Ward correspondence.  This relates self-dual Yang-Mills on $\R^4$ with holomorphic BF theory on twistor space.

Fix a simple Lie algebra $\g$. Self-dual Yang-Mills theory has fields
\begin{equation} 
	\begin{split} 
		A &\in \Omega^1(\R^4,\g) \\
		B & \in \Omega^2_-(\R^4,\g)
	\end{split}
\end{equation}
with the Chalmers-Siegel action \cite{CS}
\begin{equation} 
	\int \op{Tr} (B \wedge F(A)_-). 
\end{equation}
If we add on the term 
\begin{equation} 
	g_{YM} \int \op{Tr}(B^2), 
\end{equation}
the theory becomes ordinary Yang-Mills theory in the first-order formulation\footnote{With a certain value of the $\theta$-angle.}.

Self-dual Yang-Mills theory on twistor space \cite{Boels:2006ir, Boels:2007qn} is represented by the gauge theory with fields
\begin{equation} 
	\begin{split} 
		\mc{A} & \in \Omega^{0,1}(\PT,\g)\\
		\mc{B} & \in \Omega^{3,1}(\PT,\g)
	\end{split}
\end{equation}
with Lagrangian
\begin{equation} 
	\int_{\PT} \op{Tr} ( \mc{B} F^{0,2}(\mc{A})). 
\end{equation}
The fields $\mc{A},\mc{B}$ are subject to two kinds of gauge transformations, with generators
\begin{equation} 
	\begin{split} 
		\chi & \in \Omega^{0,0}(\PT,\g) \\
		\nu & \in \Omega^{3,0}(\PT,\g)
	\end{split}
\end{equation}
where
\begin{equation} 
	\begin{split} 
		\delta \mc{A} &= \dbar \chi + [\mc{A},\chi] \\
		\delta \mc{B} &= \dbar \nu + [\mc{B},\chi].
	\end{split}
\end{equation}

\subsection{The non-linear graviton construction}
Finally, we will discuss a more complicated twistor transform, which relates the self-dual limit of Einstein gravity with a certain BF theory on twistor space \cite{Mason:2007ct,Sharma:2021pkl}.

Let us introduce the holomorphic Poisson tensor
\begin{equation} 
	\pi = \partial_{v_1} \partial_{v_2} 
\end{equation}
on twistor space. This vanishes to order $2$ at $z = \infty$, and so can be thought of as a Poisson tensor twisted by $\Oo(-2)$.

The fields of self-dual gravity on twistor space consist of a field
\begin{equation} 
	\mc{H} \in \Omega^{0,1}(\PT, \Oo(2)) 
\end{equation}
and a Lagrange multiplier field
\begin{equation} 
	\beta \in \Omega^{3,1}(\PT, \Oo(-2)). 
\end{equation}
We can think of $\mc{H}$ has having a pole of order $2$ at $z = \infty$, and $\beta$ as being a $(3,1)$ form with a zero of order two at $z = \infty$.  

The Lagrangian is
\begin{equation}
	\begin{split}
		& \int \beta \dbar \mc{H} + \tfrac{1}{2} \beta \{\mc{H}, \mc{H}\} \\
		&= \int \beta \dbar \mc{H} + \tfrac{1}{2} \beta \eps_{ij} \partial_{v_i}\mc{H} \partial_{v_j} \mc{H}. 
	\end{split} 
\end{equation}
The integrand in both terms is a $(3,3)$ form with no poles.  For the kinetic term, the order two zero in $\beta$ cancels the order two pole in $\mc{H}$, and for the interaction, the order four pole coming from the two copies of $\mc{H}$ is canceled by the order $2$ zero from $\beta$ and the order two zero from the Poisson tensor $\partial_{v_1} \partial_{v_2}$.

Just like holomorphic BF theory, this theory has two kinds of gauge transformations, generated by
\begin{equation} 
	\begin{split} 
		\chi & \in \Omega^{0,0}(\PT,\Oo(2)) \\
		\nu & \in \Omega^{3,0}(\PT,\Oo(-2))
	\end{split}
\end{equation}
where
\begin{equation} 
	\begin{split} 
		\delta \mc{H} &= \dbar \chi + \eps_{ij} \partial_{v_i} \mc{H} \partial_{v_j} \chi \\
		\delta \beta &= \dbar \nu + \eps_{ij} \partial_{v_i} \beta \partial_{v_j} \chi 
	\end{split}
\end{equation}
The gauge transformations commute as
\begin{equation} 
	\begin{split} 
		[\chi_1,\chi_2] &= \eps_{ij} \partial_{v_i} \chi_1\partial_{v_j} \chi_2 \\
		[\chi,\nu] &= \eps_{ij} \partial_{v_i} \chi \partial_{v_j} \nu \\
		[\nu_1,\nu_2] &= 0.
	\end{split}
\end{equation}

\section{Anomalies on twistor space}\label{s:anomalies}
Before we proceed, we should remark that holomorphic field theories on twistor space tend to suffer from anomalies \cite{Costello:2019jsy,Costello:2021bah}. One can not build the chiral algebra at the quantum level (at least using our methods) unless the theory is anomaly free. 

The anomalies for the Poisson BF theory related to self-dual gravity are currently not understood.   For the holomorphic BF theory giving self-dual Yang-Mills, the twistor anomaly is well understood and in fact easy to calculate using the index theorem.  There is an anomaly associated to the box diagram in Figure \ref{fig:anomaly}.   
\begin{figure}
	\begin{center}
		\includegraphics[scale=0.25]{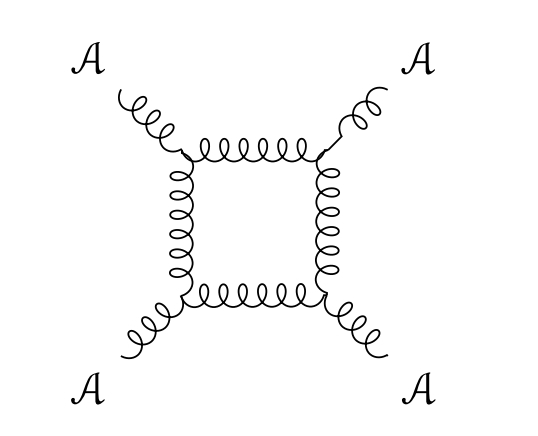}
	\end{center}
	\caption{\label{fig:anomaly} The anomaly in holomorphic BF theory  } 
\end{figure} 
In \cite{Costello:2021bah} it is shown that this anomaly can be canceled by a Green-Schwartz mechanism in certain cases.  For this to work, we need the gauge Lie algebra $\g$ to be $\mf{sl}_2$, $\mf{sl}_3$, $\mf{so}(8)$ or one of the exceptional algebras.  (We can also include matter and cancel the anomaly for $\mf{so}(N_c)$ with $N_f = N_c - 8$, or $\mf{sl}(N_c)$ with $N_f = N_c$, though we will not discuss these cases in this paper).

In these cases, we introduce a new field on twistor space
\begin{equation} 
	\eta \in \Omega^{2,1}(\PT) 
\end{equation}
constrained so that $\partial \eta = 0$, and subject to gauge transformations $\eta \mapsto \eta + \dbar \gamma$ for $\gamma \in \Omega^{2,0}(\PT)$.  The Lagrangian for $\eta$ is the free limit of the Kodaira-Spencer Lagrangian \cite{BCOV}:
\begin{equation} 
	\tfrac{1}{2} \int \dbar \eta \partial^{-1} \eta. 
\end{equation}
This is then coupled to the gauge field $\mc{A}$ by 
\begin{equation} 
	\frac{\lambda_{\g}}{4 (2 \pi \i)^{3/2} \sqrt{3} }\int \eta \mc{A} \partial \mc{A} 	 
\end{equation}
where $\lambda_{\g}$ is a constant such that
\begin{equation} 
	\op{Tr}(X^4) = \lambda_{\g}^2 \op{tr}(X^2)^2 
\end{equation}
where $\op{Tr}$ means the trace in the adjoint and $\op{tr}$ that in the fundamental. 

In four dimensions, the field $\eta$ introduces a new axion field $\rho$, which couples to self-dual Yang-Mills by
\begin{equation} 
	 \half \int (\Lap \rho)^2  +  \lambda_{\g} \frac{1}{8 \pi  \sqrt{3} }\int  \d \rho CS(A) \label{eqn_sdym_axion}
\end{equation}
where $CS(A)$ is the Chern-Simons three-form.  The field $\rho$ is related to $\eta$ by
\begin{equation} 
	 \rho(x) = 	\frac{\i }{ \sqrt{8 \pi \i} }     \int_{\CP^1_x} \partial^{-1} \eta 
\end{equation}
where we are integrating over the $\CP^1$ corresponding to $x \in \R^4$. 

It is best to understand the coupling between $\eta$ and $\mc{A}$ in the BV formalism, where $\mc{A}$ gets extended to a field in $\Omega^{0,\ast}(\PT,\g)[1]$, and $\eta$ to one in $\Omega^{2,\ast}(\PT, \g)[1]$. The symbol $[1]$ indicates a shift in ghost number, so that fields in Dolbeault degree $i$ are in ghost number $1-i$. Then, the interaction takes the same form,
\begin{equation} 
	\frac{\lambda_{\g}}{4 (2 \pi \i)^{3/2} \sqrt{3} }\int \boldsymbol{\eta} \boldsymbol{\mc{A}} \partial\boldsymbol{ \mc{A}} 	 \label{eqn:bv_action} 
\end{equation}
where $\boldsymbol{\eta}$ and $\boldsymbol{\mc{A}}$ refer to the fields including all Dolbeault degrees.  By decomposing this action into components, we can read off the non-linear terms in the gauge transformations. These come from components where one of the fields is in Dolbeault degree zero, one in Dolbeault degree $1$, and one in Dolbeault degree $2$ (and so is an anti-field).  Noting that the Dolbeault degree $2$ component of $\boldsymbol{\mc{A}}$ is the anti-field to $\mc{B}$, we find the following extra gauge transformations:
\begin{equation}
	\begin{split} 
		\delta_{\chi} \eta &=  \frac{\lambda_{\g}}{2 (2 \pi \i)^{3/2} \sqrt{3} } \partial \chi \partial \mc{A}  	\\
		\delta_{\gamma} \mc{B} &= \frac{\lambda_{\g}}{2 (2 \pi \i)^{3/2} \sqrt{3} }   \gamma  \partial \mc{A} .	
	\end{split}	
\end{equation}
Finally, from the term in equation \eqref{eqn:bv_action} where two fields are in Dolbeault degree $0$ and one is in Dolbeault degree $3$, we find two extra terms in the commutators of the gauge transformations. The first is where a $\gamma \in \Omega^{2,0}$ and a $\chi \in \Omega^{0,0}$ commute to become a $\nu \in \Omega^{3,0}$: 
\begin{equation} 	
	[\gamma,\chi] =    \frac{\lambda_{\g}}{2 (2 \pi \i)^{3/2} \sqrt{3} } \gamma \partial \chi \in \Omega^{3,0}	
\end{equation}
The second is where two $\chi$'s commute to become a $\gamma$:
\begin{equation} 	
	[\chi_1,\chi_2] =    \frac{\lambda_{\g}}{2 (2 \pi \i)^{3/2} \sqrt{3} } \partial \chi_1 \wedge \partial \chi_2  \in \Omega^{2,0}.	
\end{equation}

\section { Celestial symmetry algebras are gauge transformations on twistor space}\label{s:gauge}
As we have reviewed above, the twistor transform relates self-dual gauge theory on $\R^4$ with holomorphic BF theory on twistor space, and self-dual gravity on $\R^4$ with a certain Poisson BF theory on twistor space. In this section, we will see that gauge transformations of the theories on twistor space match the Lie algebra of modes of the chiral algebras found in \cite{Guevara:2021abz,  Strominger:2021lvk}.

\subsection{Chiral algebra for gauge theory}
Let us first recall the chiral algebra of positive helicity gluons in \cite{Guevara:2021abz}.  As briefly reviewed in the introduction, one considers the holomorphic collinear OPE, reorganized to incorporate $SL(2, \mathbb{R})_R$ descendants of any given primary:

\begin{equation}
    \mc{O}^a_{\Delta_1}(z_1, \bar{z}_1) \mc{O}^b_{\Delta_2}(z_2, \bar{z}_2) \sim {-i f^{ab}_c \over z_{12}}\sum_{n=0}^{\infty}B(\Delta_1 - 1 + n, \Delta_2 - 1){\bar{z}_{12}^n \over n!}\bar{\partial}^n \mc{O}^c_{\Delta_1 + \Delta_2 - 1}(z_2, \bar{z}_2).
\end{equation}

The (closed) subalgebra of asymptotic symmetries come from studying conformally soft operators, given by $\Delta_1, \Delta_2 \in \left\lbrace 1, 0, -1, -2, \ldots \right\rbrace$. The authors then consider the following expansion of these  operators in the conformally soft limit:
\begin{equation}
    \textrm{lim}_{\epsilon \rightarrow 0} \epsilon \mc{O}^a_{k + \epsilon}(z, \bar{z}) =  \textrm{lim}_{\epsilon \rightarrow 0} \sum_{n= {k-1 \over 2}}^{{1-k \over 2}} {\epsilon \mc{O}^a_{k + \epsilon, n}(z) \over \bar{z}^{n + (k-1)/2}}
\end{equation} with $k \in \mathbb{Z}_{\leq 1}$. As explained in \cite{Guevara:2021abz}, outside the given range of $n$, the $SL(2, \mathbb{R})_R$ invariant norm vanishes (up to possible contact terms).

One then defines the holomorphic modes $R^{a}_{k, n}(z) := \textrm{lim}_{\epsilon \rightarrow 0} \mc{O}^{a}_{k + \epsilon, n}(z)$, which naturally furnish $(2-k)$-dimensional $SL(2, \mathbb{R})$-representations, $(k-1)/2 \leq n \leq (1-k)/2$ \footnote{There is an equivalent approach to obtaining holomorphic currents from the light transform \cite{Himwich:2021dau}. The latter has been interpreted as a half-Fourier transform on twistor space in \cite{Sharma:2021gcz}, which may be more natural for our purposes}. 

To sum up, the chiral algebra can be written in terms of a sequence of conformal primaries
\begin{equation} 
	R^a_{k,n}(z) 
\end{equation}
where $a$ is a  Lie algebra index, $k$ is an integer telling us the spin of the $SU(2)$ \footnote{Here we refer to $SU(2)$ rather than $SL(2, \mathbb{R})_R$, hoping that it will not cause confusion. Twistor space is signature-agnostic, and the two choices differ by a choice of real form on $SO(4, \mathbb{C})$ that will not be important for us.} representation the operator lives in, and $n$ indicates the weight of a vector in this representation. These operators satisfy the algebra
\begin{equation} 
[R^a_{k, n}, R^b_{l, m}] = -i f^{a b}_c {k' - n + l' - m\choose k' - n}{k' + n + l' + m\choose k' + n} R^c_{k + l - 1, m + n}.	 
\end{equation} where we have introduced the shorthand $j' := (1-j)/2$.
Following \cite{Strominger:2021lvk}, it is more convenient to define the generators via
\begin{equation}
\hat{R}^a_{k', m} = (k' - m)!(k' + m)! R^a_{k, m},    \end{equation} where we relabel the first argument by $k'= (1-k)/2= 0, 1/2, 1, \ldots$ instead of $k=1, 0, -1, \ldots$. These generators obey
\begin{equation}
  [\hat{R}^a_{k', m}, \hat{R}^b_{l', n}]= -i f^{ab}_c \hat{R}^c_{k' + l', m + n}.   
\end{equation}
Finally, we will define the generators
\begin{equation} 
	J^a[m, n] = \hat{R}^a_{m' + n', m - n}. 
\end{equation} (Altogether, the generators $J^a[k, l]$ are related to the currents $S^p_m$ of \cite{Strominger:2021lvk}, via $p = 1 + (k+l)/2, m = (k-l)/2$). 

Using our definition of $m' + n'$, we can see that these are conformal primaries of spin $1-m/2-n/2 $. This is because $v_i$ have charge $-\half$ under rotation of $z$.   They satisfy the OPE
\begin{equation} 
	J^a[m, n] (z) J^b[r, s](z') = \frac{1}{z-z'} f^{ab}_c J^c(z)  [m+r,n+s] \label{J_ope} 
\end{equation}
This OPE implies that the Lie algebra of modes 
\begin{equation} 
	J^a[m,n,k] = \oint J^a[m, n] z^{k} \d z
\end{equation}
satisfies the simple commutation relation
 \begin{equation} 
	 [J^a[m, n, k], J^b[r,s, l] ]= f^{ab}_c J^c[m+ r, n+ s, k+l].	  
 \end{equation}
This Lie algebra appears naturally on twistor space.  

Consider the open subset of twistor space where $z$ is not $0 $ or $\infty$.  On this open subset we have holomorphic coordinates $v_1, v_2, z, z^{-1}$.  We can consider the infinitesimal gauge transformations of the field $\mc{A}$ which preserve the vacuum field configuration $\mc{A}= 0$,  $\mc{B}=0$.  Such gauge transformations are holomorphic maps
\begin{equation} 
	\C \times \C \times \C^\times \to \mf{g}. 
\end{equation}
This is the sub-algebra of the triple loop algebra of $\mf{g}$ that we can write as
\begin{equation} 
	\g[v_1, v_2, z, z^{-1}].	 
\end{equation}
This is identified with the Lie algebra of modes of the celestial chiral algebra, by 
\begin{equation} 
	J^a[m,n,k] = \t^a v_1^m v_2^n z^k. 
\end{equation}
Note that $v_i$ transform, under coordinate transformations of the $z$ plane, as $(\d z)^{1/2}$. This explains why $J^a[m,n,k]$ has spin $-k-m-n$.

The full Lie algebra of gauge transformations of holomorphic BF theory also includes the transformations of $\mc{B}$.  On the same patch of twistor space, these are indexed by
\begin{equation} 
	\til{J}^a[m,n,k] 
\end{equation}
with the commutators 
\begin{equation} 
	\begin{split}
		[ \til{J}^a[m,n,k], \til{J}^b[r,s,l] ] &= 0\\ 
		[J^a[m,n,k], \til{J}^b[r,s,l] ] &= f^{ab}_c \til{J}^c[m+r,n+s,k+l] ].
	\end{split}
\end{equation}
These can be obtained by enlarging the chiral algebra with OPE \eqref{J_ope} by adjoining additional primary operators $\til{J}^a[m,n]$ of spin $-1 -m-n $ with the OPE
\begin{equation} 
	J^a[m, n] (z) \til{J}^b[r, s](z') = \frac{1}{z-z'} f^{ab}_c \til{J}^c(z)  [m+r,n+s]. \label{ope_tilJ} 
\end{equation}
These additional conformal primaries correspond to states of negative helicity. 

As explained in \cite{He_2016, Distler:2018rwu, Stieberger:2015kia} in the energetically soft basis, one does not expect two independent copies (holomorphic and anti-holomorphic) of the same chiral algebra when considering both positive and negative helicity particles. This is due to order-of-limits ambiguities when multiple particles of opposite helicity become soft. Indeed, here the negative helicity states transform under the adjoint representation of the algebra generated by the holomorphic states. 

A priori, we do not expect the simple OPE \eqref{ope_tilJ} to correspond to a celestial chiral algebra of full Yang-Mills theory: only of the self-dual limit at tree level. 

\subsection{Gravitational celestial symmetries}
For gravity, a similar analysis holds.  In \cite{Guevara:2021abz}, it was shown that the celestial symmetries for gravity are given by the Kac-Moody algebra built from a certain infinite-dimensional Lie algebra, described in \cite{Strominger:2021lvk} as the wedge algebra of the Kac-Moody algebra of $w_{1+\infty}$.

The Lie algebra $w_{1+\infty}$ is the Lie algebra of polynomial functions on a cylinder $\C \times \C^\times$, with coordinates $v_1,v_2$:
\begin{equation} 
	w_{1+\infty} = \C[v_1,v_2,v_2^{-1}]. 
\end{equation}
The Lie bracket is the Poisson bracket with respect to the Poisson tensor $\partial_{v_1} \partial_{v_2}$.  Inside this Lie algebra is a copy of $\mf{sl}_2$ with basis $v_1^2$, $v_1 v_2$, $v_2^2$.  

The  wedge algebra of $w_{1+\infty}$ is the subalgebra consisting of vectors which transform in finite-dimensional representations of this $\mf{sl}_2$. Recall that in the celestial holography picture, this was the Lie algebra of $SL(2, \mathbb{R})_R$ \footnote{We will later see, from the geometry of twistor space, that it is natural and expected that $SL(2, \mathbb{R})_L,  SL(2, \mathbb{R})_R$ play very different roles when deriving the chiral algebra.}.   It is easy to see that the wedge algebra $\wedge w_{1+\infty}$ is the subalgebra consisting of polynomials regular at $v_1 = 0$:
\begin{equation} 
	\wedge w_{1+\infty} = \C[v_1,v_2] 
\end{equation}
under the Poisson bracket. In other words, it is the Lie algebra of Hamiltonian vector fields on the plane $\C^2$.

As such, we will use the more standard notation
\begin{equation} 
	\op{Ham}(\C^2) = \wedge w_{1+\infty}
\end{equation}
denoting the algebra of Hamiltonian vector fields on $\mathbb{C}^2$.

In \cite{Strominger:2021lvk}, it is shown that the celestial chiral algebra at tree-level is the Kac-Moody algebra for $\op{Ham}(\C^2)$ at level zero. This can be obtained analogously to the gauge theory case. Namely, we study the conformally soft limits of positive helicity graviton operators, $w_k(z, \bar{z}):= \textrm{lim}_{\epsilon \rightarrow 0}G_{k + \epsilon}((z, \bar{z})$ in the mode expansion
\begin{equation}
    w_k(z, \bar{z}) = \sum_{n = {k-2 \over 2}}^{{2-k \over 2}} {w_{k, n}(z) \over \bar{z}^{n + {k-2 \over 2}}}. 
\end{equation}
We change the normalization of the generators in the same way as the gauge theory case, to absorb some pesky factorials as well as an additional factor of ${1 \over \sqrt{32 \pi G}}$. Similarly, we will relabel the generators by their highest weight under $SL(2, \mathbb{R})$ to define the generators of the chiral algebra
\begin{equation} 
	w[m,n]  
\end{equation}
of spin $2-m/2-n/2$,  corresponding to $v_1^m v_2^n \in \op{Ham}(\C^2)$. (These are related to the generators denoted by $w^p_m$ in \cite{Strominger:2021lvk} using the relations $p-1 = {k+ l \over 2}, m = {k-l \over 2}$, as with the gauge theory re-indexing). 

These have OPEs
\begin{equation} 
	w[m,n](0) w[r,s](z) = \frac{1}{z} ( ms-nr ) w[r+m-1,n+s-1]. 
\end{equation}
Notice that these OPEs are \emph{not} those of the $W_{\infty}$ chiral algebra (see, e.g., \cite{Pope:1991ig} and references therein for the latter) \footnote{See also \cite{Mago:2021wje} for a discussion of the quantum deformations of the algebra in the presence of non-minimal couplings, and how the deformations differ from $W_{\infty}$.}.

We will identify the mode algebra of this Kac-Moody algebra with the gauge symmetries of Poisson BF theory on twistor space.

The mode algebra of this Kac-Moody algebra is the loop algebra of $\op{Ham}(\C^2)$, which is
\begin{equation} 
	\op{Ham}(\C^2) [z,z^{-1}] = \C[v_1,v_2,z,z^{-1}] \label{eqn_ham} 
\end{equation}
under the Poisson bracket using the Poisson tensor $\partial_{v_1} \partial_{v_2}$.  A basis of this Lie algebra is given by the expressions
\begin{equation} 
	w[m,n,k] = v_1^m v_2^n z^{k},  
\end{equation}
and these have the commutation relations
\begin{equation} 
	[w[m,n,k], w[r,s,l] ] = (ms-nr)w[m+n-1,r+s-1, k+l].  
\end{equation}

Let us compare to this to what we find from Poisson BF theory. There, we are interested in the gauge transformations on a patch $\C^\times \times \C^2$ which preserve the vacuum field configuration $\mc{H} = 0$, $\beta = 0$. 

The infinitesimal gauge transformations of $\mc{H}$ are precisely the Lie algebra \eqref{eqn_ham}. 

There are also the gauge transformations associated to the field $\beta$, which form the vector space $\C[v_1,v_2,z,z^{-1}]$.   These commute with each other and live in the adjoint representation of $\op{Ham}(\C^2)[z,z^{-1}]$.

In sum, the Lie algebra of gauge transformations consists of the mode algebra of the chiral algebra found in \cite{Strominger:2021lvk}, together with a copy of its adjoint representation.  As in the gauge theory case, additional elements in our Lie algebra correspond to states of the opposite helicity.

We can build a chiral algebra whose mode algebra is this enlarged (by the copy of the adjoint) Lie algebra, just as before. It consists of conformal primaries
\begin{equation} 
	w[r,s] , \til{w}[r,s] 
\end{equation}
where $w[r,s]$ has spin $2-r/2-s/2$, $\til{w}[r,s]$ has spin $-2-r/2-s/2$, and they satisfy the OPEs
\begin{equation} 
	\begin{split}
		w[m,n](0) w[r,s](z) &\sim \frac{1}{z} (ms-nr) w[m+r-1,n+s-1]  (0) \\
		w[m,n](0)\til{w}[r,s](z)&\sim \frac{1}{z}(ms-nr)\til{w}[r,s](0) \\
		\til{w}(0) \til{w}[r,s](z) &\sim 0. 
	\end{split}
\end{equation}
We will explain how to derive these OPEs directly from twistor space in section \ref{s:Koszul}.

\subsection{Including the axion}\label{s:axion}
As we have seen, to cancel the anomaly on twistor space we need to introduce an axion field.  Here we show how to enlarge the gauge theory symmetry algebra by including the axion contribution.  

On twistor space, the gauge symmetry for the axion is given by a closed holomorphic two-form. Working as above on a patch of the form $\C^2 \times \C^\times$, with coordinates $v_i, z$, we can write a basis for the space of closed two-forms as
\begin{equation}
	\begin{split} 
		E[r,s,k]   	 &= \frac{1}{r+s} z^k \d z \d ( v_1^r v_2^s) \\
		F[r,s,k]  &= \d \left( z^k \frac{1}{r+s+2}(v_1^{r+1} v_2^s \d v_2 - v_1^{r} v_2^{s+1} \d v_1 )   \right)\label{eqn:axionmodes} 
	\end{split}	
\end{equation}
As before, for fixed $r+s$, these transform in the representation of spin $(r+s)/2$ of $\mf{sl}_2$.  

We can read the Lie brackets between the gauge theory symmetries and the axion symmetries from the term $\int \op{tr}  ( \mc{A} \partial \mc{A}) \eta$.  To do this, we should interpret all fields in the BV formalism, so that $\mc{A} \in \Omega^{0,\ast}(\PT,\g)$. The generator of gauge symmetry for $\mc{A}$ is the component in $\Omega^{0,0}$, and the anti-field to the generator for the gauge symmetry of $\mc{B}$ is the component in $\Omega^{0,3}$. Using standard BV machinery, we determine that there are terms in the Lie bracket whereby:
\begin{enumerate} 
	\item The commutator of a $J$ with a $J$ becomes the closed two-form symmetry, by
		\begin{equation} 
			[J , J] = (\partial J \wedge \partial J)  
		\end{equation}
		Here the colour indices are contracted with the Killing form, and we are viewing $J$ as a holomorphic function on $\C^2 \times \C^\times$.  The right hand side is a closed two-form on this space.
		
		The fact that there are two copies of $\partial$ on the right hand side is because the BV anti-bracket for the $\eta$ field involves a $\partial$, see for instance \cite{Costello:2019jsy}.
	\item The commutator of a $2$-form with a $J$ yields a $\til{J}$, by
	\begin{equation}
		\begin{split} 	
		[E, J ] &= E \wedge \partial J \\
		[F, J] &= F \wedge \partial J 
	\end{split}
	\end{equation}
		where the right hand side is interpreted as in the $\til{J}$ part of the Lie algebra, which consists of $3$-form on $\C^2 \times \C^\times$. 
\end{enumerate}
In all of these expressions, we will find a factor of
\begin{equation} 
	\what{\lambda}_{\g} = \frac{\lambda_{\g}}{ (2 \pi \i)^{3/2} \sqrt{12} } . 
\end{equation}
We can if we like absorb this factor into a redefinition of  $E,F$ and $\til{J}$.  In this section we will keep it. 

Let us now write the brackets out in components.

We find that
\begin{equation}
	\begin{split} 
		\frac{1}{\what{\lambda}_{\g}} 	[J^a[r,s,k], J^b[t,u,l] ] =&   K^{ab} \d ( v_1^r v_2^s z^k ) \d ( v_1^t v_2^u z^l )  \\
		=& K^{ab}   z^{k+l-1} \d z (tk - rl) v_1^{r+t -1}  v_2^{s+u} \d v_1 \\ 
		&+   K^{ab} z^{k+l-1} \d z (uk - sl)  v_1^{r+t} v_2^{s+u-1} \d v_2 \\
		&+  K^{ab} z^{k+l} (ru-st) v_1^{r+t-1} v_2^{u+s-1} \d v_1 \d v_2   \\ 
		= &K^{ab}	(k(u+ t) -l(r+s) )  E[r+t, s+u, k+l-1] \\
		&+ K^{ab} (ru - st) F[r+t-1,s+u-1,k+l] 
	\end{split}
\end{equation}

Similarly, the commutator between $J$ and $E$ is given by
\begin{equation} 
	\begin{split} 
		\frac{1}{\what{\lambda}_{\g}} 		[J^a[r,s,k], E[t,u,l] ] &= \frac{1}{t+u} z^l \d z \d ( v_1^t v_2^u) \d (v_1^r v_2^s z^k ) \\ 
		&= \frac{1}{t+u}z^{k+l} \d z (t s - u r) v_1^{t+r-1} v_2^{s+u-1} \d v_1 \d v_2 \\
		&= \frac{1}{t + u}(ts - ur) \til{J}^a[t+r-1,s+u-1,k+l] . 
	\end{split}
\end{equation}
Finally, the commutator between $J$ and $F$ is given by 
\begin{equation} 
	\begin{split} 
		\frac{1}{\what{\lambda}_{\g}} 		[J^a[r,s,k], F[t,u,l] ] &=   \d \left( z^l \frac{1}{t+u+2}(v_1^{t} v_2^u \eps^{ij} v_i \d v_j)    \right)  \d (v_1^r v_2^s z^k ) \\
			&= z^{k+l-1} v_1^{r+t } v_2^{s+u} \d z \d v_1 \d v_2 \frac{ k (t + u+2) - l (r+s) }{t+u+2} \\
		&= \til{J}^a[r+t, u+s, k+l - 1]  \frac{ k (t + u +2) - l (r+s) }{t+u+2}.  
	\end{split}
\end{equation}

From these commutators one can read off the OPEs.  Our conventions are that the modes $J^a[r,s,k]$ are obtained as the modes $\frac{1}{2 \pi \i} \oint z^k \d z J^a[r,s]$ of the state $J^a[r,s]$ in the chiral algebra, and similarly for the towers $\til{J}$, $E$, $F$.  From this we find that the term in the $JJ$ OPE  which include the axion fields is the following:
\begin{equation} 
	\begin{split} 
		\frac{1}{\what{\lambda}_{\g}} 		J^a[r,s](0) J^b[t,u](z) =& \frac{1}{z} K^{ab} (ru-st) F[r+t-1,s+u-1] (0)\\
		&- \frac{1}{z} K^{ab} (t+u)  \partial_z E[r+t,s+u](0) \\
		&- \frac{1}{z^2} K^{ab} (r+s+t+u) E[r+t,s+u](0).  
	\end{split}
\end{equation}
To check that this gives the correct commutator, we compute the commutator
\begin{equation} 
	\left[ \frac{1}{2 \pi \i}\oint_{z_1}\d z_1 z_1^{k} J^a[r,s] , \frac{1}{2 \pi \i} \oint_{z_2} z_2^l \d z_2 J^b[t,u] \right]. 
\end{equation}
This is computed by moving one contour past each other, as usual, leaving us with
\begin{equation} 
	\frac{1}{2 \pi \i} \oint_{z_1}  \d z_1 \oint_{\abs{w} = \eps} \d w z_1^k (z_1 + w)^l J^a[r,s] (z_1) J^b[r,s](z_1 + w).  
\end{equation}
Expanding the right hand side using the OPE and performing the contour integral over $w$ gives the desired commutator. 

Similarly, the $JE \to \til{J}$ and $JF \to \til{J}$ OPEs are given by
\begin{equation} 
	\begin{split} 
		\frac{1}{\what{\lambda}_{\g}} 		J^a[r,s](0) E[t,u](z) &= \frac{1}{z} \frac{(ts - ur)}{t + u} \til{J}^a [t+r - 1, s + u -1](0)  \\
		\frac{1}{\what{\lambda}_{\g}} 	J^a[r,s](0) F[t,u](z) &=  - \frac{1}{z} \partial_z \til{J}^a[r+t, s + u](0)  - \frac{1}{z^2} (1 + \frac{r + s}{t+u+2}) \til{J}^a[r+t, s+u](0)  
	\end{split}
\end{equation}

We should emphasize that that there is a tree level anomaly in the field theory on twistor space coming from axion exchange, which cancels a one loop anomaly in the gauge theory via a Green-Schwarz mechanism.  This tells us that we should not expect a fully consistent chiral algebra where we include the axion but do not quantum-correct the OPEs from the gauge theory sector.  The quantum chiral algebra is not fully understood.  However, we describe some corrections to the $JJ$ and $J \til{J}$ OPEs in section \ref{s:quantum_chiral}.  

\section{States from the point of view of twistor space}\label{s:states}
We complete our survey of basic aspects of the celestial/twistor correspondence by describing how certain states on twistor space give rise bijectively to the states of negative conformal dimension in 4d which generate conformal primary states of the form studied in \cite{Pasterski:2016qvg}. We emphasize that our states capture the conformally soft modes of mostly-negative integral conformal weights, rather than the principal series of \cite{Pasterski:2016qvg}\footnote{See \cite{Donnay:2020guq} for a discussion of the relationships between these two bases.}. We will illustrate this in the simplest example, though the relationship extends to the more general massless 4d/holomorphic twistor theories we are interested in.
\subsection{Preliminaries} 
Recall that we use holomorphic coordinates $v_i,z$ on twistor space. These are related to coordinates $x_i$ on $\R^4$ by
\begin{equation} 
	\begin{split}
		v_1 &= x_1 + \i x_2 + z(x_3 - \i x_4)\\
		v_2 &= x_3 + \i x_4 - z(x_1 - \i x_2).
	\end{split}
\end{equation}
If we fix $z$, then $v_i(z)$ are holomorphic functions on $\R^4$ in the complex structure determined by $z$. They are also null and orthogonal, $\ip{v_i(z), v_j(z)} = 0$.  

In the celestial holography literature, an important role is played by a parameterization of the space of complexified null momenta by a triple of complex numbers $(\omega, z,\zbar)$.  As in \cite{Pate:2019lpp, Guevara:2021abz} we will treat $z,\zbar$ as independent, i.e.\ not require that $z$ is the complex conjugate of $\zbar$.  When we want to use real momenta in various signatures, we need to impose reality conditions on $\omega,z,\zbar$.  

We identify momenta (living in the dual $\R^4$) with vectors in $\R^4$ using the inner product. Then, the formula for the null momentum $p(\omega,z,\zbar)$ is 
\begin{equation}
	\begin{split} 
		 	p(\omega,z,\zbar) &= \eps \omega (-\i - \i z \zbar, 1 - z \zbar,  -z - \zbar, -\i (z-\zbar) ) \\
		&=  \omega\eps \i \left(- v_1(z) + \vbar_2(z) \zbar  \right)
		\end{split}
\end{equation}
where $\eps = \pm 1$ indicates whether the particle is incoming or outgoing.

From this, we see that the parameter $z$ used in the celestial holography literature is the same as the coordinate $z$ on the twistor $\CP^1$.  From the point of view of twistor space, the coordinate $\zbar$ has a quite different interpretation.  The fibre over $z \in \CP^1$ in $\PT$ is a copy of $\C^2$, and $\zbar$ is coordinate on the projectivization of this $\mathbb{C}^2$, i.e. a coordinate on a different $\mathbb{CP}^1$. That is, $\zbar$ can be understood as $v_2/v_1$.  

Thus, from the twistor perspective, treating the parameters called $z$ and $\zbar$ in the celestial holography literature as complex conjugate is very unnatural. To avoid confusion, in this section in what follows we will tend not to use the symbol $\zbar$, and instead use $\lambda$:
\begin{equation} 
	\lambda := \zbar. 
\end{equation}

In celestial holography, one considers states which are conformal primaries under the $SL_2$ rotating the celestial sphere. These are obtained as a Mellin transform of the usual basis of states. For a scalar field theory \cite{Pasterski:2016qvg}, a state of dimension $\Delta$ is an expression of the form
\begin{equation}
	\ip{p(z,\lambda),x}^{-\Delta} = \left( - v_1(z) + \vbar_2(z) \lambda \right)^{- \Delta}. 
\end{equation}

\subsection{States on twistor space}

We would like to explain how the basis of states described in \cite{Pasterski:2016qvg} appears in a natural way by writing states in terms of twistor space.  To explain this point, we will work with the very simplest example of the twistor correspondence: a free scalar field theory $\phi$ on $\R^4$.

As reviewed in \S \ref{s:twistor}, this corresponds on twistor space to a free holomorphic Chern-Simons theory. The fundamental field is
\begin{equation} 
	\mc{A} \in \Omega^{0,1}(\PT, \Oo(-2)) 
\end{equation}
i.e. an Abelian holomorphic Chern-Simons gauge field twisted by $\Oo(-2)$.  The Lagrangian is
\begin{equation} 
	\int \mc{A} \dbar \mc{A} 
\end{equation}
(which makes sense, because $\mc{A} \dbar \mc{A}$ is a $(0,3)$ form with values in $\Oo(-4)$, and $\Oo(-4)$ is the canonical bundle of twistor space).  $\mc{A}$ is subject to gauge transformations $\mc{A} \mapsto \mc{A} + \dbar \chi$, where $\chi$ is a section of $\Oo(-2)$.

The space of solutions to the equations of motion of this theory, modulo gauge, are 
\begin{equation} 
	H^1_{\dbar}(\PT, \Oo(-2)). 
\end{equation}
As we reviewed, this is the space of harmonic functions $\phi$ on $\R^4$  which are entire analytic functions of $x_1,\dots,x_4$. 

A single-particle on-shell state is a solution to the equations of motion, and so can be represented by a $\dbar$ closed $(0,1)$ form on $\PT$, twisted by $\Oo(-2)$. We can build states localized along the complex surface $z = z_0$ in $\PT$ by the expression
\begin{equation}\label{states} 
	 \d z\delta_{z = z_0} (v_1 + \lambda v_2)^{n}. 
\end{equation}
(We include $\d z$ to emphasize that this is a section of $\Oo(-2)$).

When we translate this into states of a free scalar field, by integrating over the $z_0$ plane, we simply get
\begin{equation}
	(v_1(z_0) + \lambda v_2(z_0))^{n}. 
\end{equation}
These states are \emph{precisely} the standard basis of conformal primary states!

The states on twistor space have a natural action of the Virasoro algebra, coming from holomorphic vector fields on the $z$ plane.  Under these transformations, $v_i$ transform as $(\d z)^{1/2}$. Explicitly,
\begin{equation}
L_k = -z^{k+1} \partial_z - \frac{k + 1}{2}z^k v_i \partial_{v_i}.
\end{equation}
The states on twistor space \eqref{states} are conformal primaries: after shifting $z$ to $z - z_0$ they are annihilated by $L_k$ for $k > 0$, and are of conformal dimension $-n/2$.

In sum:
\begin{proposition}
	There is a bijection between
	\begin{enumerate} 
		\item Conformal primary states of conformal dimension $n \le 0$, at $z=z_0$.
		\item Conformal primary states on twistor space of conformal dimension $n/2$.
	\end{enumerate}
\end{proposition}

\subsection{States of positive conformal dimension}
We can also elucidate the twistor space origin for states of positive conformal dimension. 
Let us start by considering the states of conformal dimension $1$. These are of the form 
\begin{equation}
\frac{1}{v_1(z_0) + \lambda v_2(z_0) + \eps }
\end{equation}
where $\eps$ is a regulating parameter, introduced to ensure that this expression satisfies the Klein-Gordon equation without a source term.

On twistor space, these states are represented by 
\begin{equation}
\frac{1}{v_1 + \lambda v_2  + \eps } \d z \delta_{z=z_0}.
\end{equation}
This expression does not satisfy the equations of motion, but it instead satisfies
\begin{equation}
\dbar \frac{1}{v_1 + \lambda v_2 + \eps } \d z \delta_{z=z_0}
= \d z \delta_{z=z_0} \delta_{v_1 + \lambda v_2 = -\eps}.
\end{equation}
This field is therefore the field sourced by the operator
\begin{equation}
\int_{z=z_0, v_1 = -\lambda v_2 - \eps} \mc{A} ,
\end{equation}
which is of course an Abelian holomorphic Wilson line.

One can ask why the field sourced by such a surface defect can be thought of as a state. This is a rather delicate question, and depends on the signature.   To understand this point, we need to understand how the Penrose correspondence relating cohomology groups and harmonic functions  depends on the signature. 

The standard formulation of the result is in terms of harmonic functions that extend to all of $\C^4$:
\begin{equation} 
	H^1(\PT, \Oo(-2))=\text{Harmonic functions on Euclidean space that analytically extend to } \C^4	 .
\end{equation}
In Euclidean signature, all harmonic functions extend analytically to all of $\C^4$. However, in other signatures, there are harmonic functions which do not extend to all of $\C^4$.  For example, if we use Euclidean coordinates $x_i$, so that $t = \i x_1$, then the expression
\begin{equation} 
	\frac{1}{x_1 + \i x_2 +  \eps}  
\end{equation}
for $\eps$ real, is a harmonic function that extends meromorphically to $\C^4$. The pole in this expression does not intersect the Lorentzian $\R^{3,1}$ where $x_1$ is imaginary and $x_2,x_3,x_4$ are real.   In Lorentzian signature, therefore, this expression satisfies the Klein-Gordan equation without a source term. We can therefore think of it as a state.  In Euclidean signature, by contrast, there is a source term. 

What this shows us is that a carefully chosen surface defect on twistor space can give rise to a state in Lorentzian (or $(2,2)$) signature but not in Euclidean signature. This will happen if the location of the defect in $\PT$ does not intersect any of the $\CP^1$'s in $\PT$ which correspond to points in Lorentzian $\R^{3,1} \subset \C^4$.

Returning to the example of the state sourced by a holomorphic Wilson line at $z = z_0$, $v_1 + \lambda v_2 + \eps = 0$, the singular locus of the state it sources is the subset of $\C^4$ defined by
\begin{equation} 
	x_1 + \i x_2 + z_0(x_3 - \i x_4) + \lambda (x_3 + \i x_4) - \lambda z_0 ( x_1 - \i x_2) + \eps = 0.  
\end{equation}
This will not intersect the Lorentzian locus where $x_1$ is imaginary and $x_2,x_3,x_4$ are real as long as $\lambda = -\zbar_0$ and $\eps$ is real and non-zero.  

In sum, states of positive conformal dimension correspond to surface defects on $\PT$, where the surface defect lives on a non-compact surface $\C \subset\PT$ whose parameters $(\lambda,z_0,\eps)$ satisfy these conditions.  

These states furnish modules for the chiral algebra generated by the states of negative conformal dimension, which correspond to the algebra generators. It would be interesting to have a more complete understand of chiral algebra modules in terms of defects on twistor space.

\section{Celestial chiral algebras as boundary algebras}\label{s:3d}

We have so far discussed celestial operators and states from a twistorial point of view, leveraging the Penrose-Ward correspondence. In particular, we reproduced certain chiral algebras on the celestial sphere, recently discovered from flat space scattering amplitudes in the conformal basis, in terms of gauge symmetries of holomorphic theories on twistor space. We would like to better understand some physical features of this 2d chiral algebra. 

In the standard $AdS/CFT$ correspondence, the chiral algebra of a two-dimensional CFT is part of the algebra of operators living at the boundary of $AdS_3$.  Here we will explain how to derive a similar picture for the celestial chiral algebras studied in \cite{Guevara:2021abz, Strominger:2021lvk}.

We will show that the celestial chiral algebras live at the boundary of a 3d theory on 
\begin{equation} 
	\R_{> 0} \times \CP^1. 
\end{equation}
This theory is built by compactifying the field theory on the open subset 
\begin{equation} 
	\PT \setminus \CP^1 
\end{equation}
of twistor space to three dimensions. Locally, this open subset is $(\C^2 \setminus 0) \times \C$, with coordinates $v_i, z$ where $v_1,v_2$ are not both zero. We can write this as
\begin{equation} 
	(\C^2 \setminus 0) \times \C = S^3 \times \R_{> 0} \times \C.
\end{equation}
We can then compactify to a theory on $\R \times \C$ along the unit three-sphere\footnote{  We could also compactify along the unit three-sphere on $\R^4$, which in these coordinates is given by $\abs{v_1}^2 + \abs{v_2}^2 = (1+\abs{z}^2)^{-1}$.  This does not change anything essential in our analysis.   } 
\begin{equation} 
	\abs{v_1}^2 + \abs{v_2}^2 = 1 
\end{equation} to obtain our 3d theory.

It is very important to note, however, that even when we compactify the twistor representation of self-dual gravity to three dimensions, we \emph{do not} find a gravitational theory on $\R_{> 0} \times \CP^1$. In the non-linear graviton construction, the allowed diffeomorphisms on twistor space preserve the coordinate $z$, and the fields in the theory do not include a Beltrami differential varying the complex structure in the $z$ direction.

Similarly, the celestial chiral algebra does not contain a stress tensor, which is a hallmark of having a gravitational theory in the bulk. From this perspective, therefore, we cannot say that the celestial chiral algebra fits into a standard holographic dictionary. One may view this bulk/boundary system as rather analogous to the Chern-Simons/WZW correspondence.

However, our three dimensional bulk theory is \emph{not} topological. Rather, it is topological in the radial direction $\R_{> 0}$ and holomorphic in the $\CP^1$ direction. Such theories were studied in \cite{Aganagic:2017tvx,Costello:2020jbh}, and arise from supersymmetric localization of 3d $\mathcal{N}=2$ theories. When the bulk theory is not topological, it implies that the boundary algebra does not have a stress tensor: see \cite{Costello:2020jbh} for details.

Because this theory is not topological, it is not correct to say that the celestial chiral algebra is part of a two-dimensional CFT.  It is the algebra of boundary operators of a  three-dimensional theory.   This is a much more general class of chiral algebras, and includes such degenerate chiral algebras as those with only non-singular OPEs.

\subsection{Description of the bulk theory}
Let us now describe the bulk $3$-dimensional theory, whose boundary algebra is given by the celestial chiral algebra. We will start with the theory for gauge theory, and then include gravity.  We will work on a patch in the $z$-plane, excluding $z = \infty$.  

We will start by describing the zero-modes of the theory obtained by reduction of holomorphic BF theory from $S^3 \times \R_{> 0} \times \C$ to $\R \times \C$, and then incorporate the KK modes. Because we are describing the zero modes we are only interested in the fields which are invariant under the $SU(2)$ rotating $S^3$.   We will prove the following result.
\begin{proposition}
	The three dimensional theory arising from the zero modes of holomorphic BF theory on twistor space is the holomorphic topological theory of \cite{Aganagic:2017tvx, Costello:2020jbh}, associated to the supersymmetric localization of $3d$ $\mathcal{N}=2$ theory with Chern-Simons level $0$ and adjoint-valued matter. 
\end{proposition}

To prove this, we will analyze field configurations on a patch of $\PT$ which are invariant under the $SU(2)$ under which the $v_i$ coordinates transform as a doublet.  The coordinates $r = \norm{v}$ and $z,\zbar$ are $SU(2)$ invariant.

For each Lie algebra index $a$, the $(0,1)$ form field $\mc{A}_a$ on $\PT$ has three components.  The most general $SU(2)$ invariant forms these three components can take are
\begin{equation}
	\begin{split}
		&A_a^{r} (z,\zbar,r) \dbar r\\ 
		&A_a^{\zbar}(z,\zbar,r) \d \zbar\\ 
		&\phi_a(z,\zbar,r) \frac{\eps^{ij} \vbar_i \d \vbar_j } { r^4} .
	\end{split}
\end{equation}

There is a single (adjoint valued) gauge parameter $\chi(z,\zbar,r)$ which is invariant under $SU(2)$.   Under this transformation,  $\phi_a^{S^3}$ transforms as a scalar in the adjoint representation, and $A_a^{r}$, $A_a^{\zbar}$ transform as a partial connection:
\begin{equation} 
	\begin{split}
		\delta A_a^{r} &= \partial_r \chi_a + f_a^{bc} \chi_b A^{r}_c\\
		\delta A_a^{\zbar}&= \partial_{\zbar} \chi_a + f_a^{bc} \chi_b A^{\zbar}_c \\
		\delta \phi_a &= f_a^{bc} \chi_b \phi_c. 
	\end{split}
\end{equation}
Under the $SO(2)$ rotating the $z$ plane, $A_a^{\zbar}$ transforms as a $(0,1)$ form and $A_a^{r}$ as a scalar.  

However, $\phi_a$ does not transform as a scalar, as one might at first expect. This is because  $ v_i$ transforms as $(\d z)^{1/2}$, so that $\vbar_i$ transforms as $(\d z)^{-1/2}$: we find that $\phi_a$ has spin $-1$.   

In a similar way, we can compute the zero modes of $\mc{B}$. We will identify the canonical bundle of twistor space with the bundle of quadratic differentials on the $z$-plane. Then, we find that the zero modes of the field $\mc{B}_a \in \Omega^{3,1}(\PT)$ are
\begin{equation} 
	\begin{split}
		&\eta_a^{r} (z,\zbar,r) (\d z)^2 \dbar r\\ 
		&\eta_a^{\zbar}(z,\zbar,r)(\d z)^2 \d \zbar\\ 
		&B_a(z,\zbar,r) \frac{\eps^{ij} (\d z)^2 \vbar_i \d \vbar_j } { r^4} 
	\end{split}
\end{equation}
From this, we see that $B_a$ transforms as an adjoint valued scalar of spin $1$, and that $\eta$ defines a partial connection in the bundle of quadratic differentials.

Under the gauge transformation with parameter $\chi$, the fields $\eta_a^{r}$, $\chi_a^{\zbar}$, $B_a$ are all adjoint valued scalars.    There is an additional gauge symmetry, where the gauge parameter $\psi$ is adjoint valued and of spin $2$, under which the fields transform as
\begin{equation} 
	\begin{split}
		\delta \eta_a^{r} &= \partial_r \psi_a\\ 
		\delta \eta_a^{\zbar}&= \partial_{\zbar} \psi_a\\
		\delta B_a &= 0.
	\end{split}
\end{equation}

To write the Lagrangian, we combine $A^{r}$, $A^{\zbar}$ into a partial connection
\begin{equation} 
	A = A^{r} \d r + A^{\zbar} \d \zbar,
\end{equation}
and combine $\eta^r$, $\eta^{\zbar}$ into
\begin{equation} 
\eta = \eta^r \d r +  \eta^{\zbar} \d \zbar. 
\end{equation}

The Lagrangian for holomorphic BF theory, when restricted to these zero modes, becomes 
\begin{equation} 
	\int \d z \op{Tr} \left(  B F(A) + \eta \d_A \phi \right).  
\end{equation}
This is precisely the field content, Lagrangian, and symmetries described in \cite{Aganagic:2017tvx,Costello:2020jbh} for the supersymmetric localization of $3d$ $\mathcal{N}=2$ gauge theory with adjoint matter.

\subsection{Including Kaluza-Klein modes}
So far, we have only described the zero modes. When we include all the KK modes we get an infinite tower of fields of the same type.   

Let us describe the result.
	Including all KK modes, the three dimensional topological-holomorphic theory associated to holomorphic BF theory on twistor space has fields 
	\begin{equation}
		\begin{split}
			A_a[k,l] &= A_a^r[k,l] \d r + A_a^{\zbar}[k,l]\d \zbar \\
			\eta_a[k,l] &= \eta_a^r[k,l] \d r + \eta_a^{\zbar} [k,l] \d \zbar \\
			J_a[k, l] &:= B_a[k,l]  \\ \til{J}_a[k, l]&:=\phi_a[k,l]  
		\end{split}
	\end{equation} (In anticipation of their identification with familiar chiral algebra generators, we have renamed the $B_a, \phi_a$ fields, though one should keep in mind their identifications with 3d fields).
	
	Here $k,l$ are integers $\ge 0$.  Under the $SU(2)$ which rotates $S^3$, these expressions live in a representation of spin $(k+l)/2$ and are weight vectors of weight $(k-l)/2$.  Under the $SO(2)$ rotating the coordinate $z$, $A_a[k,l]$ is of spin $(k+l)/2$, $J_a[k,l]$ is of spin $1-(k+l)/2$, $\eta_a[k,l]$ is of spin $2+(k+l)/2$ , $\til{J}_a[k,l]$ is of spin $-1-(k+l)/2$. 

Explicitly, we can expand certain six-dimensional field configurations in terms of these fields as 
\begin{equation} 
	\begin{split}
		\mc{A}_a(v_1,v_2,r,z)  =& \sum A_a[k,l] v_1^k v_2^l \\
		&+ \til{J}_a[k,l] \frac{1}{k!} \frac{1}{l!} \partial_{v_2}^k \partial_{v_1}^l  \frac{\eps^{ij} \vbar_i \d \vbar_j } { r^4} \\ 
		\mc{B}_a^{6d}(v_1,v_2,r,z) =& \sum \eta_a[k,l] v_1^k v_2^l \\
		&+ J_a[k,l] \frac{1}{k!} \frac{1}{l!} \partial_{v_2}^k \partial_{v_1}^l  \frac{\eps^{ij} \vbar_i \d \vbar_j } { r^4}. \\ 
	\end{split}
\end{equation}
There are other fields in six dimensions. However, they are massive and do not propagate in three dimensions and as such can be integrated out. 

A detailed calculation of the KK reduction of holomorphic theories along spheres was performed in \cite{gwilliam2018higher}; we refer to that work for more details. 

It is convenient to organize the KK modes into generating functions
\begin{equation} 
	\begin{split}
		A_a(v_1,v_2) &= \sum_{k,l \ge 0} v_1^k v_2^l A_a[k,l] \\
		\eta_a(v_1,v_2) &= \sum_{k,l \ge 0} v_1^k v_2^l \eta_a[k,l] \\
		J_a(v_1,v_2) &=  \sum_{k,l \ge 0} v_1^{-l}v_2^{-k} J_a[k,l] \\
		\til{J}_a(v_1,v_2) &=  \sum_{k,l \ge 0} v_1^{-l}v_2^{-k} \til{J}_a[k,l]. 
	\end{split}
\end{equation}
where $v_1,v_2$ are an $SU(2)$ doublet of spin $1/2$.  Note that $A_a(v_1,v_2)$ defines a partial connection on $\R_{> 0} \times \C$ for the infinite-dimensional Lie algebra $\g[v_1,v_2]$. We can view $J$ as living in the adjoint representation of $\g[v_1,v_2]$, and $\eta,\til{J}$ as living in the co-adjoint representation, using the residue pairing against $v_1^{-1} v_2^{-1} \d v_1 \d v_2$ to identify $\g[v_1^{-1},v_2^{-1}]$ with the dual of $\g[v_1,v_2]$. 

As before, there are two kinds of gauge transformations which we can also arrange into generating functions:
\begin{equation} 
	\begin{split}
		\chi_a(v_1,v_2) &= \sum v_1^k v_2^l \chi_a[k,l] \\
		\psi_a(v_1,v_2) &= \sum v_1^k v_2^l \psi_a[k,l].
	\end{split}
\end{equation}
The gauge transformations are
\begin{equation} 
	\begin{split}
		\delta_{\chi} A_a(v_1,v_2) &= \d r \partial_r \chi_a(v_1,v_2) + \d \zbar \partial_{\zbar} \chi_a(v_1,v_2) + f_a^{bc} \chi_b(v_1,v_2) A_c(v_1,v_2) \\
		\delta_{\chi} \til{J}_a(v_1,v_2) &= f_{a}^{bc} \chi_b(v_1,v_2) \til{J}_c(v_1,v_2) \\
		\delta_{\chi} \eta_a(v_1,v_2) &= f_{a}^{bc} \chi_b(v_1,v_2) \eta_c(v_1,v_2) \\
		\delta_{\chi} J_a(v_1,v_2) &= f_{a}^{bc} \chi_b(v_1,v_2) J_c(v_1,v_2) \\
		\delta_{\psi} \eta_a(v_1,v_2) &= \d r \partial_r \psi_a(v_1,v_2) + \d \zbar \partial_{\zbar} \psi_a(v_1,v_2).
	\end{split}
\end{equation}
In these expressions, when we multiply a polynomial in $v_1,v_2$ by a polynomial in $v_1^{-1}$, $v_2^{-1}$ we drop any terms which involve any positive powers of either $v_1$ or $v_2$. 

	The Lagrangian is
	\begin{equation} 
		\int_{r,z} \oint_{v_1,v_2} v_1^{-1} v_2^{-1} \d v_1 \d v_2 \d z \left[ \op{Tr} ( J(v_1,v_2) F(A(v_1,v_2)) ) + \op{Tr} ( \til{J}(v_1,v_2) \d_{A(v_1,v_2)} \eta(v_1,v_2) )    \right] .	
	\end{equation}
We can expand this out terms of the components. We will write the terms components involving $J$ and $A$ explicitly:
\begin{equation} 
	\int_{r,z} \d z \left(\delta_{k=s, l = r} J_a[k,l] \d A_a[r,s] + f^{abc} \delta_{k = n+s, l = m+r}  J_a[k,l] A_b[m,n] A_c[r,s] \right) 
\end{equation}

This theory is simply the holomorphic-topological theory of \cite{Aganagic:2017tvx,Costello:2020jbh},  with gauge symmetry given by the infinite dimensional Lie algebra $\g[v_1,v_2]$ and matter in the co-adjoint representation $\g[v_1^{-1},v_2^{-1}]$.  

Of course, as with all theories with an infinite number of fields, loop-level computations can really only be done in the original six-dimensional context. 

It is important to point out (as K. Zeng has explained to us) that this Lagrangian acquires corrections coming from exchanges of the massive fields we have integrated out.  These corrections should be the bulk version of the quantum corrections to the chiral algebra presented in section \ref{s:quantum_chiral}.  

\subsection{KK reduction of the non-linear graviton construction}
Recall that Poisson BF theory on twistor space corresponds to self-dual gravity on $\R^4$.  Poisson BF theory is obtained by adding an interaction term to Abelian holomorphic Chern-Simons theory, where the gauge field is twisted by $\Oo(2)$.  Therefore the field content of the theory on $\R_{>0} \times \C$ is the same as that obtained from Abelian holomorphic BF theory, except that some of the spins have been changed.

The fields are the following.  From the field $\mc{H}$ on twistor space we get the fields 
\begin{equation} 	
	H^t[k,l] \ \ H^{\zbar}[k,l] \ \ \til{w}[k,l] 
\end{equation}
The fields $H^t[k,l]$ and $H^{\zbar}[k,l]$ combine into partial connection on $\R \times \C$ of spin $(k+l)/2 -1$ under rotation of the $z$ plane. The field $\til{w}[k,l]$ is of spin $-2-(k+l)/2$.  

From the field $\beta$ on twistor space, we get towers of fields
\begin{equation} 
	\beta^t[k,l] \ \ \beta^{\zbar}[k,l] \ \ w[k,l] 
\end{equation}
As before, $\beta^t[k,l]$ and $\beta^{\zbar}[k,l]$ combine into a partial connection of spin $(k+l)/2 + 3$.  The field $w[k,l]$ is of spin $2-(k+l)/2$.  

We can combine all these fields into generating functions, as before. We let $\op{Ham}(\C^2)$ be the Lie algebra of Hamiltonian vector fields on $\C^2$, under the Poisson bracket.  Then, we can view 
\begin{equation} 
	H = \sum v_1^k v_2^l H[k,l]  
\end{equation}
as a partial connection on $\R \times \C$ for the Lie algebra $\op{Ham}(\C^2)$.  Note that $\op{Ham}(\C^2)$ is a graded Lie algebra, where the function $v_1^k v_2^l$ has charge $(k+l)/2-1$.  The spin of the components of the connection are given by this grading.

Similarly, we can view the fields $\beta[k,l]$ as a partial connection with values in $\op{Ham}(\C^2)$, where the spin is shifted by $4$ from that coming from the grading on $\op{Ham}(\C^2)$. 

Let $\op{Ham}(\C^2)^\vee$ be the linear dual of $\op{Ham}(\C^2)$. Using the residue pairing we can write an element of $\op{Ham}(\C^2)^\vee$ as a series in $v_1^{-1}$, $v_2^{-1}$:
\begin{equation} 
	\op{Ham}(\C^2)^\vee = v_1^{-1} v_2^{-1} \d v_1 \d v_2 \C[v_1^{-1}, v_2^{-1}].  
\end{equation}
In this way, we can arrange the fields $\til{w}[k,l]$ into a field valued in $\op{Ham}(\C^2)^\vee$, where the spin is given by the natural grading on $\op{Ham}(\C^2)^\vee$, shifted by $-3$. The fields $w[k,l]$ arrange into a field valued in $\op{Ham}(\C^2)^\vee$, where the spin is shifted by $1$ from that induced from the grading on $\op{Ham}(\C^2)^\vee$.

Putting all this together, we find that the Lagrangian is
\begin{equation} 
	\int_{R_{> 0} \times \C} \d z \ip{w, \d H + \tfrac{1}{2}\eps_{ij} \partial_{v_i} H \partial_{v_j}H } + \int_{\R_{> 0} \times \C} \d z \ip{\til{w}, \d \beta} + \ip{\til{w}, \eps_{ij} \partial_i H\partial_j \beta}. 
\end{equation}
There are, as before, two kinds of gauge transformations, with parameters $\chi$ and $\psi$, and gauge transformations 
\begin{equation} 
	\begin{split}
		\delta_{\chi} H(v_1,v_2) &= \d r \partial_r \chi(v_1,v_2) + \d \zbar \partial_{\zbar} \chi(v_1,v_2) + \eps_{ij}\partial_{v_i} \chi(v_1,v_2)\partial_{v_j} H(v_1,v_2) \\
		\delta_{\chi} \til{w}(v_1,v_2) &= \eps_{ij} \partial_{v_i} \chi(v_1,v_2) \partial_{v_j} \til{w}(v_1,v_2) \\
		\delta_{\chi} \beta(v_1,v_2) &=\eps_{ij} \partial_{v_i} \chi(v_1,v_2) \partial_{v_j} \beta(v_1,v_2) \\
		\delta_{\chi} w(v_1,v_2) &= \eps_{ij} \partial_{v_i} \chi(v_1,v_2) \partial_{v_j} w(v_1,v_2) \\
		\delta_{\psi} \eta(v_1,v_2) &= \d r \partial_r \psi_a(v_1,v_2) + \d \zbar \partial_{\zbar} \psi(v_1,v_2).
	\end{split}
\end{equation}
This Lagrangian, and these gauge transformations, are those for the  partially-topological theories of \cite{Aganagic:2017tvx, Costello:2020jbh} where the gauge group is the infinite dimensional Lie algebra $\op{Ham}(\C^2)$ and the matter lives in the co-adjoint representation $\op{Ham}(\C^2)^\vee$.

It can be convenient to write the Lagrangian in components:
\begin{align*}
		 	&\int_{\R_{> 0} \times \C} \d z w[k,l]\d H[l,k]  + \int_{\R_{> 0} \times \C}\d z {1 \over 2} ( ms-nr  )  w[n+s-1,m+r-1]  H[m,n] H[r,s]   \\
		&+ \int_{\R_{> 0} \times \C} \d z \til{w}[k,l] \d \beta[l,k] + \int_{\R_{> 0} \times \C}\d z (ms-nr) \til{w}[n+s-1,m+r-1] H[m,n]\beta[r,s]. 
\end{align*}

\subsection{The boundary algebra for KK modes }
Boundary algebras for holomorphic-topological theories were studied in \cite{Costello:2020jbh}. Here we will review these methods, and apply them to the partially topological theories described above. We will find that it is an enlargement of the chiral algebra studied in celestial holography \cite{Guevara:2021abz}. 

Since this is a first-order Lagrangian, the boundary conditions involve setting half the fields to zero.   In the $3d$ theory coming from holomorphic BF theory on twistor space, we will ask that 
\begin{equation} 
\begin{split}
	\eta_a[k,l] &= 0 \\
	A_a[k,l] &= 0
\end{split}
\end{equation}
at the boundary $r = \infty$.  

The remaining boundary operators, as studied in \cite{Costello:2020jbh}, are functions of the fields $\til{J}_a[k,l]$ and $J_a[k,l]$.  The generators of the chiral algebras are thus elements of the vector space $\g[v_1,v_2]  \oplus \g[v_1,v_2]$.  

In \cite{Costello:2020jbh,gwilliam2019one} the boundary OPEs were computed, at tree level.  We will reproduce the calculation here.  We will change coordinates and let $s = 1/r$, so that the boundary is at $s=0$.  The bulk-boundary propagator for field theories of this type is discussed in equations (5.42), (5.7)  of \cite{Costello:2020jbh}\footnote{Actually, there is a small typo in the expression for the propagator which we correct}.   From these, we find that the field sourced by the operator $J_a[k,l]$ placed at $s = 0$, $z = 0$ is  
\begin{equation} 
	A_a[l,k] = \frac{3}{16 \pi \i}	\frac{  \zbar \d s - \tfrac{1}{2} s \d \zbar   }{ (\norm{z}^2 + s^2)^{3/2}    }. 
\end{equation}
All other field components are zero. Thus, the OPE can be computed by the Feynman diagram with one external line labeled by the $J$-field:
\begin{equation}
	\begin{split}
		J_b[k,l](0,0) J_c[m,n](z,0) &= f^{a}_{bc} \int_{z',s}  J_a[m+k,n+l](z',s)\d z'   \frac{  (\zbar' \d s - \frac{1}{2} s \d \zbar') ((\zbar'-\zbar) \d s - \frac{1}{2} s \d \zbar' )     }{ (\norm{z'}^2 + s^2)^{3/2}   (\norm{z-z'}^2 + s^2)^{3/2}    }  \\
		&= f^{a}_{bc} \int_{z',s}  J_a[m+k,n+l](z',s) \d z'  \frac{  \frac{1}{2} \zbar s \d s  \d \zbar'     }{ (\norm{z'}^2 + s^2)^{3/2}   (\norm{z-z'}^2 + s^2)^{3/2}    }\\
		&= \frac{1}{2} f^{a}_{bc} \zbar  \int_{z',s}  J_a[m+k,n+l](z',s)  \frac{  \d z \d s  \d \zbar'     }{ (\norm{z'}^2 + s^2)^{3/2}   (\norm{z-z'}^2 + s^2)^{3/2}    }\\
	\end{split}
\end{equation}
(We have absorbed factors of $\pi$ into normalization of the operators).  To perform the integral on the last line, we Taylor expand $J[m+k,n+l](z',s)$ as a series in $z',\zbar', s$.  It is easy to check for all terms except the constant term, the integral has no singularities as $z \to 0$.  Therefore, the OPE is computed by replacing $J_a[m+k,n+l](z',s)$ by $J_a[m+k,n+l](0,0)$, giving  
\begin{equation} 
	\frac{1}{2} f^{a}_{bc} \zbar  J_a[m+k,n+l](0,0) \int_{z',s} \frac{  \d s  \d \zbar'     }{ (\norm{z'}^2 + s^2)^{3/2}   (\norm{z-z'}^2 + s^2)^{3/2}    }
\end{equation}
It remains to compute the integral
\begin{equation} 
	\int_{z',s} \frac{ s  \d s \d z \d \zbar'     }{ (\norm{z'}^2 + s^2)^{3/2}   (\norm{z-z'}^2 + s^2)^{3/2}    }  
\end{equation}
which one can see by dimensional reasons is a non-zero multiple of  $\frac{1}{\norm{z}^2}$. The precise constant is unimportant, as it can be absorbed into a redefinition of the operators $J_a[m,n]$. We have thus found that
\begin{equation} 
	J_b[k,l] (0) J_c[m,n](z) \sim \frac{1}{z} f_{bc}^a J_a[k+m,l+n]. 
\end{equation}
This matches the OPE for celestial symmetries of gauge theory computed in \cite{Guevara:2021abz}.

In a similar way, we get the OPE
\begin{equation} 
	\til{J}_b[k,l](0) J_c[m,n](z) \sim \frac{1}{z} f_{bc}^a\til{J}_a[k+m, n+l] (0). 
\end{equation}

\subsection{Gravitational theory}
For the gravitational theory, the computation is almost identical, except that the derivatives with respect to $v_i$ change the result slightly.  Recall that (in components) the interaction term is
\begin{equation} 
	\int_{\R_{> 0} \times \C}\d z{1 \over 2} ( ms-nr  )  w[n+s-1,m+r-1]  H[m,n] H[r,s]      
\end{equation}
This leads to the OPE
\begin{equation} 
	w[m,n](0)w[r,s](z) \simeq\frac{1}{z} (ms-nr)w[m+r-1,s+n-1](0). 
\end{equation}
This is the Kac-Moody algebra for the Lie algebra of Hamiltonian vector fields on the plane, as found in \cite{Strominger:2021lvk}.

Similarly, the interaction
\begin{equation} 
	\int_{\R_{> 0} \times \C} \d z (ms-nr)   \til{w}[n+s-1,m+r-1]  H[m,n] \beta[r,s]      
\end{equation}
leads to the OPE
\begin{equation} 
	\til{w}[m,n](0) w[r,s](z) \simeq\frac{1}{z}(ms-nr) \til{w}[m+r-1,n+s-1]. 	 
\end{equation}

\section{The chiral algebra by Koszul duality}\label{s:Koszul}

We have shown how gauge symmetries of theories on twistor space lead naturally to the celestial chiral algebras of \cite{Guevara:2021abz, Strominger:2021lvk}. In this section, we will explain how to reproduce these chiral algebras using an alternative and calculationally efficient method. 

The underlying idea is very simple. Consider a situation where one is trying to couple two systems together along a common lower-dimensional submanifold in spacetime. One can determine constraints on the space of possible couplings imposed by the requirement that the coupling is gauge invariant. At the classical level, gauge transformations will take the form of Hamiltonian symplectomorphisms, so it is perhaps unsurprising that there are many ways to recover algebras such as the loop algebra of $w_{1 + \infty}$. In holographic contexts, one starts with a coupled ``bulk/defect'' system of closed strings or supergravity, and branes. The program of twisted holography leverages this setup, as well as the cohomological properties of twisted theories, to deduce results about the AdS/CFT dual pairs that result from top-down string theory models \cite{Costello:2018zrm, Costello:2020jbh}. Very loosely speaking, we think of the bulk/defect system in a holographic setup as a bulk/boundary dual pair \footnote{In examples of the AdS/CFT correspondence, one must incorporate into this setup backreaction of the D-branes on the closed string modes, done in the twisted context in \cite{Costello:2018zrm, Costello:2020jbh}, but to recover the celestial algebras pertinent to flat space holography, we can avoid this complication.}. Ultimately, we conjecture that twisted holography in twistor space provides the origin for the aspects of the celestial holography program governed by universal, asymptotic symmetries. Following \cite{Costello:2020jbh}, this point of view then has an immediate connection to Koszul duality. In this section, we will simply review the computation of \cite{Costello:2020jbh} that produces $\textrm{Ham}(\mathbb{C}^2)$ and remark on a twisted holography interpretation in \S \ref{s:conclusions}.

One proceeds in this approach order-by-order in perturbation theory (holographically, in a $1/N$-expansion) to compute the constraints imposed by gauge-invariance: one must evaluate the BRST variation of Feynman diagrams representing the bulk/defect coupling at a given order and demand that the total BRST variation vanish. For a given bulk theory, this imposes nontrivial constraints on a putative defect operator product. Quantum mechanical effects will in general deform the classic gauge algebras; see \cite{CWY} for the result when coupling 4d Chern-Simons theory to a topological line defect.  

As explained in \cite{Costello:2020jbh, Gaiotto:2019wcc}, the algebras resulting from these constraints can be encapsulated mathematically by the notion of Koszul duality. We refer to \cite{PW} for a recent exposition and more details on this point of view.

Let us now review the setup and computations of \cite{Costello:2020jbh}. The celestial symmetry algebra for gauge theory will come about when considering the coupling of 6d holomorphic Chern-Simons theory to a 2d holomorphic\footnote{A 2d defect on $\mathbb{C}$ is holomorphic if antiholomorphic translations $\partial_{\bar{z}}$ are trivial in the cohomology of the (twisted) BRST-differential.} defect, while the symmetry algebra for gravity will come about when coupling 6d Poisson BF theory to a 2d holomorphic defect. Because the computations are so similar, we will consider them in parallel. In both cases, we will work at tree-level to start. 

Consider holomorphic Chern-Simons theory (the worldvolume theory of Euclidean D-branes in the B-twist) with Lie algebra $\mathfrak{g}$ on $\mathbb{C} \times \mathbb{C}^2$.
The Lagrangian can be expressed in terms of a partial connection $A \in \Omega^{0, 1}(\mathbb{C}^3, \mathfrak{g})$, with equations of motion $F^{0, 2}(A) = 0$.

We consider a defect along a holomorphic plane $\mathbb{C}_z \subset \mathbb{C}_{z, v_1, v_2}^3$, which we will endow with holomorphic coordinate $z$. Our analysis is purely local, so this copy of $\mathbb{C}$ should be viewed as an open subset of the twistor fiber $\mathbb{CP}_z^1$ in $\mathbb{PT}$; for this local analysis, considering flat spacetimes suffices. Next, as argued in \cite{Costello:2020jbh}, one can consider the most general bulk/defect coupling 
\begin{equation}
	\sum_{k_1, k_2 \geq 0}\int_{\mathbb{C}_z}{1 \over k_1! k_2!} \partial_{v_1}^{k_1}\partial_{v_2}^{k_2}A^a_{\bar{z}} J_a[k_1, k_2] \label{eqn:general_coupling}
\end{equation} in terms of some general defect operators $J_a[k_1, k_2]$.

The tree-level chiral algebras come from requiring that the BRST variation of the Feynman diagrams in Figure  \ref{fig:cancel1} for gauge theory (we refer the reader to \cite{Costello:2020jbh} for the details).

\begin{figure}
	\begin{center}
	\begin{tikzpicture}
	\begin{scope}
		\node[circle, draw] (J1) at (0,1) {$J$};
		\node[circle, draw] (J2) at (0,-1) {$J$};
		\node (A1) at (2,1)  {$A_{\zbar}$};
		\node (A2) at (2,-1)  {$A_{\zbar}$};
		\draw[decorate, decoration={snake}] (J1) --(A1);
		\draw[decorate, decoration={snake}] (J2) --(A2);
		\draw (0,2) -- (J1) --(J2) -- (0,-2); 
	\end{scope}	
	\begin{scope}[shift={(4,0)}];
		\node[circle, draw] (J) at (0,0) {$J$};
		\node (A1) at (3,1)  {$A_{\zbar}$};
		\node (A2) at (3,-1)  {$A_{\zbar}$};
		\node[circle,draw,fill=black, minimum size = 0.2pt]  (V) at (1.5,0) {};  
		\draw[decorate, decoration={snake}] (J) -- (1.5,0) --  (A1);
		\draw[decorate, decoration={snake}] (1.5,0) --(A2);
		\draw (0,2) -- (J)-- (0,-2);
	\end{scope}
	\end{tikzpicture}
	\caption{Cancellation of the gauge anomaly of these two diagrams leads to the OPEs of the currents $J[k_1,k_2]$. \label{fig:cancel1}}
	\end{center}
\end{figure}
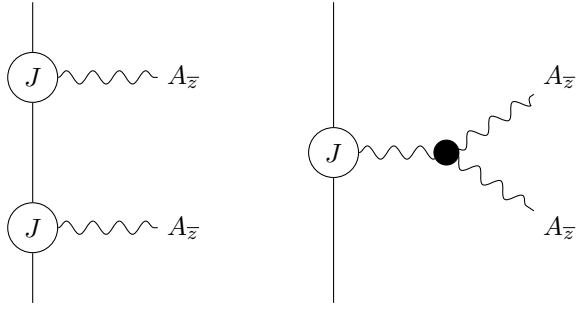
The resulting algebra is 
\begin{equation}
 J_b[r, s](0)J_c[k, l](z) \sim {1 \over z}f^a_{bc}J_a[r + k, s + l](z)\\
\end{equation} which is nothing but the algebra of holomorphic maps into $\mathfrak{g}$. 

One of the nice things about the Koszul duality point of view is that, even at tree-level, this perspective automatically places in the chiral algebra generators in the basis that makes manifest the large symmetry algebra they generate.

We have stated the result for holomorphic Chern-Simons theory. The same analysis for holomorphic BF theory gives us a second set of generators, $\til{J}[k,l]$, which couple the field $\mc{B}$.  The OPEs can be read from  the requirement of gauge invariance, as above.  We find the familiar OPEs
\begin{equation} 
	\begin{split} 
	J_b[r, s](0)J_c[k, l](z) &\sim {1 \over z}f^a_{bc}J_a[r + k, s + l](0)\\	
J_b[r, s](0)\til{J}_c[k, l](z) &\sim {1 \over z}f^a_{bc}\til{J}_a[r + k, s + l](0).\\
	\end{split}
\end{equation}

For Poisson BF theory, which is the twistor space uplift of self-dual gravity, we find something very similar.  As before, we describe the fields of Poisson BF theory as $\mc{H} \in \Omega^{0,1}(\PT, \Oo(2))$ and $\beta \in \Omega^{3,1}(\PT, \Oo(-2))$.  If we consider the universal one-dimensional holomorphic theory that can couple to Poisson BF theory, we have operators $w[r,s]$ that couple to the normal derivatives of $\mc{H}$, and $\til{w}[r,s]$ that couple to the normal derivatives of $\beta$.

The OPEs, which can again be read off from gauge invariance, are
\begin{equation} 
	\begin{split} 
		w[r, s](0) w[k, l](z) &\sim {1 \over z}(rl-ks) w[r + k-1, s + l-1](0)\\	
		w[r, s](0)\til{w}[k, l](z) &\sim {1 \over z}(rl-ks)\til{w}[r + k-1, s + l-1](0)
	\end{split}	
\end{equation}
These OPEs can be computed by a small variant of the computations of \cite{Costello:2020jbh}, section 7.3 (see in particular Theorem 7.3.1). 

\subsection{States and generators of the Koszul dual algebra}
For any holomorphic field theory on twistor space, there is a natural bijection
\begin{center}
	Single particle conformal primary states $\longleftrightarrow$ Conformal primary generators of the chiral algebra
\end{center}
We can see this by considering equation \eqref{eqn:general_coupling}, describing the universal chiral algebra we can couple along a $\CP^1$ in twistor space.  We get a state in the vacuum module of this chiral algebra by studying the chiral algebra in the presence of an on-shell background field which is localized at $z = 0$.     Such a background field is a single-particle state in the celestial CFT, as we saw in \S \ref{s:states}.  If this state is a conformal primary, then so is the state in the chiral algebra.   

In this way, we find there is a natural map from single-particle conformal primary states to the generators of the chiral algebra. It is easy to see that this is a bijection. Perhaps this is best illustrated with the example of holomorphic Chern-Simons theory.    There, conformal primary states are expressions like
\begin{equation} 
	\mc{A} = \delta_{z =0} v_1^{k_1} v_{2}^{k_2} \t_a. 
\end{equation}
These get sent to  the generators $J_a[k_1,k_2]$ of the algebra.

These twistor representatives were also recently employed in \cite{Adamo:2021zpw, Adamo:2019ipt} to obtain the celestial OPEs from a worldsheet computation in the ambitwistor string. It would be desirable to understand the relationship between these two perspectives in four dimensions. 

\subsection{The axion from Koszul duality}
We have already computed the contribution of the axion fields to the chiral algebra by considering gauge transformations.  Here we will redo the computation at tree level by considering Koszul duality.  As we have seen, the axion field contributes two towers $E[r,s]$ and $F[r,s]$.  Here we will derive these operators and their OPEs from the point of view of Koszul duality.

The derivation in this case is a little more complicated than that for gauge theory because the fundamental field is a \emph{constrained} field: it is a $(2,1)$-form $\eta$ constrained to satisfy $\partial \eta = 0$.  Perhaps the simplest way to proceed is to work instead with a $(1,1)$ form $\alpha$ with $\partial \alpha = \eta$.  This $(1,1)$ form is subject to the usual $(1,0)$-form gauge symmetries, where the gauge variation is given by the $\dbar$ operator. In addition, we have a gauge variation by a $(0,1)$ form, so that the gauge invariant quantity is $\eta = \partial \alpha$.

We let $\alpha_i$, $\alpha_z$ be the three $(0,1)$-form components of $\alpha$.  
Throughout, we will use the notation $D_{r, s} = {1 \over r! s!}\partial_{v_1}^r \partial_{v_2}^s$. The most general coupling invariant under the $\dbar$ gauge transformations is
\begin{equation} 
	\int_{z}  D_{r,s} \alpha_i e^i[r,s] + D_{r,s} \alpha_z e_z [r,s]. \label{eqn:axion_coupling} 
\end{equation}
This must be be invariant under the $(0,1)$ form gauge transformation $\alpha \mapsto \partial \gamma$, for $\gamma \in \Omega^{0,1}(\PT)$.  This only happens if
\begin{equation} 
	\int_{z} D_{r,s} \partial_{v_1} \gamma e^1[r,s] + D_{r,s} \partial_{v_2} \gamma e^2[r,s] + D_{r,s} \partial_z \gamma e_z[r,s] = 0. 
\end{equation}
The last term can be integrated by parts.  Inserting $\gamma = \delta_{z = z_0} v_1^r v_2^s$, we find the identity
\begin{equation} 
	r e^1[r-1,s] + s e^2[r,s-1] = \partial_z e_z[r,s]  \label{eqn:axion_relation}  
\end{equation}
Thus, the operators $e^i[r,s]$ are not independent, a linear combination of them can be expressed as a descendent of $e_z[r,s]$.

Let us find the linear combination of these operators that match with the notation in section \ref{s:axion}. There, we defined $E[r,s]$ and $F[r,s]$ by certain explicit closed $2$-forms, which were presented as the de Rham operator applied to $\dbar$ closed $(1,1)$-forms: 
\begin{equation} 
	\begin{split} 
		E[r,s]   	 &= - \partial \left(  \frac{1}{r+s} \d z \delta_{z = 0}  ( v_1^r v_2^s) \right)  \\
		F[r,s]  &= \partial \left( \delta_{z = 0} \frac{1}{r+s+2}(v_1^{r+1} v_2^s \d v_2 - v_1^{r} v_2^{s+1} \d v_1 )   \right) 
	\end{split}
\end{equation}
The corresponding operators are those obtained by replacing $\alpha$ by these  $(1,1)$-forms in equation \eqref{eqn:axion_coupling}.  We find
\begin{equation}
	\begin{split} 
		E[r,s] &= - \frac{1}{r+s} e_z[r,s] \\ 
		F[r,s] &= \frac{1}{r+s+2} \left( e_2[r+1,s] - e_1[r,s+1]  \right) . 
	\end{split} \label{eqn:axion_ef}	
\end{equation}
The linear relation \eqref{eqn:axion_relation} tells us that these form a basis for the generators of the Koszul dual algebra; except that we miss $e_z[0,0]$ (because $E[r,s]$ is only defined for $r+s> 0$).  

The relation \eqref{eqn:axion_relation} tells us that $e_z[0,0]$ is a topological operator, and so it must have trivial OPE with all other operators.  The algebra does not change in a significant way if we include or remove $e_z[0,0]$, and in fact it is most natural to remove it. This is because it is the operator coming by coupling to the bulk theory in the background where the $(1,1)$-form is $\d z \delta_{z = 0}$.  This $(1,1)$ form is closed, so we do not see it if we view the $(2,1)$ form as our fundamental field.  

One can reproduce the OPEs found in \S \ref{s:axion} by studying the BRST variation of these couplings directly.  It is automatic that the  Koszul duality approach will give the same answer, because at tree level the mode algebra of the Koszul dual algebra always reproduces the Lie algebra of gauge transformations that preserve the vacuum field configuration.  

Let us illustrate this with a particularly simple example, in which we supress the factor of $\frac{\lambda_{\g}}{2 (2 \pi \i)^{3/2}}$ present in the axion coupling

Let us consider the BRST variation of the coupling 
\begin{equation}
    \sum_{m,n}\int_{z''} \til{J}_a[m,n](z'')D_{m,n}B^a(z'')
\end{equation} in the presence of an axion. One of the terms in the BRST variation of $B$  is 
\begin{equation}
	\begin{split} 
		\delta_{BRST} B^a \sim \epsilon_{ij} (\partial_i \chi^a) \eta_{j z}   + \eps_{ij} (\partial_i A^a) \gamma_{j z} \\
		\eps_{ij} \partial_i \chi^a \partial_j  \alpha_z  + \eps_{ij} \partial_i A^a \partial_j \theta_z .  \label{eqn:Bvariation} 	 
	\end{split}	
\end{equation}
where $\gamma \in \Omega^{2,0}$ is the gauge transformation which shifts $\eta$ by $\dbar \gamma$, and $\theta \in \Omega^{1,0}$ is the gauge transformation shifting $\alpha$ by $\dbar \theta$.

We can insert this into the path integral, giving 
\begin{equation}
	\sum_{m,n}\int_{z''} \til{J}_a[m,n](z'')D_{m,n} \left(
\eps_{ij} \partial_i \chi^a \partial_j  \alpha_z  + \eps_{ij} \partial_i A^a \partial_j \theta_z  \right).	
\end{equation} 
This expression involves the ghost field $\chi$, which shifts $A$ by $\dbar \chi$, and the ghost $\theta$, shifting $\alpha$ by $\dbar \theta$. 

As such, this can be cancelled by the linearized BRST variation of the  bi-local expression involving the coupling of the gauge field $A$ and the $(1,1)$-form $\alpha$: 
\begin{align}\label{eq:JE}
	\sum_{r, s}\sum_{k, l}\int_{z, z'}J^a[r, s](z)D_{r, s}A^a_{\bar{z}}(z) e_z[k, l](z')D_{k,l}\alpha_z(z').
\end{align} 
using equations \eqref{eqn:axion_coupling} and \eqref{eqn:axion_ef}. As usual, requiring that the BSRT variation  \eqref{eq:JE} cancels with that of \eqref{eqn:Bvariation} will constrain the OPEs between the $E$ and $J^a$ towers.  

The linearized BRST variation of equation \eqref{eq:JE} replaces $A^a$ by $\dbar \chi^a$, or $\alpha$ by $\dbar \theta$.  Inserting this gives
\begin{align}
	\sum_{r, s}\sum_{k, l}\int_{z, z'}  J^a[r, s](z) e_z[k,l] (z') \left(  D_{r, s}A^a_{\bar{z}}(z)  D_{k,l}\dbar\theta _z(z')+  D_{r, s}\dbar \chi  D_{k,l}\alpha_z(z') \right)
\end{align} 
Integrating by parts then gives the required equality for BRST invariance:
\begin{multline}\label{eq:equality}
	\sum_{r, s}\sum_{k, l} \int_{z, z'}\bar{\partial}_{\bar{z}'}(J^a[r, s](z)e_z[k, l](z'))  D_{r, s}\chi^a(z) D_{k,l}\alpha_z(z') \\ = \sum_{m,n}\int_{z''} \til{J}_a[m,n](z'')D_{m,n}(\epsilon_{ij}\partial_i \chi^a \partial_j \alpha_z)(z''),
\end{multline} which must hold for general field configurations $\chi, \alpha_z$.

To constrain the OPE we can therefore insert test functions (suppressing the Lie algebra data for ease of notation)
\begin{align}
   \chi &= G(z, \bar{z})v_1^r v_2^s \\
   \alpha_z &= H(z, \bar{z}) v_1^k v_2^l
\end{align} where $G, H$ are both arbitrary functions on the defect. 
Inserting the test functions into \eqref{eq:equality} for arbitrary $G, H$ yields the following equality on integrands:
\begin{align}
    \bar{\partial}_{\bar{z}'}(J^a[r, s](z)e_z[k, l](z')) \\ &=  \delta_{z = z', \bar{z}= \bar{z}'}(lr-ks)\til{J}_a[k+r-1,l + s-1](z)
\end{align} which, translating from $e_z[k, l]$ to $E[k, l]$ and re-introducing the factor of $\what{\lambda}_{\g}$  gives exactly the desired OPE:
\begin{align}
	J^a[r, s](0)E[k, l](z) =  \what{\lambda}_{\g} \frac{1}{z}\frac{(lr-ks)}{k + l}\til{J}_a[k+r-1,l + s-1](0).
\end{align}

The other axion contributions to the OPEs can be extracted analogously. 

\subsection{Quantum corrections to the Koszul dual algebra}\label{s:quantum_chiral}

The Koszul duality point of view can also be used to readily obtain loop-level corrections to OPEs.  Of course, any attempt to quantum-correct the Koszul dual algebra will run into difficulties if the theory is anomalous on twistor space.   For anomalous theories, we can still perform leading-order quantum corrections to the algebra.  However, we do not expect that this persists to all orders\footnote{More formally, only for anomaly free theories do we expect these quantum corrections to define a flat family of vertex algebras. ``Flat'' means that the family of vertex algebras does not jump in size when we quantize.  We expect that in the anomalous cases, the quantum-corrected OPEs will force some states to vanish that do not vanish in the classical limit. }. 

For $5$-dimensional cousins of Chern-Simons theory, this was studied in \cite{Costello:2017fbo}.  A related analysis appears in the forthcoming work of \cite{Yehao}.  We state the result in \cite{Costello:2020jbh}.  We find that the diagram in Figure \ref{fig:correction} is not invariant under gauge transformations.    
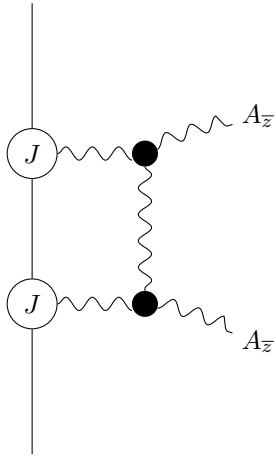
\begin{figure}
	\begin{center}
	\begin{tikzpicture}	
		\node[circle, draw] (J1) at (0,1) {$J$};
		\node[circle, draw] (J2) at (0,-1) {$J$};
		\node (A1) at (3,1.5)  {$A_{\zbar}$};
		\node (A2) at (3,-1.5)  {$A_{\zbar}$};

		\node[circle,draw,fill=black ]  (V1) at (1.5,1) {}; 
		\node[circle,draw,fill=black]  (V2) at (1.5,-1) {};  
		\draw[decorate, decoration={snake}] (J1) -- (V1) --  (A1);
		\draw[decorate, decoration={snake}] (J2) -- (V2) -- (A2);
		\draw[decorate, decoration={snake}] (V1) -- (V2); 
		\draw (0,3) -- (J1) -- (J2) -- (0,-3);

	\end{tikzpicture}
	\end{center}
	\caption{This diagram has a gauge anomaly leading to a quantum correction of the chiral algebra. \label{fig:correction}}
\end{figure}

The gauge anomaly is proportional to
\begin{equation} 
	\hbar \int_{w_1 = w_2 = 0} \eps_{ij} \left( \partial_{w_i} A_{\zbar}^a \right) (\partial_{w_j} \mf{c}^b ) K^{fe}  f^c_{ae} f^d_{bf}   J_c J_d    + \dots 
\end{equation}
where the ellipses indicate terms with more than two derivatives applied to the bulk fields and $K^{fe}$ is the Killing form on $\mf{g}$.   We would like this anomaly to be canceled by the first Feynman diagram in figure \ref{fig:cancel1}.   A necessary condition for this to happen is that the classical OPE of the operators $J[1,0]$ and $J[0,1]$ acquires a quantum correction:
\begin{equation} 
	J_a [1,0] (0)  J_b [0,1] (z) \simeq \frac{1}{z} f^c_{ab} J_c[1,1] + \hbar \frac{1}{z}  K^{fe} f_{ae}^c f_{bf}^d J_c[0,0] J_d[0,0]. \label{eqn:quantum_correction} 
\end{equation}
If we have such a quantum-corrected OPE, then the gauge variation of the expression
\begin{equation} 
	\int_{z_1, z_2} J_a[1,0] (z_1) \partial_{w_1}A_{\zbar}^a(z_1) J_b[0,1] (z_2)  
\end{equation}
gives us, at order $\hbar$, 
\begin{equation} 
	\hbar \int_{w_1 = w_2 = 0} \eps_{ij} \left( \partial_{w_i} A_{\zbar}^a \right) (\partial_{w_j} \mf{c}^b ) K^{fe}  f^c_{ae} f^d_{bf}   J_c J_d    
\end{equation}
which cancels the anomaly from the diagram in figure \ref{fig:correction}.

This relation is particularly powerful in the case that the indices $a,b$ correspond to commuting elements of $\g$. Then, classically, $J_a[1,0]$, $J_b[0,1]$ has a non-singular OPE, but it acquires polar part at the quantum level.

In the case of holomorphic BF theory, the relation coming from this diagram is different. This is because the propagator connects $\mc{A}$ with $\mc{B}$.   
\begin{figure}
	\begin{center}
	\begin{tikzpicture}
		\begin{scope}
		\node[circle, draw] (J1) at (0,1) {$\til{J}$};
		\node[circle, draw] (J2) at (0,-1) {$J$};
			\node (A1) at (3,1.5)  {$\mc{A}_{\zbar}$};
			\node (A2) at (3,-1.5)  {$\mc{A}_{\zbar}$};

		\node[circle,draw,fill=black ]  (V1) at (1.5,1) {}; 
		\node[circle,draw,fill=black]  (V2) at (1.5,-1) {};  
		\draw[decorate, decoration={snake}] (J1) -- (V1) --  (A1);
		\draw[decorate, decoration={snake}] (J2) -- (V2) -- (A2);
		\draw[decorate, decoration={snake}] (V1) -- (V2); 
		\draw (0,3) -- (J1) -- (J2) -- (0,-3);
		\end{scope}

		\begin{scope}[shift={(5,0)}]
			\node[circle, draw] (J1) at (0,1) {$\til{J}$};
		\node[circle, draw] (J2) at (0,-1) {$\til{J}$};
			\node (A1) at (3,1.5)  {$\mc{A}_{\zbar}$};
			\node (A2) at (3,-1.5)  {$\mc{B}_{\zbar}$};

		\node[circle,draw,fill=black ]  (V1) at (1.5,1) {}; 
		\node[circle,draw,fill=black]  (V2) at (1.5,-1) {};  
		\draw[decorate, decoration={snake}] (J1) -- (V1) --  (A1);
		\draw[decorate, decoration={snake}] (J2) -- (V2) -- (A2);
		\draw[decorate, decoration={snake}] (V1) -- (V2); 
		\draw (0,3) -- (J1) -- (J2) -- (0,-3);
		\end{scope}	
	\end{tikzpicture}
	\end{center}
	\caption{These diagrams quantum correct the chiral algebra for holomorphic BF theory. \label{fig:anomaly2}}
\end{figure}
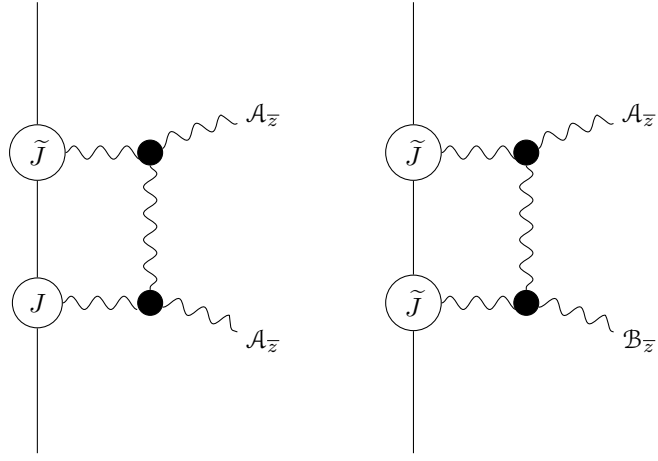

Let us assume that the Lie algebra elements $\t_a$, $\t_b$ corresponding to the external lines commute with one another.  Then, these diagrams give rise to quantum corrections
\begin{equation} 
	\begin{split} 
		J_a [1,0] (0)  J_b [0,1] (z)  \simeq C \frac{1}{z}  K^{fe} f_{ae}^c f_{bf}^d ( \til{J}_c[0,0] J_d[0,0] + \til{J}_d[0,0] J_c[0,0]  )   \\  
J_a [1,0] (0)  \til{J}_b [0,1] (z)  \simeq C \frac{1}{z}  K^{fe} f_{ae}^c f_{bf}^d \til{J}_c[0,0] \til{J}_d[0,0]    
		\label{eqn:quantum_correction_bf}  
	\end{split}
\end{equation}
Here $C$ is a constant we have not determined.  

Computing the analogous corrections for the gravitational theory should be quite interesting; in that setting, one must contend with the anomaly of Poisson BF theory which is not yet understood. 

It will also be interesting to compute higher-loop corrections to the gauge theory chiral algebra, including the axion field. We plan to pursue this in future work.

\section{Conformal blocks and local operators}\label{s:blocks}

So far, we have constructed a chiral algebra associated to any local holomorphic QFT on twistor space.  In this section, we will show how the vector space of conformal blocks of the chiral algebra is isomorphic to the space of local operators of the $4d$ theory. We will consider a local operator inserted at the origin of $\mathbb{R}^4$ and subsequently suppress the position dependence of the operator $\mathcal{O}:=\mathcal{O}(0)$ and its corresponding conformal block $\langle \mathcal{O}|$. We discuss infinitesimal translations of the local operator below.

Since the term conformal blocks is somewhat overloaded (with slightly different meanings in the math and physics literature), let us explain precisely what we 
mean.  Let $\mc{C}$ be any vertex algebra, which we \emph{do not} assume has a stress-energy tensor.  We will assume that the algebra $\mc{C}$ is the algebra of local operators at the boundary of a three-dimensional partially holomorphic theory, on $\R_{>0} \times \C$, with coordinates $r$ and $z$. We will assume that the bulk $3d$ theory has a stress-energy tensor, but that the boundary algebra may not; this means that $\mc{C}$ has an action of the Virasoro algebra but that it does not come from a Virasoro current. As we have seen, this is the case for  the algebras relevant to celestial holography.  We will assume the algebra $\mc{C}$ lives at $r = \infty$.

We will define the conformal blocks of $\mc{C}$ to be the Hilbert space of the $3d$ theory on $\R_{> 0 } \times S^2$, at $r = 0$.  If $\psi$ is a conformal block (with this definition) then, for any collection $\mc{O}_i$ of local operators in the algebra $\mc{C}$, we can define
\begin{equation}
	\ip{\psi \mid \mc{O}_1(z_1) \dots \mc{O}_n(z_n)}.
\end{equation}
This is a correlator in the bulk-boundary system on $[0,\infty] \times \CP^1$, where at $0$ we place the state $\psi$, at $\infty$ we have the boundary condition giving rise to the algebra $\mc{C}$, and we insert the operator $\mc{O}_i(z_i)$ at points $z_i \in \infty \times \CP^1$. 

Let us explain how conformal blocks of the Koszul dual vertex algebra can be matched with local operators of the $4d$ theory. Consider any holomorphic theory on twistor space.  Let $\PT'$ denote the complement of the $\CP^1$ at the origin in $\PT$.  We have a double fibration
\begin{equation} 
	\R^4 \setminus 0 =  R_{> 0} \times S^3 \longleftarrow  \PT' = S^3 \times \CP^1 \times \R_{> 0} \longrightarrow \R_{> 0} \times \CP^1
\end{equation}
The left-hand arrow is the standard $\CP^1$ fibration of twistor theory, and the right-hand arrow is the $S^3$ fibration we used to build our $3d$ theory.  

We are starting with a field theory on $\PT'$, and performing KK reduction (including all KK modes, if they are present) along the left or right arrows. When we do this, the Hilbert space of the lower-dimensional theory is the same as that of the theory on $\PT'$, since we have not really changed anything:
\begin{equation} 
	\mc{H}(S^3 \times \CP^1) = \mc{H}(S^3) = \mc{H}(\CP^1).  
\end{equation}
In each case, this is the Hilbert space at the locus where the radial direction $r$ is zero, working in radial quantization. 

The Hilbert space $\mc{H}(\CP^1)$ of the $3d$ theory is, by definition, the space of conformal blocks. The Hilbert space $\mc{H}(S^3)$ is the space of local operators in the theory on $\R^4$, since it is obtained by radial quantization.  It is important to note here that compactifying from $\PT$ to $\R^4$ does not introduce KK modes, so we find the space of local operators of an ordinary $4d$ theory. 

In this way, by considering compactification of the theory on twistor space in two different ways, we have identified conformal blocks with local operators.

For the rest of the section we will give several different perspectives on this result.

\subsection{Defining conformal blocks axiomatically}
At a more abstract level, following Beilinson and Drinfeld \cite{beilinson2004chiral}, we can define the  vector space of conformal blocks in an axiomatic way. We say that a conformal block $\psi$ is the data needed to define correlation functions for elements of the algebra $\mc{C}$. 
 Let $\mc{C}$ to mean the vacuum module of the vertex algebra, and let $\mc{C}_z$ be this module placed at $z \in \CP^1$ (the vacuum module naturally forms a bundle on $\CP^1$).  Beilinson-Drinfeld's definition takes a conformal block to be the data of a linear map
 \begin{equation}
\begin{split}
 \psi &: \mc{C}_{z_1} \otimes \dots \otimes \mc{C}_{z_n} \to \C \\
 \mc{O}_1 \otimes \dots\otimes \mc{O}_n &\mapsto \ip{\psi \mid \mc{O}_1(z_1) \dots \mc{O}_n(z_n)} 
 \end{split}
 \end{equation}
 These linear maps must satisfy all the properties expected for correlation functions:
\begin{enumerate} 
	\item First, $\ip{\psi \mid \mc{O}_1(z_1) \dots \mc{O}_n(z_n)}$  must be rational functions of $z_i$. 
	\item Secondly, as  
$z_i$ approaches $z_j$, we can replace $\mc{O}_i(z_i) \mc{O}_j(z_j)$ by the OPE
 \begin{equation}
	 \mc{O}_i(z_i) \mc{O}_j(z_j) \sim \sum_{k > 0}  (z_i - z_j)^{-k} \mc{O}'_{ij,k}(z_i)  
\end{equation}
in the correlator. We write the expression for $i = 1$, $j = 2$ for simplicity:
		\begin{equation}
			\begin{split} 	
			\ip{\psi \mid \mc{O}_1(z_1) \dots \mc{O}_n(z_n)} =   \sum_{k > 0}  (z_1 - z_2)^{-k} \ip{\psi \mid  \mc{O}'_{12,k}(z_1) \mc{O}_3(z_3) \dots \mc{O}_n(z_n)} \\
				+ \text{ expressions regular in } z_1 - z_2
				\end{split}
		\end{equation}
	\item Finally, differentiating the correlator with respect to $z_i$ is the same as replacing the operator $\mc{O}_i$ by $L_{-1} \mc{O}_i$. 
\end{enumerate}
		This definition is the linear dual of Beilinson-Drinfeld's \emph{factorization homology}.  Factorization homology\footnote{ It is important in some cases that Beilinson-Drinfeld's theory is homological in nature, and so produces a graded vector space.  We have described $H_0$ of their construction.  The (linear duals of) parts of this vector space in other degrees will correspond to operators in the $4d$ theory of non-zero ghost number.  }  is the universal vector space in which correlation functions can take values. An element of its linear dual gives a set of correlation functions valued in $\C$.

\subsection{Explicitly matching conformal blocks with local operators}
\label{sec:explicit_match}
		To indicate how this axiomatic definition works in practice, let us consider a very simple example: a free scalar field theory on $\R^4$.  The corresponding theory on $\PT$ is an Abelian holomorphic Chern-Simons theory, where the gauge field lives is $\mc{A} \in \Omega^{0,1}(\PT, \Oo(-2))$. We will identify sections of the bundle $\Oo(-2)$ with functions which vanish to order $2$ at $z = \infty$.  If we do this, then the Lagrangian on twistor space is
		\begin{equation} 
			\int \mc{A} \dbar \mc{A} \d v_1 \d v_2 \d z 
		\end{equation}
		where, as before, $v_i$ are linear functions on the $\Oo(1)$ fibres on twistor space, that have a first order pole at $z = \infty$.  

In this situation, the Koszul dual chiral algebra is Abelian, generated by conformal primaries $J[k,l]$ with trivial OPEs.    Because $J$ couples to $\mc{A}$ by
\begin{equation} 
	\int \mc{A} \d z J[0,0] 
\end{equation}
the zero of order $2$ in $\mc{A}$ at $\infty$ cancels the pole in $\d z$. Therefore, the correlation functions of $J[0,0](z)$ only have poles at the location of the other operators, and not at at $z = \infty$.  More generally, the correlation functions of $J[k,l](z)$ can have a pole of order $k+l$ at $\infty$. 

Under the correspondence between local operators and conformal blocks,  the operator $\mc{O}$ with $\mc{O}(\phi) = \phi(0)$ corresponds to the conformal block where
		\begin{equation} 
			\ip{\mc{O} \mid J(z)} = 1. 
		\end{equation}

		To understand how to relate other operators to conformal blocks, we need to know how to differentiate conformal blocks with respect to the space-time coordinates $u_i$, $\br{u}_i$ (where $u_i$ are holomorphic with respect to the complex structure associated to $z = 0$, and $\ubar_i$ with respect to the complex structure associated to $z = \infty$).  

		Differentiation of conformal blocks comes from a symmetry of the vertex algebra associated with the action of the translation symmetry $\C^4$ on twistor space.  The vector fields $\partial_{u_i}$, $\partial_{\ubar_i}$ on $\R^4$ become the vector fields
		\begin{equation} 
		\begin{split}
			\partial_{u_1} &= \partial_{v_1} \\
			\partial_{u_2} &= \partial_{v_2} \\
			\partial_{\ubar_1} &= -z\partial_{v_2} \\
			\partial_{\ubar_2} &= z \partial_{v_1}
		\end{split}
		\end{equation}
(where on the right hand side, we have dropped anti-holomorphic vector fields whose action on everything is BRST exact; see (5.2.3) of \cite{Costello:2021bah} for the complete expressions). 

From the equation for how the fields $\mc{B}$, $\mc{A}$ on twistor space couple to the generators $J[k,l]$, $\til{J}[k,l]$ of the Koszul dual algebra, we find that 
\begin{equation} 
	\begin{split} 
		\partial_{v_1} J[k,l] =-k J[k-1 ,l] \\
		\partial_{v_2} J[k,l] = -l J[k,l-1].
	\end{split}
\end{equation}
This means that, if $\mc{O}$ as above is the operator which measures $\phi(0)$, the operators measuring the derivatives of $\phi$ give rise to the correlation functions
\begin{equation} 
	\begin{split} 
		\ip{\partial_{u_1} \mc{O} \mid J[1,0]  }  &= 1\\
		\ip{\partial_{u_2} \mc{O} \mid J[0,1]  }  &= 1\\
		\ip{\partial_{\ubar_1} \mc{O} \mid J[0,1]  }  &= -z\\
		\ip{\partial_{\ubar_2} \mc{O} \mid J[1,0]  }  &=z.
	\end{split}
\end{equation}
Continuing in this vein, we note that this prescription is compatible with the equations of motion, since
\begin{equation} 
\begin{split}
	\ip{\partial_{u_1} \partial_{\ubar_1}  \mid J[1,1]   }  &= -z\\
\ip{\partial_{\ubar_2} \partial_{u_2}  \mid J[1,1]   }  &= z. 
\end{split}
\end{equation}

\subsection{Conformal blocks as Lie algebra cohomology}
In this subsection we will use a result of Beilinson-Drinfeld \cite{beilinson2004chiral} to give a mathematical proof that conformal blocks of the chiral algebra match local operators of the corresponding $4d$ theory, at the classical level.   

The argument is quite general, but  we will illustrate the result for gauge theory without the axion field, where our chiral algebra is built from $J$ and $\til{J}$.  This chiral algebra is a Kac-Moody vertex algebra built from a holomorphic bundle of Lie algebras on $\CP^1$, which we now describe. 

Give $\g \oplus \g^\vee$ its natural Lie bracket, where $\g^\vee$ transforms in the coadjoint representation of $\g$ and the bracket between two elements of $\g^\vee$ vanishes.   In the same way, we can make $\g \otimes \Oo \oplus \g^\vee \otimes \Oo(-4)$ into a holomorphic bundle of Lie algebras on twistor space $\PT$.  There is a natural map $\pi : \PT \to \CP^1$, and we can push forward this bundle of Lie algebras to get a sheaf $\mc{L}$ of Lie algebras on $\CP^1$:
\begin{equation} 
	\mc{L} = \pi_\ast (\g \otimes \Oo \oplus \g \otimes \Oo(-4)). 
\end{equation}
As a sheaf on $\CP^1$, we have
\begin{equation} 
	\mc{L} = (\g \otimes \Oo \oplus \g^\vee \otimes \Oo(-4) ) \otimes_{\Oo} \Sym^\ast \left( \Oo(-1) \oplus \Oo(-1)  \right) . 
\end{equation}

The chiral algebra built from $J$ and $\til{J}$ is described in terms of $\mc{L}$ by a construction that Beilinson-Drinfeld call the \emph{chiral envelope}; it is a vertex algebra version of the universal enveloping algebra of a Lie algebra. Concretely, this means the chiral algebra is a kind of Kac-Moody algebra built from $\mc{L}$, so that:
\begin{enumerate} 
	\item The mode algebra of the vertex algebra is the universal enveloping algebra $U \mc{L}(\C^\times)$ of sections of $\mc{L}$ on the punctured complex line $\C^\times$.
	\item The vacuum module of the vertex algebra is the induced module for the trivial module of $U \mc{L}(\C)$. 
\end{enumerate}

Beilinson-Drinfeld prove a general result about the cochain complex of (derived) conformal blocks in this context.  To state the theorem, let us consider the \v{C}ech cohomology
\begin{equation} 
	 H^\ast (\CP^1,\mc{L})  
\end{equation}
of $\CP^1$ with coefficients in $\mc{L}$. This is a graded Lie algebra. 

Beilinson-Drinfeld show that\footnote{More precisely, there is a spectral sequence starting at the right hand side of the equation and converging to conformal blocks, but in this case it degenerates.} 
\begin{equation} 
	\text{ conformal blocks} = H^\ast_{Lie} ( H^\ast(\CP^1,\mc{L} ) ) 
\end{equation}
On the right hand side we have the Lie algebra cohomology of the graded Lie algebra $H^\ast(\CP^1, \mc{L})$.  

We would like to identify this with the space of local operators of self-dual gauge theory.  To do this, we will compute the cohomology groups $H^\ast(\CP^1, \mc{L})$.    Because $\mc{L}$ is built as a sheaf push-forward from $\PT$, we have
\begin{equation} 
	H^\ast(\CP^1,\mc{L}) = H^\ast(\PT, \g \otimes \Oo \oplus \g \otimes \Oo(-4)) . 
\end{equation}
It is easy to see that 
\begin{equation} 
	 H^0(\PT, \g \otimes \Oo \oplus \g \otimes \Oo(-4))  = \g. 
\end{equation}
Further, by the Penrose-Ward correspondence, we have 
\begin{equation}
	\begin{split} 
		H^1(\PT,  \Oo ) &= \{ F \in \Omega^2_+(\R^4) \mid \d F = 0 \}   	\\
		H^1(\PT, \Oo(-4) ) &= \{ B \in \Omega^2_- (\R^4) \mid \d B = 0 \}. 
	\end{split}\label{eqn:PW}	
\end{equation}

Now let us compute the Lie algebra cohomology of $H^\ast (\CP^1, \mc{L})$.  This graded Lie algebra is concentrated in degrees $0$ and $1$, and in degree $0$ it is $\g$. Therefore the Lie algebra cohomology is
\begin{equation} 
	H^\ast(\g, \Sym^\ast (H^1(\CP^1,\L)^\vee).
\end{equation}
For any representation $R$ of a simple Lie algebra $\g$, the Lie algebra cohomology of $\g$ with coefficients in $R$ is
\begin{equation} 
	H^\ast(\g,R) = R^{G} \otimes H^\ast(\g) 
\end{equation}
where $R^G$ is the $G$-invariants.   In particular,
\begin{equation} 
	H^\ast\left(\g, \Sym^\ast (H^1(\CP^1,\L)^\vee)	\right) = \left( \Sym^\ast (H^1(\CP^1,\L)^\vee \right)^G \otimes H^\ast(\g). \label{eqn:liecoho}
\end{equation}

It follows from equation \eqref{eqn:PW} that the symmetric algebra of the dual of $H^1(\PT, \L)$ is polynomials in the value of $F^a_{\alpha \beta}$, $B^a_{\dot{\alpha} \dot{\beta}}$ and their derivatives, modulo the linearized equation of motion: 
\begin{equation} 
	\Sym^\ast H^1(\PT,\L)^\vee = \C[F^a_{\dot\alpha\dot\beta}, \partial_{i} F^a_{\dot\alpha \dot\beta}, \dots , B^a_{{\alpha}{\beta}}, \partial_{i} B^a_{{\alpha}{\beta}}, \dots ] / \ip{ \Gamma_i^{\alpha \dot \alpha}  \partial_i F^a_{\dot\alpha \dot\beta},   \Gamma_i^{\alpha \dot \alpha}  \partial_i B^a_{ \alpha  \beta}  }\label{eqn:opng} 
\end{equation}
The right hand side consists of local operators, without imposing gauge invariance.  

Gauge invariance here is equivalent to being invariant under the constant gauge transformations $G$ (as all expressions are gauge-covariant, if derivatives are taken to be covariant).

Equation \eqref{eqn:liecoho} tells us that the Lie algebra cohomology  is, in ghost number zero, the $G$-invariants of \eqref{eqn:opng}. This completes the proof that the Lie algebra cohomology of $H^\ast(\CP^1, \mc{L})$ is isomorphic to gauge-invariant local operators of self-dual gauge theory, in ghost number zero.   

\subsection{Conformal blocks from a \v{C}ech picture}
Next, let us discuss a \v{C}ech picture for conformal blocks, which connects closely with Ward's \cite{Ward:1977ta} description of self-dual gauge fields.  We will use the sheaf $\mc{L}$ of Lie algebras on $\CP^1$ introduced in the previous subsection.  

The Lie algebra $\mc{L}(\C^\times)$ is the gluing data for solutions of self-dual gauge theory.  More precisely, let $\exp(\mc{L}(\C^\times))$ be the group exponentiating $\mc{L}(\C^\times)$. (It is isomorphic to the group of holomorphic maps $\C^2 \times \C^\times \to G$).  Similarly, we can define $\exp(\mc{L}(\C_0))$  and $\exp(\mc{L}(\C_{\infty}))$ as the groups associated to $0$ and $\infty$.

The Penrose-Ward correspondence tells us that the space of complexified solutions to the equations of motion of self-dual gauge theory on $\R^4$ is an open subset of the double quotient
\begin{equation} 
	\exp(\mc{L}(\C_{0})) \backslash \exp(\mc{L}(\C^\times)) / \exp(\mc{L}(\C_{\infty})). 
\end{equation}
It is an open subset as we are exluding the locus corresponding to holomorphically non-trivial bundles on $\CP^1$.

The space of local operators of the theory can be written as functions on the solutions to the equations of motion. Since we are working perturbatively, we will replace the space of solutions to the EOM by a formal neighbourhood of the trivial solution.  To do this, we will replace the group $\exp(\mc{L}(\C^\times))$ by the formal group $\what{\exp}(\mc{L}(\C^\times))$, i.e.\ the formal neighbourhood of the identity in the group.  We will do the same for $\exp(\mc{L}(\C_0))$, $\exp(\mc{L}(\C_{\infty}))$.  Then, the formal moduli space of solutions to the equations of motion is the double quotient
\begin{equation} 
	\what{\op{EOM}} = 	\what{\exp}(\mc{L}(\C_{0}))\backslash  \what{\exp}(\mc{L}(\C^\times)) / \what{\exp}(\mc{L}(\C_{\infty})). 
\end{equation}
The vector space of local operators is then functions on this:
\begin{equation} 
	\text{local operators} = \Oo( \what{\op{EOM}} ). 
\end{equation}
(To be precise, we should look at functions which are finite sums of functions which are eigenvalues under the scaling action of $\C^\times$ on $\what{\op{EOM}}$.  This corresponds to looking at operators which involve only finitely many derivatives). 

We would like to identify this with conformal blocks.  To do this, we will introduce a \v{C}ech description of conformal blocks, which might be familiar to some readers from the study of conformal blocks of the WZW model.  The algebra of operators on the theory $\C^\times$ is the universal enveloping algebra $U(\mc{L}(\C^\times))$.  The vacuum modules at zero and infinity are the induced module 
\begin{equation} 
	\begin{split}
		\op{Vac}_0 &= \op{Ind}_{U(\mc{L}(\C)_0)}^{U(\mc{L}(\C^\times) }  \C \\
\op{Vac}_\infty &= \op{Ind}_{U(\mc{L}(\C)_\infty)}^{U(\mc{L}(\C^\times) }  \C  
	\end{split}
\end{equation}
The conformal blocks are then
\begin{equation} 
	\Hom ( \op{Vac}_0 \otimes_{U(\mc{L}(\C_0))} U(\mc{L}(\C^\times) ) \otimes_{U(\mc{L}(\C_{\infty})} \op{Vac}_{\infty} , \C ).  
\end{equation}
This is the same as 
\begin{equation} 
	\Hom ( U ( \mc{L}(\C^\times) ), \C  )^{\mc{L}(\C_0) \oplus \mc{L}(\C_{\infty} ) } 
\end{equation}
i.e.\ the linear maps from $U (\mc{L}(\C^\times) )$ to $\C$ which are invariant under the action of $\mc{L}(\C_0)$, acting on the left, and $\mc{L}(\C_{\infty})$, acting on the right.

The connection of this to correlation function definition of conformal blocks is as follows.  The algebra $U(\mc{L}(\C^\times))$ is the mode algebra of the vertex algebra, generated by the modes $\oint J[r,s] z^k \d z$, $\oint \til{J}[r,s] z^k \d z$.  Suppose we  have a set of correlation functions denoted by $\ip{ \dots }$. Then, expressions like
\begin{equation} 
	\oint_{\abs{z}_1 < \dots < \abs{z_k} } \ip{J[r_1,s_1](z_1) \dots \til{J}[r_k,s_k] (z_k) } z_1^{n_1} \dots z_k^{n_k} \d z_1 \dots \d z_k	 
\end{equation}
give a linear functional
\begin{equation} 
	U(\mc{L}(\C^\times)) \to \C. 
\end{equation}
This linear functional is invariant under left multiplication by an element of $\mc{L}(\C_0)$ and right multiplication by $\mc{L}(\C_{\infty})$, because these are the  modes that preserve the vacuum at $0$ and $\infty$. 

Now let us connect the \v{C}ech definition of conformal blocks to local operators.  The first point we will need is that -- as is well-known by algebraists -- we can identify the linear dual of $U(\mc{L}(\C^\times))$ with functions on the formal group:
\begin{equation} 
	\Hom( U(\mc{L}(\C^\times)), \C) = \Oo( \what{\exp} ( \mc{L}(\C^\times) ) ). 
\end{equation}
Now, conformal blocks are the invariants of the left hand side with respect to $\mc{L}(\C_0) \oplus \mc{L}(\C_{\infty})$, and so can be identified with functions on the double quotient:
\begin{equation}
	\begin{split} 
		\Hom( U(\mc{L}(\C^\times)), \C)^{\mc{L}(\C_0) \oplus \mc{L}(\C_{\infty} )}  &=\Oo\left( \what{\exp}(\mc{L}(\C_{0}))\backslash  \what{\exp}(\mc{L}(\C^\times)) / \what{\exp}(\mc{L}(\C_{\infty})) \right) \\
		&= \Oo( \what{\op{EOM}} ). 
	\end{split}	
\end{equation}
In this way, conformal blocks (classically) are identified in a canonical way with (classical) local operators.

\subsection{Conformal blocks after quantizing and factorization algebras}
We have presented several perspetives for why conformal blocks are the same as local operators: by thinking about different dimensional reductions of the $6d$ theory, or by using more-or-less standard computations of conformal blocks of Kac-Moody type algebras. 

Here, let us describe briefly another rather abstract perspective, based on the theory of factorization algebras \cite{Costello:2016vjw}.   The argument is quite general. We start with any holomorphic quantum field theory on twistor space, whose factorization algebra of quantum observables is $\Obs^q_{6d}$.  Restricted to the $\CP^1$ over $0\in \R^4$, $\Obs^q_{6d}$ can be viewed as  a dg vertex algebra. 

For self-dual gauge theory, at the classical level, this factorization algebra sends $U \in \CP^1$ to 
\begin{equation} 
	\Obs^{cl}_{6d}(U) =  C^\ast(\mc{L}(U)).  
\end{equation}

Let us consider the Koszul dual $(\Obs^q_{6d})^!$.   The general theory of Koszul duality of vertex algebras has not been fully developed; here we will assume that it works in a similar way to Koszul duality for associative algebras.

Self-dual gauge theory on $\R^4$, with the axion field, comes from dimensional reduction of the theory on twistor space. The factorization-algebra way to perform dimensional reduction is the push forward \cite{Costello:2016vjw}, also known as factorization homology. In the case of chiral theories on a curve coincides with Beilinson-Drinfeld's chiral homology \cite{beilinson2004chiral}.  Therefore, $4d$ observables are
\begin{equation} 
	\Obs^q_{4d} = \pi_\ast \Obs^q_{6d} 
\end{equation}
where $\pi_\ast$ is the push-forward along the map from $\CP^1$ to a point. 

Factorization homology (or push-forward) is the linear dual of conformal blocks of $\Obs^q_{6d}$.   In the context of topological factorization algebras, it is known \cite{Ayala:2014zkd} that factorization homology sends Koszul duality to linear duality.   If we assume that this result holds in the vertex algebra context, then we deduce that 
\begin{equation} 
	\left( \pi_\ast (\Obs^q_{6d})^! \right)^\vee = \pi_\ast \Obs^q_{6d} = \Obs^q_{4d} 
\end{equation}
The left hand side of this equation is the conformal blocks of the Koszul dual vertex algebra,  which we have identified with the right hand side, local operators of the $4d$ theory.

\subsection{Form factors and correlators }
We have seen that conformal blocks of the vertex algebra are in bijection with local operators in the $4d$ CFT.   We also know that generators of the chiral algebra are single-particle conformal primary states of the $4d$ theory, in the language of celestial holography.  

It is essentially a formal consequence of this that correlators of the chiral algebra using a particular conformal block, are the same as scattering amplitudes in the presence of an insertion of the corresponding local operator, i.e. form factors. More precisely, to obtain the complete form factor from our form factor integrand, one must integrate over the positions of the operator insertions; while one can do this by hand, the chiral algebra formulation most naturally computes the integrand, namely with fixed positions of operator insertions, so in what follows ``form factor'' should be understood to mean ``form factor integrand.''

It is perhaps easiest to understand the form factor/correlation function correspondence by thinking about a conformal block which arises by introducing some new degrees of freedom along the $\CP^1$ in twistor space living over $0 \in \R^4$.

Suppose we can couple, in some gauge invariant way, some $2d$ chiral CFT to the bulk system along this $\CP^1$, given by some collection of chiral fermions $\psi_i$. We will not be explicit about the nature of the coupling or how many fermions we have, as our goal is to give an inuitive understanding of the relationship between scattering amplitudes and correlation functions.

Let $\mc{C}$ is the algebra of operators of the $2d$ free fermion system we couple. By the definition of Koszul duality, we have a homomorphism from the Koszul dual algebra to $\mc{C}$. This means that we have states 
\begin{equation}
	J[k,l], \ \til{J}[k,l], \ E[k,l],\ F[k,l] 	
\end{equation}
in $\mc{C}$.  These could be written schematically, in Lagrangian terms,
\begin{equation} 
	\sum D_{k,l}  \mc{A} J[k,l] ( \psi) + \dots
\end{equation}
where $D_{k,l}  = \frac{1}{k! l!} \partial_{v_1}^k \partial_{v_2}^l $, and $J[k,l]$ is some even polynomial in the fermionic fields and their derivatives.   

Since our $2d$ system is a system of free fermions, it has only one conformal block.  Then, it gives rise to a conformal block for the Koszul dual chiral algebra, defined by the correlation functions of the operators $J[k,l],\dots$ in the $2d$ system.  These have a path integral representation
\begin{equation} 
	\ip{J[k,l] (z_1) \dots \til{J}[r,s](z_n) } = \int_{\psi} e^{\int_{\CP^1} \psi \dbar \psi} J[k,l](z_1) \dots \til{J}[r,s](z_n) \label{eqn:fermion_coupling} 
\end{equation}
This conformal block corresponds to a local operator in the $4d$ CFT, which is of course that obtained by integrating out the fermionic degrees of freedom. 

The operators $J[k,l], \dots$ in the $2d$ CFT are obtained by coupling to background bulk fields which are conformal primary state corresponding to soft modes. To compute form factors, we typically resum those soft modes as in \eqref{eqn:generating_function} to obtain standard momentum eigenstates. These serve as the asymptotic states in the explicit form factor computations that follow.   

Our claim is that this correlator is the same as the form factor integrand of the $4d$ CFT, with the corresponding choice of local operator insertion.  Since the $4d$ system with a local operator arises from the $6d$ system with a defect by dimensional reduction, we can compute scattering in the $6d$ + defect system. 

In $6d$, scattering in the presence of the defect is given by a path integral just like \eqref{eqn:fermion_coupling}, but where we also need to perform a path integral over the $6d$ fields $\mc{A}$, $\mc{B}$, ...  The key point is that the exchange of $6d$ fields can not contribute to the form factor. 

This is because, working in $6d$, we can choose an axial gauge\footnote{Axial gauges are often too singular to work at the quantum level, but we can make it a little less singular by allowing the field to propagate a very small amount in the $z$ direction.} where fields propagate only in the $v_{1}$, $v_{2}$ plane, and not in the $z$ direction.  If we do this, the conformal primary states at different values of $z$ do not talk to each other in the bulk, and any scattering process is entirely mediated by the exchange fermions. 

This kind of reasoning generalizes to apply to any conformal block of the Koszul dual chiral algebra, not just one that arises by coupling to a free fermion system.

This argument also implies our Theorem \ref{mainthm}.  This result gives an expansion for the integrand of form factors, in terms of the operator product expansion
\begin{equation} 
	\op{tr}(B^2)(0) \op{tr}(B^2)(x_1) \dots \op{tr}(B^2)(x_{n-1)} \sim \sum F^i(x_1,\dots, x_{n-1}) \mc{O}_i(0) 
\end{equation}
where $\mc{O}_i$ runs over a basis of local operators in the $4d$ CFT.  The statement is that scattering amplitudes in the presence of the operators  $\op{tr}(B^2)(0) \op{tr}(B^2)(x_1) \dots \op{tr}(B^2)(x_{n-1})$ also have an expansion 
	\begin{equation} 
		\begin{split} 
			\sum F^i(x_1,\dots, x_{n-1})           \Big\langle \mc{O}_i(0) \mathrel{\Big|} J^{a_1}[r_1,s_1] (z_1) & \dots J^{a_n}[r_1,s_1](z_n)  \\
				&  \til{J}^{b_1}[t_1,u_1](z'_1) \dots \til{J}^{b_m}[t_m,u_m](z'_m) \Big \rangle. 
		\end{split}
	\end{equation}	
This is a consequence of what we already know: the form factors for $\prod \op{tr}(B^2)(x_i)$ are equivalent to the form factors in the presence of the OPE of those operators.

\section{Correlation functions for the operator $B^2$ and the Parke-Taylor formula}
In this section, we will show that the correlation functions of our vertex algebra, built from the conformal block corresponding to $B^2$, give the Parke-Taylor formula for the color-ordered tree-level MHV scattering of $n$ gluons:
\begin{equation}
    \mathcal{A}_n = \frac{\langle \lambda_i \lambda_j \rangle^4}{\langle \lambda_1 \lambda_2 \rangle \langle \lambda_2 \lambda_3 \rangle \ldots \langle \lambda_n \lambda_1 \rangle},
\end{equation}where we have omitted the standard group theory factor. Here, we have expressed the amplitude in terms of the standard homogeneous coordinates $\lambda_{\alpha}, \ \alpha =0,1$ on the twistor base $\mathbb{CP}^1$. In affine coordinates, $\lambda_{i, \alpha} = (1, \ z_i)$, and we can re-express the brackets as, e.g., $\langle \lambda_i \lambda_j \rangle = z_i - z_j$. 

Let us explain why we should expect this to be true (in fact, this computation has a lot in common with the computation by Lionel Mason in \cite{Mason:2005zm}).  Self-dual Yang-Mills deformed by the operator $B^2$ is equivalent to ordinary Yang-Mills.  We will view the $B^2$ term as a bi-valent vertex added to self-dual Yang-Mills. Tree-level scattering processes which involve only one $B^2$ vertex are the MHV amplitudes, with two negative helicity particles and an arbitrary number of positive helicity.  

If we have only one $B^2$ vertex, we get essentially the same answer if we view it as an operator, placed at the origin, or if we integrate over the position of the operator. The only difference is whether we include the conservation of momentum $\delta$-function in the amplitude.  Our chiral algebra formulation naturally connects to scattering of self-dual Yang-Mills in the presence of an operator, corresponding to the choice of conformal block.   

Now let us turn to the computation.  For self-dual gauge theory, our field $\mc{A}$ on twistor space is a $(0,1)$ form valued in $\g$.  The coupling 
\begin{equation} 
	\int_{\CP^1} \mc{A} \d z J[0,0](z) 
\end{equation}
makes sense as long as $J[0,0](z)$ vanishes to order $2$ at $\infty$.  Similarly, the coupling 
\begin{equation} 
	\int_{\CP^1}\frac{1}{k!}\frac{1}{l!} \partial_{v_1}^k \partial_{v_2}^l  \mc{A} \d z J[k,l](z) 
\end{equation}
makes sense if $J[k,l]$ vanishes to order $2-k-l$ at $z = \infty$. 

Since $\mc{B}$ is a section of $\Oo(-4)$, which is the square of the canonical bundle on $\CP^1$, we can build a $(1,1)$ form on $\CP^1$ by contracting $\mc{B}$ with the vector field $\partial_z$, which vanishes to order $2$ at $z = \infty$. Thus,
\begin{equation} 
	\int_{\CP^1}\frac{1}{k!}\frac{1}{l!} \partial_{v_1}^k \partial_{v_2}^l \iota_{\partial_z} \mc{B} \d z \til{J} [k,l](z) 
\end{equation}
makes sense as long as $\til{J}[k,l]$ has a pole of order at most $2+k+l$ at $z = \infty$.

Next, we need to identify the operator corresponding to $B^2$.  Clearly, since $B$ comes from the field $\mc{B}$ on twistor space, and $\mc{B}$ couples to the elements $\til{J}[k,l]$ of our chiral algebra, the corresponding conformal block must give an expectation value to $\til{J}[0,0](z_1) \til{J}[0,0](z_2)$:
\begin{equation} 
	\ip{\op{tr}(B^2) \mid \til{J}_a[0,0](z_1) \til{J}_b[0,0](z_2) } = F(z_1,z_2) \op{tr}(\t_a \t_b)  
\end{equation}
where the constraints on the behaviour of the operators $\til{J}$ at $z = \infty$ tell us that $F(z_1,z_2)$ is at most quadratic in each variable.

We can fix $F(z_1,z_2)$ by symmetry.  As in the discussion above, it is most natural to view $F(z_1,z_2) \partial_{z_1} \partial_{z_2}$ as a bivector on $\CP^1 \times \CP^1$, and then we can ask that it is invariant under the action of the $SU(2)$ rotating $\CP^1$.  This invariance corresponds to the fact that $\op{tr}(B^2)$ is $SO(4)$ invariant, and hence in particular $SU(2)$ invariant. There is only one bivector on $\CP^1 \times \CP^1$ invariant under $SU(2)$, namely $(z_1 - z_2)^2 \partial_{z_1} \partial_{z_2}$.  

We conclude that, up to a constant,
\begin{equation} 
	\ip{\op{tr}(B^2) \mid \til{J}_a[0,0](z_1) \til{J}_b[0,0](z_2) } = z_{12}^2 \op{tr}(\t_a \t_b). 
\end{equation}
Further, correlation functions containing the operators $J[k,l]$ or $\til{J}[k,l]$ for $k + l > 0$ must vanish (using the tree-level chiral algebra). This is because the operator $\op{tr}(B^2)$ does not have any derivatives.  Correlation functions also vanish unless there is an insertion of exactly $2$ $\til{J}$ operators. 

It turns out that this identity essentially fixes the conformal block corresponding to $\op{tr}(B^2)$.  
\begin{lemma} 
	There is a unique conformal block in our chiral algebra (at tree level) satisfying the properties listed above. 
\end{lemma}
\begin{proof} 
 The operators $J[k,l]$, $\til{J}[k,l]$ for $k + l > 0$ form an ideal in the chiral algebra, so it is consistent to define a conformal block where the correlation functions involving these operators all vanish. Further, the chiral algebra has a grading by the number of $\til{J}$'s we have.  This means it is consistent to define a conformal block by saying that the correlation functions vanish unless we have exactly two $\til{J}$ insertions. 

Correlation functions of the form
	\begin{equation} 
		\ip{\op{tr}(B^2) \mid \til{J}_{a_1}[0,0](z_1) \til{J}_{a_2}[0,0](z_2) J_{a_3}[0,0](z_3) \dots J_{a_{n}} [0,0] (z_n)   }   
	\end{equation}
	are uniquely determined by the correlation function with no $J$ insertions  by the poles  coming from the OPEs and the fact that we must have a second order zero at $z_k = \infty$, $k = 3 \dots n$.  These constraints are somewhat over-determined, but can be solved precisely because the two-point correlator is invariant under the $G$ symmetry.  
\end{proof}

Finally, we will show the following:
\begin{proposition}
	In the case the gauge group is $U(n)$, the colour-ordered correlator in our chiral algebra is
	\begin{equation} 
		\ip{\op{tr}(B^2) \mid J_{a_1} (z_1) \dots \til{J}_{a_i}(z_i) \dots \til{J}_{a_j} (z_j) \dots J_{a_n}(z_n) } = \frac{z_{ij}^4}{z_{12} z_{23} \dots z_{n1} } \op{tr}( \t_{a_1} \dots \t_{a_n} ).
	\end{equation}
	\end{proposition}
Here, we simply write $J,\til{J}$ instead of $J[0,0]$, $\til{J}[0,0]$. By the colour-ordered correlator we mean the same thing as is meant in the Parke-Taylor formula: the full correlator is a sum over terms where the colour indices have been contracted using a single trace. We are focusing on the term where the order in the trace is the same as the order in which we wrote our operators in the correlation function.  Since the full correlator does not care which order we write our operators, there is no loss of generality: the full correlator function is simply a sum over permutations of the colour-ordered correlator.  
\begin{proof} 
	We prove this by induction.  It is true for $n  = 2$.  Let us assume it is true with $n$ operators, and prove the case with $n+1$ operators. Without loss of generality we can assume that the $n+1$st operator is a $J$, not a $\til{J}$.  Consider the  correlator 
\begin{equation} 
	\ip{\op{tr}(B^2) \mid J_{a_1} (z_1) \dots \til{J}_{a_i}(z_i) \dots \til{J}_{a_j} (z_j) \dots J_{a_n}(z_n) J_{a_{n+1}} (z_{n+1})} 
\end{equation}
	There are poles in this expression when $z_{n+1} = z_i$. The residue at these poles is a correlator with $n$ insertions, where we have removed $J_{a_{n+1}}(z_{n+1})$ and replaced $J_{a_i}(z_i)$ by $f_{a_i a_{n+1}}^b J_{b}(z_i)$. 

	When we do this, the colour indices are contracted in an order in which $n+1$ is adjacent to $i$; it can be before or after $i$, and the two possibilities have opposite signs.

	This cannot contribute to the colour-ordered correlator, however, unless $i = n$ or $i=1$.  We conclude that, in the colour-ordered correlator,
	\begin{equation}
		\begin{split} 
			&\ip{\op{tr}(B^2) \mid J_{a_1} (z_1) \dots \til{J}_{a_i}(z_i) \dots \til{J}_{a_j} (z_j) \dots J_{a_n}(z_n) J_{a_{n+1}} (z_{n+1})}  \\	&= \left( \frac{1}{z_{n+1} - z_1}  - \frac{1}{z_{n+1} - z_n}  \right) \ip{\op{tr}(B^2) \mid J_{a_1} (z_1) \dots \til{J}_{a_i}(z_i) \dots \til{J}_{a_j} (z_j) \dots J_{a_n}(z_n) } \\
			&= \frac{z_{n1} }{ z_{n,n+1} z_{n+1,1} } \ip{\op{tr}(B^2) \mid J_{a_1} (z_1) \dots \til{J}_{a_i}(z_i) \dots \til{J}_{a_j} (z_j) \dots J_{a_n}(z_n). } 
		\end{split}	
	\end{equation}
Thus, the formula is proved by induction.
\end{proof}

\section{CSW rules}\label{s:CSW}
In the introduction, we stated that our method of writing the form factor integrand also allows one to understand part of the structure of the unintegrated NMHV amplitude.  Let us explain how this works in more detail.

The first step is to understand the OPE $\op{tr}(B^2)(0) \op{tr}(B^2)(x)$.  If we work at tree level, it is easy to see that the result is an operator cubic in $B$.   The coefficient of $\norm{x}^{-2}$ must be an operator of dimension $6$.  The only Lorentz invariant operator of this nature is 
\begin{equation} 
	\op{tr}(B^3) := B^a_{\alpha_1 \beta_1} B^b_{\alpha_2 \beta_2} B^c_{\alpha_3 \beta_3} f_{abc} \eps^{\beta_1 \alpha_2} \eps^{\beta_2 \alpha_3} \eps^{\beta_3 \alpha_1}  
\end{equation}
Therefore, the tree-level OPE must be of the form
\begin{equation} 
	\op{tr}(B^2)(0) \op{tr}(B^2)(x) \sim C \norm{x}^{-2} \op{tr}(B^3) + \dots  
\end{equation}
where $\dots$ indicates terms which are less singular, and $C$ is a constant.  One might worry that $C$ is zero, but a simple explicit computation with the Feynman diagram in figure \ref{fig:KMfig} tells us that it is not. 
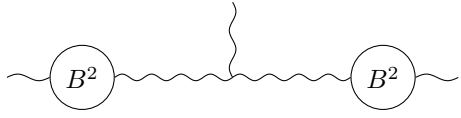
\begin{figure}
	\begin{center}
	\begin{tikzpicture}	
		\node[circle, draw] (J1) at (-2,0) {$B^2$};
		\node[circle, draw] (J2) at (2,0) {$B^2$};

			\draw[decorate, decoration=complete sines]  (-3,0) -- (J1);  
	\draw[decorate, decoration=complete sines]  (J1) -- (0,0);  
\draw[decorate, decoration=complete sines]  (J2) -- (0,0); 
\draw[decorate, decoration=complete sines]  (J2) -- (3,0);  
			
			\draw[decorate,decoration=complete sines] (0,0) -- (0,1); 
		\end{tikzpicture}
	\end{center}
	\caption{The Feynman diagram capturing the $\op{tr}(B^2)\op{tr}(B^2)$ OPE. \label{fig:KMfig}}
\end{figure}

Since $\op{tr}(B^3)$ is a Lorenz invariant operator, symmetry considerations tell us that we have
\begin{equation} 
	\ip{\op{tr}(B^3) \middle|  \til{J}^a(z_1) \til{J}^b(z_2) \til{J}^c(z_3) } = z_{12} z_{13} z_{23} f_{abc}. 
\end{equation}
We can compute the correlators with insertions of $3$ $\til{J}$'s and $n$ $J$'s by the same method we used to compute the MHV amplitudes .  We want to show that these correlators factorize as a product of MHV amplitudes, in the same way as the CSW rules.

Let $V_i(z_i)$ denote the $n$ chiral algebra insertions, states, $3$ of which are $\til{J}$ and $n-3$ are $J$. The precise formula we want to prove is that 
\begin{multline} 
	\ip{\op{tr}(B^3) \middle| V_1(z_1) \dots V_n(z_n)   } \\ = -\frac{1}{6} \sum \ip{\op{tr}(B^2) \middle| V_{i_1} (z_{i_1}) \dots V_{i_k}(z_{i_k}) \til{J}^a (z) }    \ip{\op{tr}(B^2) \middle|J_a(z)  V_{j_1} (z_{j_1}) \dots V_{j_{n-k}}(z_{j_{n-k} }) }   \label{eqn:NMHV2}  
\end{multline}
where $z$ is arbitrary.  The sum on the right hand side is over all ways of distributing the chiral algebra insertions among the correlators.

We will prove this by induction on the number of $J$ insertions.  We first check that the right hand side of equation \eqref{eqn:NMHV} has the same structure of poles and zeroes as the left hand side as we vary the position of the $J$ insertions.  That is,  the function on the right hand side should vanish to second order when a $J$ insertion goes to $z = \infty$, and have a pole determined by the OPE when a $J$ insertion hits a $J$ or $\til{J}$ insertion. This is easily seen to be the case. (To see this,  it is important to note that given two insertions $J_i(z_i)$, $V_j(z_j)$, where $V$ is either a $J$ or a $\til{J}$,  there are always  terms in the sum on the right hand of \eqref{eqn:NMHV2} where the insertions are placed in the same factor).

However, in principle, there are spurious poles when the insertion $J^{b_i}(z_i)$ coincides with the insertion of $J^a(z)$ or $\til{J}_a(z)$.  These cancel, as 

	\begin{align} 
		\Big\langle \cdots J^{b_i}(z_i) J^a(z) \Big\rangle \Big \langle  \til{J}^a(z) \cdots \Big\rangle  
		&= \Big\langle \cdots \frac{1}{z-z_i} f^{b_i a c} J^c(z) \Big\rangle \Big \langle  \til{J}^a(z) \cdots \Big\rangle  \\ 
		\Big\langle \cdots  J^a(z) \Big\rangle \Big \langle  \til{J}^a(z) J^{b_i}(z_i) \cdots \Big\rangle 
		&= \Big\langle \cdots J^a(z) \Big\rangle \Big \langle  \til{J}^c(z) \frac{1}{z-z_i} f^{b_i c a}  \cdots \Big\rangle  
	\end{align}	
(We are freely lowering and raising indices here using the Killing form).

This allows us to reduce the proof of the equality \eqref{eqn:NMHV2} to the case that there are only $3$ insertions on the left hand side, all of which are $\til{J}$.  That is, we need to show that
\begin{equation}
	\begin{split} 
	-3	\ip{\op{tr}(B^3) \middle| \til{J}^{a_1}(z_1) \til{J}^{a_2} (z_2) \til{J}^{a_3}(z_3)   } =&  \ip{\op{tr}(B^2) \middle| \til{J}^{a_1}(z_1)  \til{J}^b (z)   }    \ip{\op{tr}(B^2) \middle|J^b(z) \til{J}^{a_2}(z_2) \til{J}^{a_3}(z_3)   }  \\
		+&  \ip{\op{tr}(B^2) \middle| \til{J}^{a_2}(z_2)  \til{J}^b (z)   }    \ip{\op{tr}(B^2) \middle|J^b(z) \til{J}^{a_3}(z_3) \til{J}^{a_1}(z_1)   }  \\
		+&  \ip{\op{tr}(B^2) \middle| \til{J}^{a_3}(z_3)  \til{J}^b (z)   }    \ip{\op{tr}(B^2) \middle|J^b(z) \til{J}^{a_1}(z_1) \til{J}^{a_2}(z_2)   }. 
	\end{split}	
\end{equation}
The right hand side is
\begin{equation} 
	f^{a_1 a_2 a_3} \left( 	\frac{ (z-z_1)^2 z_{23}^3 }{(z-z_2)(z-z_3)  } +  \frac{ (z-z_2)^2 z_{31}^3 }{(z-z_3)(z-z_1)  }  +        \frac{ (z-z_3)^2 z_{12}^3 }{(z-z_1)(z-z_2)  }       \right)  
\end{equation}
whereas the left hand side is
\begin{equation} 
	-3 f^{a_1 a_2 a_3} z_{12} z_{13} z_{23}. 
\end{equation}
To complete the proof we need simply check that
\begin{equation} 
	-3 z_{12} z_{13} z_{23} (z-z_1)(z-z_2)(z-z_3) = (z-z_1)^3 z_{23}^3 + (z-z_2)^3 z_{31}^3 + (z-z_3)^3 z_{12}^3 . 	 
\end{equation}
This we obtain by cubing the  identity
\begin{equation} 
	(z-z_1) z_{23} + (z-z_2) z_{31} + (z-z_3)z_{12} = 0.  
\end{equation}

In sum, we have demonstrated that our chiral algebra correlators result in a close match to NMHV amplitudes obtained using the CSW rules. We emphasize that in our form factor integrand, where we do not integrate over operator positions, the common operator position in the factorized expression is undetermined; in the standard CSW prescription for integrated amplitudes, it is fixed. It would be fascinating to find a more systematic match between the CSW rules and our chiral algebra correlation functions, perhaps by deriving the CSW rules as Feynman rules from our twistor space theories. We plan to pursue this, as well as further connections between the two approaches at loop-level, and at the level of integrated amplitudes, in future work.

\section{One-loop amplitudes by axion exchange}
\label{s:oneloop}
In the introduction, we showed that in the conformal block associated to $(\tr \rho)^2$, the four-point all $+$ amplitude matches the known one-loop all $+$ amplitude of gauge theory. This is to be expected, as the $4d$ theory by itself does not have any scattering amplitudes, so that the known one-loop amplitude of self-dual gauge theory must be cancelled by the axion exchange.  

Here we will verify that this holds for the $n$-point all $+$ amplitude.  The formula for the one-loop colour-ordered amplitude is \cite{Bern:1993qk,Mahlon:1993fe}
\begin{equation} 
	\ip{1^+ \dots n^+}_{\text{colour-ordered}} = \frac{ H_n  }{\ip{12} \dots \ip{n1} } 
\end{equation}
where
\begin{equation} 
	 H_n = \sum_{1 \le i_1 < i_2 < i_3 < i_4 \le n} \ip{i_1 i_2} [i_2 i_3] \ip{i_3 i_4} [i_4 i_1]  
\end{equation}
Let us thus write the full amplitude, up to normalization, and including the colour factors:
\begin{equation} 
	\ip{1^+ \dots n^+} = \frac{1}{n} \sum_{\sigma \in S_n}  \frac{   \sum_{1 \le i_1 < i_2 < i_3 < i_4 \le n} \ip{\sigma_{i_1} \sigma_{i_2}} [\sigma_{i_2} \sigma_{i_3} ] \ip{\sigma_{i_3} \sigma_{i_4} } [\sigma_{i_4} \sigma_{i_1} ]       }{\ip{\sigma_1 \sigma_2} \dots \ip{\sigma_n \sigma_1} } \op{Tr}(\t_{\sigma_1} \dots \t_{\sigma_n} ) \label{eqn:oneloop_amplitude} 
\end{equation}
where we are summing over elements of the symmetric group $S_n$. The pre-factor of $\tfrac{1}{n}$ is a symmetry factor, accounting for the fact that the colour-ordered amplitude has cyclic group symmetry.

Our construction works with gauge group $SU(2)$, $SU(3)$, $SO(8)$ and the exceptional groups, with no matter. In these cases, $\op{Tr}$ denotes the trace in the adjoint representation.  If we want to consider $SU(N_c)$ gauge group, then we need to introduce matter with $N_f = N_c$. In that case, we will have fermions in the loop, which will change the colour factor. In that case, $\op{Tr}$ means
\begin{equation} 
	\op{Tr} = \op{tr}_{\op{adj}} - N_c \op{tr}_{\op{fun}} - N_c \op{tr}_{\br{\op{fun}}} 
\end{equation}
i.e.\, trace in the adjoint, minus $N_c$ times trace in the fundamental and anti-fundamental.   

To write this in terms of chiral algebra correlators, we will write a generating function for the generators of the chiral algebra
\begin{equation} 
	J(\mu_i^{\alpha}, z_i) = \sum J[r,s](z_i) \frac{1}{r! s!} (\mu_i^{\dot{1}})^r (\mu_i^{\dot{2}})^s 
\end{equation}
In this notation we will let
\begin{equation} 
	[ij] = \eps_{\dot\alpha \dot\beta} \mu_i^{\dot\alpha} \mu_j^{\dot\beta}. 
\end{equation}
Here $\mu$ is, as usual, an auxiliary spinor which together with $\lambda=(1,z)$ specifies the momentum of the external state. 

Let us write a proposal for the correlator, which we will then check satisfies the properties dictated by the OPE. Our proposed correlator is given by the same expression:
\begin{multline} 
	\ip{ (\tr \rho)^2 \mid J_{a_1}(\mu_1, z_1) \cdots J_{a_n}(\mu_n, z_n) }^{\text{proposed} } \\
	= \frac{1}{n}   \sum_{\sigma \in S_n}  \frac{   \sum_{1 \le i_1 < i_2 < i_3 < i_4 \le n} \ip{\sigma_{i_1} \sigma_{i_2}} [\sigma_{i_2} \sigma_{i_3} ] \ip{\sigma_{i_3} \sigma_{i_4} } [\sigma_{i_4} \sigma_{i_1} ]       }{\ip{\sigma_1 \sigma_2}\ip{\sigma_2 \sigma_3}  \dots \ip{\sigma_n \sigma_1} } \op{Tr}(\t_{a_{\sigma_1}} \dots \t_{a_{\sigma_n}} ) \label{eqn:proposed_correlator} 
\end{multline}
It is clear that the proposed correlator reproduces the correct scattering amplitude.
	
Let us rewrite the correlator slightly in a way so that the indices that appear in the numerator are $1,2,3,4$.  For a given permutation $\sigma$, we order the elements $1,2,3,4$ according to how they appear in the list $\sigma_1,\dots, \sigma_n$.  We let $1_{\sigma}$ be the first element in the set $\{1,2,3,4\}$ in this order, and similarly $2_{\sigma}$, $3_{\sigma}$, $4_{\sigma}$.  Then, according to the formula \eqref{eqn:proposed_correlator}  we have
\begin{equation}
	\begin{split} 
		\ip{(\tr \rho)^2 \mid J_{a_1}[1](\mu_1, z_1) J_{a_2}[1](\mu_2, z_2) J_{a_3}[1](\mu_3, z_3) J_{a_4}[1] (\mu_4, z_4) J_{a_5}(z_5) \dots J_{a_{n}}(z_n) }  	\\
		=\frac{1}{n} \sum_{\sigma \in S_n} \op{Tr}\left(t_{a_{\sigma_1}} \dots t_{a_{\sigma_n}} \right) \frac{ \ip{1_{\sigma} 2_{\sigma}} [2_{\sigma} 3_{\sigma} ] \ip{3_{\sigma} 4_{\sigma} }[4_{\sigma} 1_{\sigma} ]  } {\ip{\sigma_1 \sigma_2} \dots \ip{\sigma_n \sigma_1 }    }  
	\end{split}\label{eqn:proposed2}	
\end{equation}
where $J[1](z_i)  = J[1,0] (z_i) \mu_i^{\dot 1} + J[0,1](z_i) \mu_i^{\dot 2}$ and $J(z_i)$ means $J[0,0](z_i)$.

What we need to check is that our proposed correlator has the correct structure of poles in $z$ to be the actual chiral algebra correlator. The first thing to check is the poles and zeroes at $z = \infty$ 	

The operator $J[0,0](z_i)$ has a zero of order $2$ at $z = \infty$, and $J[1,0](z_i)$, $J[0,1](z_i)$ has a zero of order $1$.  Each index appears twice in the denominator of \eqref{eqn:proposed_correlator}.  Therefore, we have a zero of order $2$ at $z_i = \infty$ except for the four indices $i,j,k,l$  appearing in the numerator. These have a pole of order $1$, because we have $\ip{ij} \ip{kl}$ in the numerator.  For these operators, we are taking a correlator of $J[1,0]$ or $J[0,1]$, because of the appearance of $[jk]$ and $[li]$.  Thus, as required, for these indices we have a first order zero at $z = \infty$. 

Next, let us assume by induction that the expression in equation \eqref{eqn:proposed2} is the correct correlator when we have $n-1$ insertions.  To see that it is the correct correlator when we have $n$ insertions, let us consider what happens when $z_n$ approaches one of the other points $z_i$.   There are two cases: $i> 4$, or $i \le 4$. In each term in the sum in \eqref{eqn:proposed_correlator}, there is a pole when $z_n$ approaches $z_i$ only if $\ip{in}$ appears in the denominator.     This can only happen  if  $\dots t_{a_i} t_{a_n} \dots$ or $\dots t_{a_n} t_{a_i} \dots$ appear in the trace, and these two terms appear with opposite signs.

This means that the pole at $z_n = z_i$ in the correlator \eqref{eqn:proposed2} is given by the $n-1$ point correlator, where we make the replacement 
\begin{equation} 
	J_{a_i}[0,0] (z_i) J_{a_n}[0,0](z_n) \mapsto  f_{a_i a_n}^b J_b[0,0](z_i) \frac{1}{z_n - z_i}. 
\end{equation}
Similarly, if $i \le 4$, the pole at $z_n = z_i$ is the $n-1$ point correlator where we have made the replacement
\begin{equation} 
	J_{a_i}[1] (v_i^{\alpha}, z_i) J_{a_n}[0,0](z_n) \mapsto  f_{a_i a_n}^b J_b[1](v_i^{\alpha}, z_i) \frac{1}{z_n - z_i}. 
\end{equation}
In each case, the poles are determined by the OPEs in the chiral algebra.

As a function of $z_n$, the correlator has a zero at $z_n = \infty$ of order $2$ and $n-1$ first-order poles. Inductively, the residue at each pole is fixed, and this fixes the $n$-point correlator in terms of the $n-1$ point correlator.

This proves by induction that our proposed correlator \eqref{eqn:proposed2} matches the actual chiral algebra correlator, as long as we check the initial case which is $n= 4$. This we have already done in the introduction,  but let us repeat it here.

By induction, the proposed amplitude must equal the actual amplitude if they do when $n = 4$. For the $n=4$ case, we have
\begin{equation} 
	\ip{ (\tr \rho)^2 \mid J_{a_1}[1] (v^{\alpha}_1, z_1) \cdots   J_{a_4}[1](v^{\alpha}_4, z_4) }^{\text{proposed}} =
	\sum_{\sigma \in S_4} \op{Tr}( \t_{a_{\sigma_{1}}} \dots \t_{a_{\sigma_{4}} } )  \frac{\ip{\sigma_{1} \sigma_{2} } [\sigma_{2} \sigma_{3}] \ip{\sigma_{3} \sigma_{4}} [\sigma_{4} \sigma_{1}]    } {  \ip{\sigma_{1} \sigma_{2} } \ip{\sigma_{2} \sigma_{3}}  \ip{\sigma_{3} \sigma_{4}} \ip{\sigma_{4} \sigma_{1}}  } 
\end{equation}
Recalling that conservation of momentum tells us that $\frac{[12][34]}{ \ip{12} \ip{34} }$ is totally symmetric, we can rewrite the right hand side as 
\begin{equation} 
	4!  \op{Tr}( \t_{ (a_1 } \dots t_{a_4)}   ) \frac{[12][34]}{ \ip{12} \ip{34} } 
\end{equation}
which, up to a factor, matches the correlator 
\begin{equation} 
	\ip{ (\tr \rho)^2 \mid J[1](z_1,v^{\alpha}_1) \dots J[1](z_4,v^{\alpha}_4}  
\end{equation}
we have already determined in the introduction.

In the introduction, we determined this using the OPE
\begin{equation} 
	J_a[1](z_1,v^{\alpha}_1) J_b[1](z_2,v^{\alpha}_2) \sim \frac{[12]}{\ip{12}}  \op{tr}(\t_a \t_b) F[0,0] 
\end{equation}
where $F[0,0]$ is an operator built from the twistor uplift of the axion field. 

\subsection{Calculating the conformal block associated to $\tr \rho$.} \label{sec:normalizing_F} 
There is perhaps one more point we need to elaborate on, which is why the conformal block for $\tr \rho$ gives a non-zero expectation value to the operator $F[0,0]$.   

This can be seen by symmetry reasons, as we mentioned in the introduction. It is also possible to check this explicitly, and we will do that now.  The operator $F[0,0]$ is corresponds to a field configuration on twistor space for the field $\eta$.     It follows from equation \eqref{eqn:axionmodes} that 
\begin{equation} 
	F[r,s]  = 2 \pi \i \partial \left( \delta_{z = z_0} \frac{1}{r+s+2}(v_1^{r+1} v_2^s \d v_2 - v_1^{r} v_2^{s+1} \d v_1 )   \right)
\end{equation}
so that in particular,
\begin{equation} 
	F[0,0] = 2 \pi \i \partial \left( \delta_{z = z_0} \eps^{ij} v_i \d v_j   \right).  
\end{equation}
According to the analysis of \cite{Costello:2021bah}, the value of the axion field $\rho(x)$ is given by the integral of $\partial^{-1} \eta$ over the curve corresponding to $x$.  

Let us work in the complex coordinates $u_i$, $\br{u}_i$ on $\R^4 = \C^2$, using the complex structure associated to $z = 0$. 

We have
\begin{equation}
	\begin{split} 
		\frac{1}{2 \pi \i}	\ip{\rho(u,\ubar) \mid F[0,0](z_0) } &= \int_{\substack{ v_1 = u_1 + z \ubar_2 \\ v_2 = u_2 - z \ubar_1} }  \delta_{z = z_0} \eps^{ij} v_i \d v_j   \\  	
		&= -  \int_{\substack{ v_1 = u_1 + z \ubar_2 \\ v_2 = u_2 - z \ubar_1} }  \delta_{z = z_0} (u_1 + z \ubar_2) \ubar_1 \d z + (u_2 - z \ubar_1) \ubar_2 \d z   \\  
		&= - \norm{u}^2. 
	\end{split}	
\end{equation}
This means that 
\begin{equation} 
	\ip{\Lap \rho(u,\ubar) \mid F[0,0] } = -16 \pi \i,  \end{equation}
We can normalize the conformal block to make this $1$.  

This completes the proof that the all $+$ one loop amplitudes in self-dual gauge theory as computed in \cite{Bern:1993qk,Mahlon:1993fe} are the correlators of the chiral algebra with respect to the conformal block associated to $(\tr \rho)^2$, up to an overall normalization of the conformal block.

\section{Tree-level scattering amplitudes in the presence of an axion}\label{s:KacMoody}
In this section we will prove that the tree-level, all $+$ scattering amplitudes of Yang-Mills theory, in the presence of an axion with a logarithmic pole, are given by Kac-Moody correlation functions.  The novelty here is that there is a non-zero level. 

The axion field on $\R^4$ arises from the closed $(2,1)$-form field $\eta$ on twistor space.   This couples to gauge field on twistor space by
\begin{equation} 
	\frac{\lambda_{\g}}{4 (2 \pi \i)^{3/2} \sqrt{3} }\int \eta \mc{A} \partial \mc{A} 	\label{eqn:etaA} 
\end{equation}
where $\lambda_{\g}$ is a constant tuned to cancel the anomaly (see \S \ref{s:anomalies}). 

Suppose we allow $\eta$ to have a singularity so that the equations of motion are modified to
\begin{equation} 
	\dbar \eta = C \delta_{v_1 = v_2 = 0} \label{eqn:dbareta} 
\end{equation}
(using the coordinates $z,v_i$ on a patch of twistor space as before). Here $C$ is some non-zero constant which we will find controls the Kac-Moody level. Such a modification is a disorder defect in the $\eta$ field; it can be realized as an order defect where we integrate $\partial^{-1} \eta$ over the curve $v_1 = v_2 = 0$. 

Consider, as before, coupling $\partial_{v_1}^r \partial_{v_2}^s \mc{A}$ to some currents $J^a[r,s]$ living in a chiral algebra at $v_1 = v_2 = 0$.  As before, we will determine the OPEs between these currents by requiring that the coupled system is gauge invariant. 

Applying the gauge variation $\delta \mc{A} = \dbar \chi + [\chi, \mc{A}]$ to the  Lagrangian \eqref{eqn:etaA} and using \eqref{eqn:dbareta} gives us an extra term, which is
\begin{equation} 
	C	\frac{\lambda_{\g}}{4 (2 \pi \i)^{3/2} \sqrt{3} } \int \delta_{v_1 = v_2 = 0} \op{tr}(\chi \partial_z \mc{A} + \mc{A} \partial_z \chi ). 
\end{equation}
The gauge variation of 
\begin{equation} 
	\int J^a(z)J^b(z') \mc{A}_a(z) \mc{A}_b(z') 
\end{equation}
cancels this, as long as the OPE between $J^a(z)$ and $J^b(z')$ has a second-order pole
\begin{equation} 
	J^a(z) J^b(z') = \op{tr}(\t^a \t^b) \frac{1}{2 \pi \i} \frac{1}{z^2}  C	\frac{\lambda_{\g}}{4 (2 \pi \i)^{3/2} \sqrt{3} } \label{eqn:KM_level} 
\end{equation}
In other words, we find that the (deformed) Koszul dual algebra is the Kac-Moody algebra with a non-zero level.  

When we pass to real space, the field $\eta$ becomes the axion field. In \cite{Costello:2021bah}, section 5.5, the field $\rho$ corresponding to an $\eta$ with $\dbar \eta = C \delta_{v_i = 0}$ was computed.   This is
\begin{equation} 
	\rho = \frac{C}{2 \pi \i} \log \norm{x}^2, 
\end{equation}
in the normalization where $\rho$ is  the integral of $\partial^{-1} \eta$ over a $\CP^1$. In this normalization, the coupling between $\rho$ and the $4d$ gauge field $A$ is
\begin{equation} 
	 \frac{\lambda_{\g} }{8  (2 \pi \i)^{3/2} \sqrt{3} } \rho F(A)^2. 
\end{equation}
We can absorb the factor of $\frac{\lambda_{\g}}{4 (2 \pi \i)^{3/2}\sqrt{3} }$ that appears in the coupling of $\rho$ to $A$ and $\eta$ to $\mc{A}$ in a rescaling of both $\rho$ and $\eta$.    If we do this, an axion profile which couples to the gauge field by 
\begin{equation} 
	\frac{C}{2 \pi \i} \log \norm{x}^2 \tfrac{1}{2}  F(A)^2 
\end{equation}
gives rise to a Kac-Moody level of $\frac{C}{2 \pi \i}$.  

This tells us that, as desired, the all $-$ scattering amplitudes of tree-level gauge theory in the presence of an axion coupled by
\begin{equation} 
	(\log \norm{x}^k)  F(A)^2 
\end{equation}
are the correlators of the Kac-Moody algebra at level $k$.

\section{Discussion \& Conclusions}\label{s:conclusions}

In this note, we have explained the twistorial origin of various aspects of the celestial holography program, including: how chiral algebras may be constructed from local holomorphic theories on twistor space, why they coincide with (more precisely, enlarge) celestial chiral algebras for certain theories that descend to self-dual limits of Yang-Mills and Einstein gravity in four dimensions, and how their generators correspond to negative-dimension conformal primary states in 4d. We further explored aspects of these correspondences using inspiration and techniques from the twisted holography program, emphasizing the role of Koszul duality in obtaining the chiral algebras \textit{and their deformations at loop-level} from tractable computations. Finally, we illustrated how correlation functions in the chiral algebra can be used to reproduce certain scattering amplitudes in Yang-Mills theory. 

In addition to the future directions mentioned in the main text, we are pursuing various open questions suggested by this study; we preview some of them below.

\begin{itemize}
    \item In work in progress with A. Sharma, we study additional gauge theory and gravity amplitudes, at tree and loop-level, from correlation functions in our Koszul dual chiral algebra. It will be fascinating to better understand which amplitudes are accessible from our methods, and to which loop-order they are effectively computable. \\
    \item The twistorial perspective on celestial symmetries was also recently emphasized in \cite{Adamo:2021lrv}, which explored celestial holography from a worldsheet, ambitwistor string construction. It would be interesting to connect this approach with ours more directly. \\
    \item Our approach to celestial holography has been inspired by the twisted holography program \cite{Costello:2018zrm, Costello:2020jbh}, which focuses on computable properties of holomorphic/partially topological theories that arise from twists of supersymmetric string constructions. Our proposal is that (at least, certain aspects of) celestial holography should be viewed as an instance of twisted holography on twistor space. From this point of view, it would be desirable to have a more concrete string theory embedding, and explore if this enables one to access, e.g., aspects of massive states in the conformal basis. As discussed in \cite{Costello:2021bah}, a natural candidate for an anomaly-free example of a holomorphic twistor space theory is a type I topological string \cite{Costello:2019jsy}, which is the result of twisting the type IIB string in the presence of an $O7^-$-plane and $D7$-branes, placed in an Omega-background. It may be interesting to study this example along the lines of \cite{Costello:2018zrm, Costello:2020jbh}, incorporating backreaction from open string sectors in a large-N limit.\\
    \item We have explored the celestial chiral algebras from their realization as boundary algebras of 3d holomorphic-topological theories. Such boundary conditions on 3d $\mathcal{N}=2$ theories support nonperturbative boundary monopole operators \cite{Dimofte:2017tpi}, which persist after twisting \cite{Costello:2020ndc, Zeng:2021zef}. What is the interpretation of these operators in 4d? It would be interesting to explore this question by pushing/pulling along the double fibration of twistor space discussed in \S \ref{s:blocks}. \\
    \item The chiral algebras on the boundary of these 3d theories also enjoy higher products arising from their interaction with bulk operators, as discussed in \cite{Costello:2020ndc}, and it would be interesting to interpret these operations from the celestial point of view. It would also be interesting if the 3d picture sheds any light (pun intended) on the role of shadow/light transforms which are often employed in the study of celestial symmetry algebras.\\
    \item We explored how states of negative conformal dimension are in correspondence with generators of the chiral algebra. We have also studied states of positive conformal dimension, which correspond to algebra modules. Physically, these are defects supported on zeros of a polynomial, as described in \S \ref{s:states}, and are sourced by certain Wilson lines. It would be interesting to determine extensions of the chiral algebra by (some subset of) such modules. Could such an extension provide a natural description of the maximal asymptotic symmetry algebra of the corresponding 4d theory? \footnote{For a recent exploration of the asymptotic symmetry algebra from a different perspective see, e.g., \cite{Donnay:2020guq}.}
\end{itemize}

We hope that these, and many other, questions will provide a fruitful bridge between twisted and celestial holography.

\section{Acknowledgements} We thank F. Cachazo, D.  Gaiotto, L. Mason, S. Shao, A. Sharma, A. Strominger, and K. Zeng for helpful discussions and correspondences. N.P. acknowledges support from the University of Washington and the DOE award DE-SC0022347. K.C. is supported by the NSERC Discovery Grant program and by the Perimeter Institute for Theoretical Physics. Research at Perimeter Institute is supported by the Government of Canada through Industry Canada and by the Province of Ontario through the Ministry of Research and Innovation.

\printbibliography

\end{document}